\documentclass[10pt]{IEEEtran} 

\usepackage{amsmath, amssymb}

\usepackage{epsfig}
\usepackage{cite}
\usepackage{color}
\usepackage[normalsize]{subfigure}

\usepackage{psfrag}

\usepackage{enumerate}

\usepackage{amsthm}

\newtheorem{theorem}{Theorem}
\newtheorem{lemma}[theorem]{Lemma}
\newtheorem{cor}[theorem]{Corollary}

\newtheorem{prop}[theorem]{Proposition}

\theoremstyle{definition}
\newtheorem{definition}{Definition}

\theoremstyle{remark}
\newtheorem{remark}{Remark}

\newcounter{myenum} 
  {\end{list}}

\newcommand{\snr}{\textnormal{\small  \fontfamily{phv}\selectfont snr}}
\newcommand{\mmse}{\textnormal{\small  \fontfamily{phv}\selectfont mmse}}

\newcommand{\snrs}{\textnormal{\footnotesize \fontfamily{phv}\selectfont snr}}

\newcommand{\mse}{\textnormal{\small\fontfamily{phv}\selectfont mse}}

\newcommand{\bE}{\mathbb{E}} 
\newcommand{\bP}{\mathbf{P}} 
\newcommand{\bY}{\mathbf{Y}} 
\newcommand{\bX}{\mathbf{X}}
\newcommand{\bW}{\mathbf{W}} 
\newcommand{\bS}{\mathbf{S}} 
\newcommand{\bM}{\mathbf{M}} 
\newcommand{\by}{\mathbf{y}} 
\newcommand{\bx}{\mathbf{x}}
\newcommand{\bw}{\mathbf{w}} 
\newcommand{\bs}{S} 

\newcommand{\bU}{\mathbf{U}}
\newcommand{\bu}{\mathbf{u}}
\newcommand{\bv}{\mathbf{v}} 
\newcommand{\bV}{\mathbf{V}} 
\newcommand{\bz}{\mathbf{z}} 

\newcommand{\bA}{\mathbf{A}}

\newcommand{\one}{\boldsymbol{1}} 
\newcommand{\rw}{\rightarrow}

\newcommand{\yes}{\textbf{yes}}

\title{The Sampling Rate-Distortion Tradeoff for Sparsity Pattern Recovery in Compressed Sensing}%: 
\author{Galen Reeves, \IEEEmembership{Member, IEEE} and Michael Gastpar \IEEEmembership{Member, IEEE}
\thanks{This work was supported in part by ARO MURI No. W911NF-06-1-0076.}
\thanks{G. Reeves was with the Department of Electrical Engineering and Computer
Sciences, University of California, Berkeley, CA 94720 USA. He is now with the Department of Statistics, Stanford University, Stanford, CA 94305-4065 USA. (e-mail: greeves@stanford.edu)}
\thanks{M. Gastpar is with the School of Computer and Communication Sciences,
Ecole Polytechnique F\'ed\'erale (EPFL), 1015 Lausanne, Switzerland,
and with the Department of Electrical Engineering and Computer
Sciences, University of California, Berkeley, CA 94720 USA (e-mail: michael.gastpar@epfl.ch).}
}
\begin{document}
\maketitle

\begin{abstract}
Recovery of the sparsity pattern (or support) of an unknown sparse vector from a limited number of noisy linear measurements is an important problem in compressed sensing. In the high-dimensional setting, it is known that recovery with a vanishing fraction of errors is impossible if the measurement rate and the per-sample signal-to-noise ratio (SNR) are finite constants, independent of the vector length. In this paper, it is shown that recovery with an arbitrarily small but constant fraction of errors is, however, possible, and that in some cases computationally simple estimators are near-optimal. Bounds on the measurement rate needed to attain a desired fraction of errors are given in terms of the SNR and various key parameters of the unknown vector for several different recovery algorithms. The tightness of the bounds, in a scaling sense, as a function of the SNR and the fraction of errors, is established by comparison with existing information-theoretic necessary bounds. Near optimality is shown for a wide variety of practically motivated signal models. %These bounds rigorously validate predictions made by the powerful but heuristic replica method from statistical physics..
\end{abstract}


%%%%%%%%%%%%%%%%%%%%%
%\input{sections/introduction_upper}
\section{Introduction} \label{sec:intro} 

\IEEEPARstart{S}{uppose} that a vector $\bx$ of length $n$ is known to have a small number $k$ of nonzero entries, but the values and locations of the nonzero entries are unknown and must be estimated from a set of $m$ noisy linear projections (or samples) given by the vector
\begin{align}\label{eq:model}
\by = A \bx + \mathbf{\bw}
\end{align}
where $A$ is a known $m \times n$ measurement matrix and $\mathbf{\bw}$ is additive white Gaussian noise. The problem of {\em sparsity pattern recovery} is to determine which entries in $\bx$ are nonzero. This problem, which is known variously throughout the literature as support recovery or model selection, has applications in compressed sensing \cite{Donoho_IT06,CandesRombergTao06A, CandesTao06}, sparse approximation \cite{DevoLore93}, signal denoising \cite{Chen_JSC98}, subset selection in regression \cite{Miller90}, and structure estimation in graphical models \cite{MeinBuhl06}.

A great deal of previous work \cite{MeinBuhl06, ZhaoYu_JMLR06, Wainwright_SharpThresholds_IEEE09, Wainwright_InfoLimits_IEEE09,FLetcher_IEEE09, Wang_IEEE10, akcakaya_IEEE10, AerSalZha10, Reeves_masters, RG08}, has focused on necessary and sufficient conditions for exact recovery of the sparsity pattern. By contrast, this paper studies the tradeoff between the number of samples $m$ and the number of detection errors. We focus on the high-dimensional setting where the sparsity rate (i.e.~the fraction of nonzero entries) and the per-sample signal-to-noise ratio (SNR) are finite constants, independent of the vector length $n$. Our results are bounds on the sampling rate $\rho = m/n$ needed to attain a desired detection error rate $D$ for several different recovery algorithms. These bounds are given explicitly in terms of the sparsity rate, the SNR, and various key properties of the unknown vector. Complementary information theoretic lower bounds are given in the companion paper \cite{RG-Lower-Bounds}. An illustration of the bounds is given in Fig.~\ref{fig:example1}.

\begin{figure*}[htbp]
\centering
\psfrag{SNR (dB)}[c]{\small SNR (dB)}
\psfrag{sparsity rate}[lb]{\small sparsity rate $\kappa$} 
\psfrag{Sampling Rate}[tb]{\small Sampling Rate $\rho$} %(no. samples / signal length)}
\psfrag{t1}[B]{}%\footnotesize Distortion $\alpha = 0.1$}
\psfrag{MC Upper Bound [?]}[l]{\small  \begin{minipage}{1.4in} \centering Thresholding Upper Bound\\  (Theorem 2) \end{minipage}}
\psfrag{ML Upper Bound [?]}[l]{\small \begin{minipage}{1.45in} \centering  Comb. Opt. Upper Bound\\ (Theorem 1) \end{minipage}}
%\psfrag{ML Upper Bound [?]}[l]{\small \begin{minipage}{1.45in} \centering  Combinatorial Optimization\\ Upper Bound (Theorem 1) \end{minipage}}
\psfrag{Lower Bound}[lb]{\small Lower Bound \cite{RG-Lower-Bounds}}
\psfrag{(This Paper)}[lT]{}
\epsfig{file=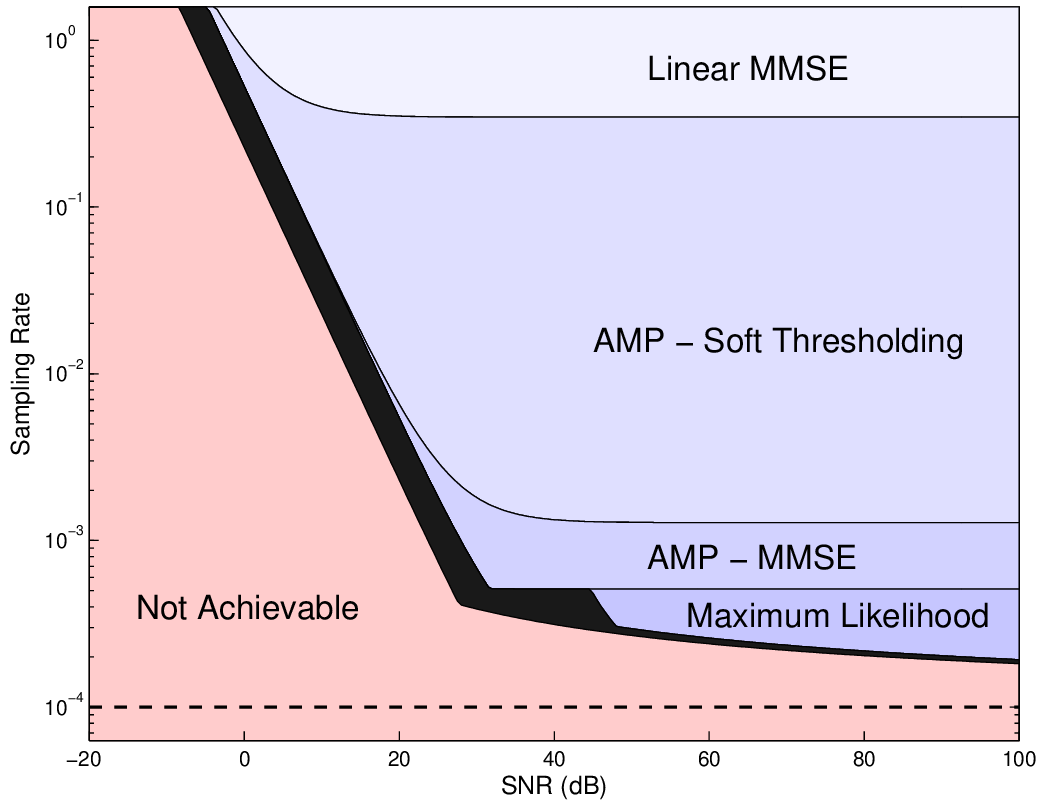, width = .49\textwidth}
\epsfig{file=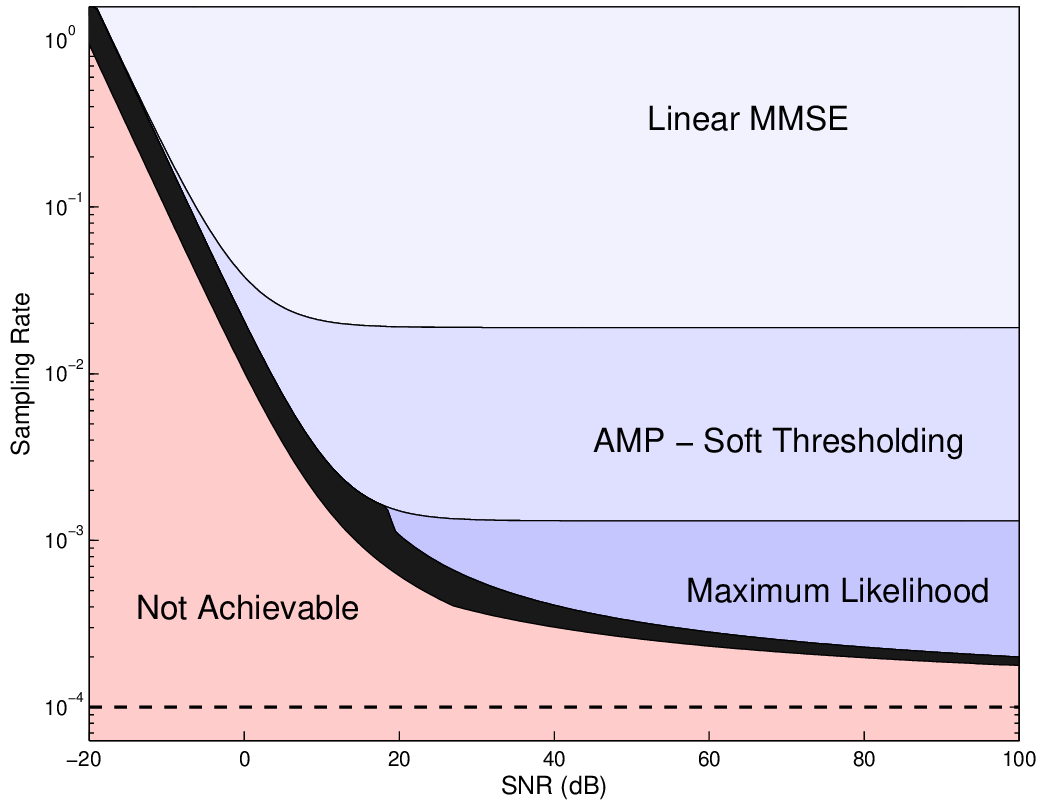, width =.49\textwidth}
\caption{Bounds on the achievable sampling rate $\rho$ as a function of the SNR for various recovery algorithms when the desired sparsity pattern detection error rate is $D= 0.05$ ($95 \%$ accuracy), the sparsity rate is $\kappa = 10^{-4}$, and the measurement matrices have i.i.d.~Gaussian entries. In the left panel, the nonzero entries are i.i.d.~zero-mean Gaussian. In the right panel, the nonzero entries are lower bounded in squared magnitude by $20\%$ of their average power but are otherwise arbitrary. The achievable bounds are given in this paper. The necessary bound is given in \cite{RG-Lower-Bounds}  }
\label{fig:example1}
\end{figure*}

\subsection{Overview of Main Contributions}

We study the high-dimensional setting where the measurement matrix $A$ is generated randomly and independently of the vector $\bx$ and the measurements are corrupted by additive white Gaussian noise. Three main contributions of the paper are the following:
\begin{enumerate}
\item {\em Fundamental Limits:} We derive an upper bound on the sampling rate needed using {\em maximum likelihood} (ML) estimation. While previous work has focused on exact recovery \cite{MeinBuhl06, ZhaoYu_JMLR06, Wainwright_SharpThresholds_IEEE09, Wainwright_InfoLimits_IEEE09,FLetcher_IEEE09, Wang_IEEE10} or the scaling behavior for approximate recovery \cite{akcakaya_IEEE10}, our work gives an explicit bound on the tradeoff between the sampling rate and the fraction of detection errors. In conjunction with the information-theoretic lower bounds in \cite{RG-Lower-Bounds}, this bound provides a sharp characterization between what can and cannot be recovered in the presence of noise. This characterization is rigorous and thus validates recent predictions made using the powerful but heuristic replica method from statistical physics \cite{Tanaka02,  Muller03, GuoVer05,KabWadTan09, RanFleGoy09, GuoBarSha09}. 

\item {\em Computationally Efficient Algorithms:} In addition to our analysis of the fundamental limits, we also derive matching upper and lower bounds on the sampling rate corresponding to three computationally efficient estimators: the {\em matched filter} (MF), the {\em linear minimum mean-squared error} (LMMSE) estimator, and an iterative recovery algorithm known as {\em approximate message passing} (AMP) \cite{DonMalMon09,DonMalMon10,BayMon10a, Montanari11}. By comparison with our fundamental bounds, we show that these estimators are near-optimal in some parameter regimes, but highly suboptimal in others. 
\item {\em Universality:} It is shown that a fixed recovery algorithm can be universally near optimal over a large class of practically motivated signal models.
\end{enumerate}

Beyond these results, our framework also permits us to prove some further insights. For instance, we show that the low-distortion behavior depends primarily on the relative size of the smallest nonzero entries whereas the high SNR behavior depends primarily on the computational power of the recovery algorithm and the complexity of the underlying signal class, and we precisely characterize this dependence. Also, we show that the sampling rate-distortion function is a convex function and that, in certain settings, i.i.d. measurement matrices are asymptotically strictly suboptimal.

%%%%%%%%%%%%%%%%%%%
\subsection{Relation to Previous Work}

A great deal of previous work has focused on the approximation of sparse vectors with respect to mean squared error (MSE) \cite{Mallat93, Tibsh96, Chen_JSC98, Fuchs_IT05,  CandesTao_IT05,  Donoho_IT06,  CandesRombergTao06A, CandesTao06, DonohoEladTemlyakov_IEEtrans06, CandesRombergTao06, Tropp_IT06, DonoForMost06, Haupt06, Massart07}. Two particularly relevant results from this literature are \cite{CandesRombergTao06} and \cite{DonohoEladTemlyakov_IEEtrans06} which show that the vector $\bx$ can be approximated with MSE inversely proportional to the SNR using $m = O( k \log(n/k))$ samples and a quadratic program known as {Basis Pursuit} \cite{Chen_JSC98}. With a few additional assumptions on the magnitude of the smallest nonzero entries in $\bx$, these bounds on the MSE can be translated into bounds on the detection error rate. However, the resulting bounds correspond to adversarial noise and are thus loose in general (see \cite{Reeves_thesis}).

Another line of previous work has focused directly on the problem of exact sparsity pattern recovery \cite{MeinBuhl06, ZhaoYu_JMLR06, Wainwright_SharpThresholds_IEEE09, Wainwright_InfoLimits_IEEE09,FLetcher_IEEE09, Wang_IEEE10}. It is now well understood that $m = \Theta( k \log n)$ samples are both necessary and sufficient for exact recovery when the SNR is finite and there exists a fixed lower bound on the magnitude of the smallest nonzero elements \cite{Wainwright_InfoLimits_IEEE09,FLetcher_IEEE09, Wang_IEEE10}. In contrast to the scaling required for bounded MSE, this scaling says that the ratio $m/k$ must grow without bound as the vector length $n$ becomes large. As a consequence, exact recovery is impossible in the setting considered in this paper, when the sparsity rate, sampling rate, and SNR are finite constants, independent of the vector length $n$. 

The fundamental limits of sparsity pattern recovery with a nonzero detection error rate have also been investigated. For the special case where the values of the nonzero entries are identical and known (throughout the system), Aeron et al. \cite[Theorem V-2]{AerSalZha10} showed that $m = C \cdot k \log(n/k)$ samples are necessary and sufficient for an ML decoder where the constant $C$ is bounded explicitly in terms of the SNR and the desired detection error rate. In the general setting where the nonzero values are unknown, Akcakaya and Tarokh \cite{akcakaya_IEEE10} showed that $m = C \cdot k \log(n/k)$ samples are necessary and sufficient for a joint typicality recovery algorithm where the constant $C$ is finite, but otherwise unspecified. (In \cite{Reeves_thesis}, it is shown that this same result is implied directly by the previous work of Cand\`{e}s et al. \cite{ CandesRombergTao06}.) An important difference between these previous results and the current paper is that we give an explicit and relatively tight characterization of the constant $C$ for a broad class of signal models.

Our analysis of linear estimation is related to work by Verd\'{u} and Shamai \cite{VerSha99} and Tse and Hanly \cite{TseHan99} on linear multiuser detectors. Our analysis of AMP relies heavily on recent results by Donoho et al. \cite{DonMalMon09,DonMalMon10} and Bayati and Montanari \cite{BayMon10a} which characterize the limiting distribution of the AMP estimate under the assumption of i.i.d.~Gaussian matrices. For an overview of related work and a generalization of the algorithm, see \cite{Rangan10b}. We note that similar results for message passing algorithms have also been shown under the assumption of sparse measurement matrices with locally tree-like properties \cite{GuoBarSha09,BarSarBar10, Rangan10a}.

Finally, the bounds in the paper are compared to predictions made by the replica method from statistical physics. This is a powerful but nonrigorous heuristic that has been used previously in the context of multi-user detection \cite{Tanaka02, Muller03, GuoVer05,KabWadTan09}, and more recently in compressed sensing \cite{RanFleGoy09, GuoBarSha09}.

%\subsection{Outline of paper}
%%This paper is organized as follows: 
%Section \ref{sec:formulation} provides a precise problem formulation. Section~\ref{sec:main_results} provides the main results. Section \ref{sec:analysis} analyzes the scaling behavior of these bounds with respect to various key properties. Section \ref{sec:examples} presents specific examples and illustrations, and proofs are given in the Appendices. 

\subsection{Notation}
When possible, we use the following conventions: a random variable $X$ is denoted using uppercase and its realization $x$ is denoted using lowercase; a random vector $\bV$ is denoted using boldface uppercase and its realization $\bv$ is denoted using boldface lowercase; and random a matrix $\bM$ is denoted using boldface uppercase and its realization $M$ is denoted using uppercase. We use $[n]$ to denote the set $\{1,2,\cdots,n\}$. For a collection of vectors $\bv_1,\bv_2,\cdots,\bv_L \in \mathbb{R}^n$, the empirical joint distribution of the entries in $\{\bv_i\}_{i \in [L]}$ is the probability measure on $\mathbb{R}^L$ that puts point mass $1/n$ at each of the $n$ points $(v_{1,i},v_{2,i},\cdots,v_{L,i})$. All logarithms are taken with respect to the natural base. Unspecified constants are denoted by $C$ and are assumed to be positive and finite.

%%%%%%%%%%%%%%%%%
%\input{sections/formulation_upper}
\section{Problem Formulation}\label{sec:formulation}

Let $\bx \in \mathbb{R}^n$ be a fixed but unknown vector and consider the noisy linear observation model given by
\begin{align}
\bY = \bA \bx + \frac{1}{\sqrt{\snr}} \bW \label{eq:estimation_problem}
\end{align}
where $\bA$ is a random $m \times n$ matrix, $\snr \in (0,\infty)$ is a fixed scalar, and $\bW \sim \mathcal{N}(0, I_{m \times m})$ is additive white Gaussian noise. Note that if $\bE[ \|\bA \bx\|^2] = m$, then $\snr$ corresponds to the per-sample signal-to-noise ratio of the problem. 

The problem studied in this paper is recovery of the sparsity pattern $S^*$ of $\bx$ which is given by 
\begin{align}
S^* = \{ i \in [n] : x_i \ne 0\}.  
\end{align} 
We assume throughout that a recovery algorithm is given the vector $\bY$, the matrix $\bA$, and a parameter $\kappa$ corresponding to the fraction of nonzero entries in $\bx$. The algorithm then returns an estimate $\hat{S}$ of size $\lceil \kappa\, n \rceil$. In some cases, additional prior information about the nonzero entries of $\bx$ is also available. We use ALG to denote a generic recovery algorithm.

%%%%%%%%%%%%%%%%%%%%%%%%%%%%%%%
\subsection{Distortion Measure}\label{sec:formulation_distortion}

To assess the quality of an estimate $\hat{S}$ it is important to note that there are two types of errors. A {\em missed detection} occurs when an element in $S^*$ is omitted from the estimate $\hat{S}$. The missed detection rate is given by
\begin{align}
\text{MDR}(S^*,\hat{S}) =\frac{1}{|S^*|} \sum_{i=1}^n \one( i \in S^*, i \notin \hat{S}).
\end{align}
Conversely, a {\em false alarm} occurs when an element not present in $S^*$ is included in $\hat{S}$. The false alarm rate is given by
\begin{align}
\text{FAR}(S^*,\hat{S}) =\frac{1}{|\hat{S}|} \sum_{i=1}^n \one( i \notin S^*, i \in \hat{S}).
\end{align}

In general, various tradeoffs between the two errors types can be considered. In this paper, however, we focus exclusively the distortion measure $d: S^* \times \hat{S} \mapsto [0,1]$ given by
\begin{align}
d(S^*,\hat{S}) &= \max\big( \text{MDR}(S^*,\hat{S}) ,\,\text{FAR}(S^*,\hat{S})\big). \label{eq:distortion}
\end{align}
This distortion measure is a metric on subsets of $[n]$.

For any distortion $D \in [0,1]$ and recovery algorithm ALG we define the error probability
\begin{align}
\varepsilon_n^{(\text{ALG})}(D) = \Pr[ d(S^*,\hat{S}) > D]
\end{align}
where the probability is taken with respect to the distribution on the matrix $\bA$, the noise $\bW$ and any additional randomness used by the recovery algorithm.

%%%%%%%%%%%%%%%%%%%%%%%%%%%%%%%%%%%%%
\subsection{Signal and Measurement Models}\label{sec:model_assumptions}
In this paper, we analyze a sequence of recovery problems $\{ \bx(n), \bA(n), \bW(n)\}_{n \ge 1}$ indexed by the vector length $n$. 

\smallskip
{\em \noindent Signal Assumptions:}
We consider a subset of the following assumptions on the sequence of vectors $\bx(n) \in\mathbb{R}^n$. 
\begin{enumerate}

\item[S1] {\em Linear Sparsity:} The number of nonzero values $k(n)$ in each vector $\bx(n)$ obeys
\begin{align}
\lim_{n \rw \infty} k(n)/n= \kappa
\end{align}
for some {\em sparsity rate} $\kappa \in (0,1/2)$.  

\item[S2] {\em Convergence in Distribution:}  The empirical distribution of the entries in $\bx(n)$ converges weakly to the distribution $p_X$ of a real-valued random variable $X$ with $\bE[X^2] = 1$ and $\Pr[X \ne 0] = \kappa$, i.e.
\begin{align}
\lim_{n \rw \infty} \frac{1}{n} \sum_{i=1}^n \one(x_i(n) \le x) = \Pr[X \le x]
\end{align}
for all $x$ such that $p_X(\{x\})=0$.

\item[S3] {\em Average Power Constraint:} The empirical second moments of the entries in $\bx(n)$ converge to one, i.e.
\begin{align}
\lim_{n \rw \infty} \|\bx(n)\|^2/n = 1.
\end{align}
\end{enumerate}

Assumption S1 says that all but a fraction $\kappa$ of the entries are equal to zero, Assumption S2 characterizes the limiting distribution of the nonzero entries, and Assumption S3 prohibits the existence of a vanishing fraction of arbitrarily large nonzero values.  

\smallskip
{\em \noindent Measurement Assumptions:} We consider a subset of the following assumptions on the sequence of measurement matrices  $\bA(n) \in \mathbb{R}^{m(n) \times n}$. 
\begin{enumerate}
\item[M1] {\em Non-Adaptive Measurements:} The distribution on $\bA(n)$ is independent of the vector $\bx(n)$ and the noise $\bW(n)$. 
\item[M2] {\em Finite Sampling Rate:} The number of rows $m(n)$ obeys
\begin{align}
\lim_{n \rw \infty} m(n)/n = \rho
\end{align}
for some {\em sampling rate} $\rho \in (0,\infty)$.
\item[M3] {\em Row Normalization:} The distribution on $\bA(n)$ is normalized such that each of the $m(n)$ rows has unit magnitude on average, i.e. %by the number of rows $m(n)$ such that
\begin{align}
\bE\big[ \| \bA(n)\|^2_F\big] = m(n)
\end{align} 
where $\|\cdot\|_F$ denotes the Frobenius norm. 
\item[M4] {\em IID Entries:} The entries of $\bA(n)$ are i.i.d.~with mean zero and variance $1/n$.
\item[M5] {\em Gaussian Entries:} The entries of $\bA(n)$ are i.i.d.~Gaussian $\mathcal{N}(0,1/n)$.
\end{enumerate}

Assumptions M1-M3 are used throughout the paper. A sampling rate $\rho < 1$ corresponds to the {\em compressed} sensing setting where the number of equations $m$ is less than the number of unknown signal values $n$. A sampling rate $\rho =1$ corresponds to the number of linearly independent measurements that are needed to recover an arbitrary vector $\bx$ in the absence of any measurement noise. Assumptions M4-M5 correspond to specific distributions on $\bA(n)$ that are used for many of the results of this paper.

%%%%%%%%%%%%%%%%%%%%%%%%
\subsection{Sampling Rate-Distortion Region}

Under Assumptions S1-S3 and M1-M3, the asymptotic recovery problem is characterized by the sampling rate $\rho$, limiting distribution $p_X$, and $\snr$.

\begin{definition}\label{def:sampling_rate_distortion_function}
A distortion $D$ is {\em achievable} for a fixed tuple $(\rho,p_X, \snr)$ and recovery algorithm ALG, if there exists a sequence of measurement matrices satisfying Assumptions M1-M3 such that
\begin{align}
\lim_{n \rw \infty}\varepsilon^{(\text{ALG})}_n(D) = 0 \label{eq:achievability}
\end{align}
for any sequence of vectors satisfying Assumptions S1-S3.
\end{definition}

More generally, we may also consider problems characterized by a class of limiting distributions with the same sparsity rate $\kappa$. Let $\mathcal{P}(\kappa)$ denote the class of all probability measures obeying the conditions of Assumption S2, i.e.
\begin{align}
\mathcal{P}(\kappa) =\big\{ \text{$p_X$ : $p_X(\{0\}) = 1\!-\!\kappa$, $\textstyle \int x^2 p_X(dx) = 1$}\big\}, \label{eq:P_kappa}
\end{align}
and let $\mathcal{P}_X$ be a subset of $\mathcal{P}(\kappa)$.

\begin{definition}\label{def:achievable_gen} A distortion $D$ is {\em achievable} for a fixed tuple $(\rho,\mathcal{P}_X, \snr)$ and recovery algorithm ALG, if there exists a sequence of measurement matrices satisfying Assumptions M1-M3 such that
\begin{align}
\lim_{n \rw \infty}\varepsilon^{(\text{ALG})}_n(D) = 0 
\end{align}
for any sequence of vectors satisfying Assumptions S1-S3 for some distribution $p_X \in \mathcal{P}_X$. 
\end{definition}

We emphasize that the recovery algorithm in Definition~\ref{def:achievable_gen} is fixed and thus cannot be a function of the limiting distribution realized by an individual sequence of problems. It may however be optimized as a function of the class $\mathcal{P}_X$, thus attaining the minimax risk of the recovery problem.  

\begin{definition}
For a fixed tuple $(D,\mathcal{P}_X, \snr)$ and recovery algorithm ALG the sampling rate-distortion function $\rho^{(\text{ALG})}(D,\mathcal{P}_X, \snr)$ is given by
\begin{align}
\rho^{(\text{ALG})}(D,\mathcal{P}_X, \snr) = \inf\{ \rho \ge 0 \, : \, \text{$D$ is achievable}\}.
\end{align}
\end{definition}

To lighten the notation, we will denote the sampling rate-distortion function using $\rho^{(\text{ALG})}$ where the dependence on the tuple  $(D,\mathcal{P}_X,\snr)$ is implicit.

%%%%%%%%%%%%%%%%
%\input{sections/bounds_upper}

\begin{table*}[t]
\caption{\label{table:comparison} Overview of the Sampling Rate Distortion Bounds}
\begin{center}
\begin{tabular}{|l|l|c|c|c|c|c|}
\hline
\multicolumn{3}{|c|}{ \small Recovery Algorithm}  & \multicolumn{4}{|c|}{ \small  Bounds} \\%[1ex]
\hline
\hline
Vector Estimator &Parameters & Comp. Efficient &Result & Matrix Assump.  & Unproven Assump. & Tight   \\
\hline
ML & $\kappa$ & no & Theorem~\ref{thm:ML} &Gaussian & none & no  \\
MF & $\kappa$& \yes&Theorem~\ref{thm:MF} &\textbf{i.i.d}. & none   & \yes \\
LMMSE & $\kappa$, $\snr$& \yes&Theorem~\ref{thm:LMMSE} &Gaussian & none & \yes  \\
AMP-MMSE &  $\kappa$, $\snr$, $p_X$  & \yes& Theorem~\ref{thm:AMP_MMSE} &Gaussian & none & \yes \\
AMP-ST & $\kappa$, $\snr$, $\alpha$  & \yes& Theorem~\ref{thm:AMP_ST} &Gaussian&  none & \yes  \\
MMSE & $\kappa$, $\snr$, $p_X$ &  no &Theorem~\ref{thm:MMSE} &\textbf{i.i.d}. & Replica Symmetry & \yes   \\
\hline
\end{tabular}
\end{center}
\label{default}
\end{table*}%

\section{Sampling Rate-Distortion Bounds}\label{sec:main_results}

This section states the main results of this paper which are bounds on the sampling rate-distortion function $\rho^{(\text{ALG})}$ for several different recovery algorithms. Each of the algorithms follows the same basic approach which is illustrated in Fig.~\ref{fig:two_stage} and consists of the following two stages:
\begin{itemize}
\item {\em Vector Estimation:} The first stage of recovery produces a random estimate $\hat{\bX}$ of the unknown vector $\bx$ based on the tuple $(\bY,\bA,\kappa)$. 

\item {\em Componentwise Thresholding:} The second stage of recovery generates an estimate $\hat{S}$ of the unknown sparsity pattern $S^*$ by thresholding the estimate $\hat{\bX}$ generated in the first stage: $$\hat{S} = \big\{ i \in [n] : |\hat{X}_i| \ge T \big\}.$$
\end{itemize}

\begin{figure}
\psfrag{S}[B]{ x}
\psfrag{X}[B]{\small $\rho_\text{total}$} 
 \graphicspath{{figures/}}
  \def\svgwidth{\columnwidth}
%%%%%%%%%%%%%%%%%%%%%    
%\input{inkscape-figs/two_stage.eps_tex}

%% Creator: Inkscape inkscape 0.48.0, www.inkscape.org
%% Accompanies image file 'two_stage.eps' (pdf, eps, ps)

\begingroup
  \makeatletter
  \providecommand\color[2][]{%
    \errmessage{(Inkscape) Color is used for the text in Inkscape, but the package 'color.sty' is not loaded}
    \renewcommand\color[2][]{}%
  }
  \providecommand\transparent[1]{%
    \errmessage{(Inkscape) Transparency is used (non-zero) for the text in Inkscape, but the package 'transparent.sty' is not loaded}
    \renewcommand\transparent[1]{}%
  }
  \providecommand\rotatebox[2]{#2}
  \ifx\svgwidth\undefined
    \setlength{\unitlength}{188.8pt}
  \else
    \setlength{\unitlength}{\svgwidth}
  \fi
  \global\let\svgwidth\undefined
  \makeatother
  \begin{picture}(1,0.27966102)%
    \put(0,0){\includegraphics[width=\unitlength]{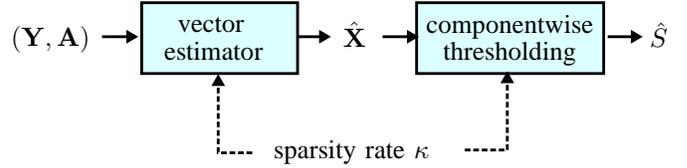}}%
    \put(0.00208918,0.19282522){\color[rgb]{0,0,0}\makebox(0,0)[lb]{\smash{$(\bY,\bA)$}}}%
    \put(0.96092521,0.19070725){\color[rgb]{0,0,0}\makebox(0,0)[lb]{\smash{$\hat{S}$}}}%
    \put(0.25255363,0.21391766){\color[rgb]{0,0,0}\makebox(0,0)[lb]{\smash{vector}}}%
    \put(0.50000123,0.19176626){\color[rgb]{0,0,0}\makebox(0,0)[lb]{\smash{$\hat{\bX}$}}}%
    \put(0.23180187,0.17155192){\color[rgb]{0,0,0}\makebox(0,0)[lb]{\smash{estimator}}}%
    \put(0.62508025,0.21163273){\color[rgb]{0,0,0}\makebox(0,0)[lb]{\smash{componentwise}}}%
    \put(0.65085473,0.171465){\color[rgb]{0,0,0}\makebox(0,0)[lb]{\smash{thresholding}}}%
    \put(0.39615674,0.0233336){\color[rgb]{0,0,0}\makebox(0,0)[lb]{\smash{sparsity rate $\kappa$}}}%
  \end{picture}%
\endgroup

\caption{\label{fig:two_stage} Illustration of the two-stage sparsity pattern recovery algorithm.}

\end{figure}

The threshold $T$ in the second stage provides a tradeoff between the two kinds of recovery errors: missed detections and false alarms. Throughout this paper, we will assume that that $T$ is chosen as a function of $(\hat{\bX},\kappa)$ such that the estimated sparsity pattern $\hat{S}$ has exactly $k = \lceil \kappa n \rceil$ elements. In practice, this is achieved by thresholding with the magnitude of the $k$'th largest entry in $\hat{\bX}$, and using additional randomness to break ties whenever the $k$'th largest magnitude is not unique.

Conceptually, it is useful to think of the estimate $\hat{\bX}$ generated in the first stage as a direct observation of the original signal that has been corrupted by additive noise, that is we can write $$\hat{\bX} = \bx + \tilde{\bW}$$ 
where $\tilde{\bW}$ is a vector of errors. Along the same lines, the componentwise thresholding in the second stage may be viewed as $n$ independent hypothesis tests under the idealized assumption that the entries of $\tilde{\bW}$ are i.i.d.~and symmetric about the origin.

The main difference between the algorithms studied in this paper is the vector estimator used in the first stage of recovery. In the following subsections, we give bounds on the sampling rate-distortion function corresponding to the maximum likelihood estimator, two different linear estimators (the matched filter and the MMSE), a class of estimators based on approximate message passing, and the MMSE estimator. Our results are summarized in Table~\ref{table:comparison} below. Analysis and illustrations are given in Section~\ref{sec:analysis} and the Appendices.

%%%%%%%%%%%%%%%%%%%
%%%%%%%%%%%%%%%%%%%
\subsection{Maximum Likelihood}

We begin with the method of {\em maximum likelihood} (ML). Conditioned on the realization of the matrix $\bA = A$, the measurements $\bY$ have a multivariate Gaussian distribution with mean $A\bx$ and covariance $\snr^{-1}I_{m \times m}$. Therefore, the ML estimate of sparsity  $k=\lceil \kappa n \rceil$ is given by
\begin{align}
\hat{\bx}^{(\text{ML})}= \arg
\min_{\tilde{\bx} \in \mathbb{R}^n\, : \, \|\tilde{\bx}\|_0 =k } \| \by - A \tilde{\bx}\| \label{eq:L0_minimization}
\end{align}
where $\|\tilde{\bx}\|_0$ denotes the number of nonzero entries in $\tilde{\bx}$. If the minimizer of \eqref{eq:L0_minimization}  is not unique, we will assume that the sparsity pattern estimate $\hat{S}$ in the second stage of the recovery algorithm is drawn uniformly at random from the set
\begin{align*}
\big\{ S \,:\, \text{$S$ is the sparsity pattern of a minimizer of \eqref{eq:L0_minimization}}\big\}.
\end{align*}
This estimator has been studied previously for the task of exact sparsity pattern recovery by Wainwright \cite{Wainwright_InfoLimits_IEEE09} and Fletcher et al. \cite{FLetcher_IEEE09}. 

Before we present our main result, two more definitions are needed. First, we define
\begin{align}
\mathcal{H}(D;\kappa) & = \kappa H_b(D) + (1-\kappa) H_b\Big ( \frac{\kappa D}{1-\kappa}\Big) \label{eq:N_D}
\end{align}
where $H_b(p) = -p \log p -(1-p) \log(1-p)$ is the binary entropy function. In \cite{RG-Lower-Bounds} it is shown that the metric entropy rate for a sequence of sparsity patterns with sparsity rate $\kappa$ under the distortion metric \eqref{eq:distortion} is given by $H_b(\kappa)- \mathcal{H}(D;\kappa)$ for any $D \le 1-\kappa$.

Also, we define
\begin{align}
P(D;p_X) = \int_0^\infty \Big( \Pr[X^2 > u] - (1-D)\kappa\Big)^+ du \label{eq:P_D}
\end{align}
where $(\cdot)^+ = \max(\cdot,0)$. This function corresponds to the average power of the smallest fraction $D$ of nonzero entries. It is a continuous and monotonically increasing function of $D$, with $P(0;p_X) = 0$ and $P(1;p_X) = 1$ for any $p_X \in \mathcal{P}(\kappa)$.

Our first result gives an upper bound on the sampling rate-distortion function corresponding to the ML estimator. The proof is given in Appendix~\ref{sec:proof:ML}.

\begin{theorem}\label{thm:ML}
Under Assumptions S1-S2 and M1-M5, a distortion $D$ is achievable for the tuple $(\rho,p_X,\snr)$ using the ML estimator if $\rho > \rho^{(\text{ML-UB})}$ where
\begin{align}
\rho^{(\text{ML-UB})} = \kappa + \max_{\tilde{D} \in [D,1]} \Lambda(\tilde{D}; p_X,\snr), \label{eq:thm:ML}
\end{align}
with $\Lambda(D;p_X,\snr)$ given in \eqref{eq:C_beta} below. 
\addtocounter{equation}{1}

Moreover, for any $\rho > \rho^{(\text{ML-UB})}$ the error probability $\varepsilon_n^{(\text{ML})}(D)$ decays at least exponentially rapidly with $n$, i.e.~there exists a constant $C$ such that
\begin{align}
\varepsilon_n^{(\text{ML})}(D) \le \exp(- C\, n).
\end{align}
\end{theorem}

\newcounter{tempequationcounter}
\begin{figure*}[!t]
\normalsize
\setcounter{tempequationcounter}{\value{equation}}
\begin{align}
\setcounter{equation}{20}
\Lambda(D; p_X,\snr) & = \min \bigg\{ \frac{2 \mathcal{H}(D;\kappa)}{\log(1 + P(D;p_X) \, \snr ) + (1 + P(D;p_X) \,\snr)^{-1} -1 }, \nonumber \\
&\qquad \qquad \min_{\theta,\mu \in (0,1)} \max\bigg(\frac{2 \mathcal{H}(D;\kappa)}{\log(1 + {\textstyle \frac{1}{4}} (1-\theta)^2 P(D;p_X) \, \snr)}, \frac{2 \mathcal{H}(D;\kappa) - D \kappa \log(1-\mu^2)}{\log(1 + \mu \theta P(D;p_X) \, \snr )}\bigg) \bigg\}
\label{eq:C_beta}
\end{align}
\setcounter{equation}{\value{tempequationcounter}}
\hrulefill
\vspace*{4pt}
\end{figure*}

\begin{remark}
Theorem~\ref{thm:ML} does {\em not} require the convergence of the empirical second moments given in Assumption S3.
\end{remark}

Theorem~\ref{thm:ML} is a significant improvement over previous results in several respects. First, it applies generally to any distribution $p_X$. Second, the bound is given explicitly in terms of the problem parameters and is finite for any nonzero distortion $D$. Finally, as we will show in Section~\ref{sec:analysis} and Appendix~\ref{sec:behavior_ML}, the behavior of the bound, in a scaling sense with respect to the SNR and distortion $D$, is optimal for a large class of distributions.

\begin{cor}\label{cor:ML} The statement of Theorem~\ref{thm:ML} holds if the function $\Lambda(D;p_X,\snr)$ is replaced with any of the following upper bounds:
\begin{align}
\tilde{\Lambda}_1(D;p_X,\snr) &=  \frac{4 \mathcal{H}(D;\kappa)}{\log\big(1 + \big [P(D;p_X) \, \snr  /e\big]^2 \big)} \\
\tilde{\Lambda}_2(D;p_X,\snr) &= \frac{ 2 \mathcal{H}(D;\kappa) +2 \log(5/3)\kappa D}{\log \big(1 + (4/25) P(D;p_X)\,\snr \big)}\\
\tilde{\Lambda}_3(D;p_X,\snr) &= \min_{i \in \{1,2\}} \tilde{\Lambda}_i(D;p_X,\snr).
\end{align}
\end{cor}
\begin{proof}
The bound $\tilde{\Lambda}_1(D;p_x,\snr)$ follows from the first term in \eqref{eq:C_beta} and the simple fact that $\log(1+ x)  - x/(1+x)  \ge (1/2)\log( 1+ x^2/e^2)$ for all $x \ge 0$. The bound $\tilde{\Lambda}_2(D;p_x,\snr)$ follows from the second term in \eqref{eq:C_beta} evaluated with $\mu = 4/5$ and $\theta = 1/5$. 
\end{proof}

%%%%%%%%%%%%%%%%%%%
%%%%%%%%%%%%%%%%%%%
\subsection{Linear Estimation}\label{sec:linear_estimation}

Next, we consider two different linear estimators. The {\em matched filter} (MF) estimate is given by
\begin{align}
\hat{\bx}^{(\text{MF})} &= \textstyle \big(\frac{n}{m}\big) A^T \by
\end{align}
and the {\em linear minimum mean-squared error} (LMMSE) estimate is given by
\begin{align}
\hat{\bx}^{(\text{LMMSE})} &=  (A^T A +  \snr\,  I_{n \times n} )^{-1} A^T \by
\end{align}
These estimators are appealing in practice due to their low computational complexity. Their performance has been studied extensively in the context of multiuser detection with random spreading (see e.g.~\cite{VerSha99,TseHan99}). More recently, the use of the matched filter for the task of sparsity pattern recovery was investigated by Fletcher et al. \cite{FLetcher_IEEE09} and early versions of this paper \cite{RG09}. 

To characterize the behavior of the MF and LMMSE algorithms in the high-dimensional setting, it is useful to introduce a scalar equivalent model of the vector observation model given in \eqref{eq:estimation_problem}. 

\begin{definition}\label{def:scalar_model}
The {\em scalar equivalent model} of \eqref{eq:estimation_problem} is given by
\begin{align}
Z = X + \sigma W \label{eq:estimation_problem_scalar}
\end{align}
where $X \sim p_X$ and $W \sim \mathcal{N}(0,1)$ are independent and $\sigma^2 \in (0,\infty)$ is a fixed parameter called the noise power. 
\end{definition}

In the context of the scalar model, the problem of support recovery is to determine whether or not $X$ is equal to zero. Let $S = \one(X\ne 0)$ be the indicator of this event and let $\hat{S}$ be an estimate of the form $\hat{S}= \one(|Z| > t)$. Then, the detection error probability corresponding to the distortion measure defined  
%s of the MDR, FDR, and distortion measure $d$ defined 
in Section~\ref{sec:formulation_distortion} is given by %the error probabilities
\begin{align}
p_D(t) &=  \max\big(\Pr[\hat{S} = 0| S = 1] ,\Pr[S = 0| \hat{S} = 1]\big)\label{eq:D_awgn}.
\end{align}

We define
\begin{align}
D_\text{awgn}(\sigma^2;p_X) = \min_t p_D(t) \label{eq:D_sig2}
\end{align}
to be a mapping between the noise power $\sigma^2$ and the minimal detection error probability $p_D(t)$ achieved by $\hat{S}$. We also define
\begin{align}
\sigma^2_\text{awgn}(D;p_X)= \sup\{ \sigma^2 \ge 0 \, : \, D_\text{awgn}(\sigma^2;p_X) \le D\}.
\end{align}
to be the inverse mapping. Here, we use the subscript ``awgn'' to emphasize the fact that this error probability corresponds to additive noise $W$ that is Gaussian and independent of $X$.

The following results give an explicit expression for the sampling rate-distortion function of the MF and LMMSE recovery algorithms. Their proofs are given in Appendices~\ref{sec:MF-proof} and \ref{sec:LMMSE-proof} respectively.

\begin{theorem}\label{thm:MF}
Under Assumptions S1-S3 and M1-M4, the sampling rate-distortion function corresponding to the MF estimator is given by
\begin{align}
\rho^{(\text{MF})} &= \frac{1}{\sigma^2\, \snr} + \frac{1}{\sigma^2} \label{eq:MF}
\end{align}
where $\sigma^2 = \sigma^2_\text{awgn}(D;p_X)$. 
\end{theorem}

\begin{remark}
Theorem~\ref{thm:MF} does {\em not} require the measurement matrix $\bA(n)$ to be Gaussian.
\end{remark}

\begin{theorem}\label{thm:LMMSE}
Under Assumptions S1-S3 and M1-M5, the sampling rate-distortion function corresponding to the LMMSE estimator is given by
\begin{align}
\rho^{(\text{LMMSE})} &=\frac{1}{\sigma^2\, \snr} + \frac{1}{1+ \sigma^2} 
\label{eq:LMMSE}
\end{align}
where $\sigma^2 = \sigma^2_\text{awgn}(D;p_X)$. 
\end{theorem}

Recall that our definition of achievability says that the probability that the distortion $d(S^*,\hat{S})$ exceeds a threshold $D$ must tend to zero as $n$ becomes large. For the MF and LMMSE estimators, convergence of the expected distortion $\bE[d(S^*,\hat{S})]$ can be established straightforwardly using results in \cite{VerSha99} and \cite{TseHan99}. Therefore, the key contribution of Theorems~\ref{thm:MF} and \ref{thm:LMMSE} is to show that this convergence holds also in probability. For the MF estimator, this is achieved using a general decoupling result which applies generally for any i.i.d.~distribution on the measurement matrix. For the LMMSE estimator, we use the fact that the LMMSE can be computed using the AMP algorithm discussed in the next section.

%%%%%%%%%%%%%%%%%%%
%%%%%%%%%%%%%%%%%%%
\subsection{Approximate Message Passing}\label{sec:AMP}

We now consider estimation using {\em approximate message passing} (AMP) \cite{DonMalMon09}. The AMP algorithm is characterized in terms of a scalar de-noising function $\eta(z,\sigma^2)$ which is assumed to be Lipschitz continuous with respect to its first argument and continuous with respect to its second argument. Starting with initial conditions $\bx^0 = \mathbf{0}_{n \times 1}$, $\bu^0 = (\frac{n}{m}) \by$ and $\hat{\sigma}^2_0 = (\snr^{-1} + 1)/\rho$, the algorithm proceeds for iterations $t=1,2,\cdots$ according to
\begin{align}
\bx^{t} &= \eta\big(A^T \bu^{t-1} + \bx^{t-1}, \hat{\sigma}^2_{t-1}\big)  \label{eq:amp_iteration_a}\\
\bu^{t} & = {\textstyle \big(\frac{n}{m} \big)} \Big[  \by - A \bx^{t}\nonumber\\
& \quad  + \bu^{t-1} \frac{1}{n} \sum_{i=1}^n \eta'\Big( \big(A^T \bu^{t-1} + \bx^{t-1}\big)_i, \hat{\sigma}^2_{t-1}\Big) \Big] \label{eq:amp_iteration_b}\\
\hat{\sigma}^2_{t}& = \frac{1}{n} \|\bu^{t}\|^2, \label{eq:sig2_hat_t}
\end{align}
where $\eta'(z,\sigma^2)$ denotes the partial derivative of $\eta(z,\sigma^2)$ with respect to $z$, and, for any vector $\bz$, $\eta(\bz,\sigma^2)$ denotes the vector obtained by applying the function $\eta(z,\sigma^2)$ componentwise.

The AMP algorithm is said to succeed if the tuple $(\bx^t,\bu^t, \hat{\sigma}^2_t)$ converges to a fixed point $(\bx^\infty,\bu^\infty, \hat{\sigma}^2_\infty)$. Various stability assumptions guaranteeing  convergence of the algorithm are discussed in \cite{DonMalMon09,DonMalMon10}. In some cases, the rate of convergence is exponential in the number of iterations. 

\begin{remark} Our update equations for the AMP algorithm differ slightly from those given in \cite{DonMalMon09,DonMalMon10,BayMon10a}. This difference is due to the fact that this paper considers row normalization of the measurement matrix (Assumption M3) whereas the previous work considers column normalization.
\end{remark}
 
Conceptually, it is useful to think of the vector $\bx^t$, generated in the $t$'th iteration of the AMP algorithm, as a noisy version of the original vector $\bx$ that has been passed through the scalar de-noising function $\eta(\cdot,\hat{\sigma}^2_{t-1})$. More specifically, we can write
\begin{align}
\bx^t = \eta(\bx + \tilde{\bw}^{t-1};\hat{\sigma}^2_{t-1})
\end{align}
where
\begin{align}
\tilde{\bw}^{t-1} = A^T \bu^{t-1} + \bx^{t-1} - \bx \label{eq:amp_error}
\end{align}
is a vector of errors. 

In \cite{DonMalMon09,DonMalMon10}, it is shown, both heuristically and empirically, that, under Assumptions S1-S3 and M1-M5 of this paper, the error vector $\bw^{t-1}$ defined in \eqref{eq:amp_error} behaves similarly to additive white Gaussian noise with mean zero and variance $\hat{\sigma}^2_{t-1}$. A precise statement of this behavior, corresponding to the empirical marginal distribution of the tuple $(\bx,\bx^t,\tilde{\bw}^t$), is proved in ensuing work by Bayati and Montanari \cite{BayMon10a}. See Appendix~\ref{sec:thm_MF_proof} for more details.

At this point, we are faced with the following question: based on the output $(\bx^\infty, \bu^\infty,\hat{\sigma}^2_\infty)$ of the AMP algorithm, what should we choose as an estimate $\hat{\bx}$ of the unknown vector $\bx$? In previous work, where the primary objective is to minimize the MSE, the output $\hat{\bx}^\infty$ is used as an estimate of $\bx$ (see e.g.~\cite{BayMon10b}). The main reason for using this estimate is that the function $\eta(\cdot,\sigma^2)$ provides a scalar de-noising step that reduces the effect of the additive error $\tilde{\bw}$. 

In this paper, however, our primary objective to is generate an estimate of $\bx$ that leads to an accurate estimate of the sparsity pattern in the second stage of estimation. As such, the final scalar de-noising step is unnecessary, and potentially counterproductive. To see why, note that the componentwise thresholding in the second stage of recovery depends entirely on the relative magnitudes of the entries in $\hat{\bx}$. If the denoiser does not preserve the ranking of these magnitudes (e.g.~if many nonzero values are mapped to zero), then relevant information about the sparsity pattern is lost. 

Accordingly, we use the vector estimate given by
\begin{align}
\hat{\bx}^{(\text{AMP})} = A^T\bu^{\infty} + \bx^{\infty}. \label{eq:x_hat_amp}
\end{align}
Since the AMP output $(\bx^\infty, \bu^\infty,\hat{\sigma}^2_\infty)$ satisfies the fixed point equation $$ \bx^\infty = \eta( A^T\bu^{\infty} + \bx^{\infty}, \hat{\sigma}^2_\infty),$$
we see that our estimate corresponds directly to the signal-plus-noise estimate $\bx + \tilde{\bw}^\infty$ prior to the scalar de-noising. 

To characterize the behavior of AMP in the high-dimensional setting, we return to the scalar equivalent model given in Definition~\ref{def:scalar_model}. We define the scalar mean-squared error function
\begin{align}
\mse(\sigma^2;p_X,\eta) = \bE\big[ \big|X - \eta(X + \sigma W,\sigma^2 )\big|^2\big]
\end{align}
where $X \sim p_X$ and $W \sim \mathcal{N}(0,1)$ are independent, and let $\{\sigma_t^2\}_{t \ge 1}$ be a sequence of noise powers defined by the recursion
\begin{align}
\sigma^2_{t} = \frac{ \snr^{-1} +  \mse(\sigma^2_{t-1};p_X,\eta)}{\rho}  \label{eq:state_evolution}
\end{align}
where $\sigma_0^2 = (\snr^{-1} + 1)/\rho$. This recursion is referred to as {\em state evolution} \cite{DonMalMon09}. 

The following result shows that the distortion corresponding to the AMP estimate is characterized by the state evolution recursion. In Appendix~\ref{sec:amp_general-proof}, it is shown how this result follows straightforwardly from recent work of Bayati and Montanari \cite{BayMon10a}.

\begin{theorem}\label{thm:amp_general}
Suppose that the noise powers defined by the state evolution recursion \eqref{eq:state_evolution} converge to a finite limit
\begin{align}
\sigma^2_\infty = \lim_{t \rw \infty} \sigma^2_t. \label{eq:sig2_infty}
\end{align}
Then, under Assumptions S1-S3 and M1-M5, the distortion $d(S^*,\hat{S})$ corresponding to the AMP estimator converges in probability as $n \rw \infty$ to the limit $D_\text{awgn}(\sigma^2_\infty;p_X)$.
\end{theorem}

\begin{remark}
We note that the limiting noise power $\sigma^2_\infty$ is a function of the tuple $(\rho,p_X,\snr)$ and the function $\eta(z,\sigma^2)$. In some cases, it is possible that $\sigma^2_\infty$ is an increasing function of $\rho$, and thus increasing the sampling rate increases the distortion.\end{remark}

In the following subsections, two special cases of the AMP estimator are considered. 

\subsubsection{Optimized AMP}
If the limiting distribution $p_X$ is known, then the limiting noise power $\sigma^2_\infty$ is minimized when $\eta(z,\sigma^2)$ is given by the conditional expectation
\begin{align}
\eta^{(\text{MMSE})} (z,\sigma^2;p_X)= \bE[X| X + \sigma W=z]
\end{align}
corresponding to the distribution $p_X$. 
We will refer to this version of the AMP algorithm as AMP-MMSE, and we define the corresponding mean-squarred error function 
\begin{align}
\mmse(\sigma^2;p_X) = \bE\big[ \big|X - \bE[X|X + \sigma W] \big|^2\big].
\end{align}

\begin{theorem}\label{thm:AMP_MMSE}
Under Assumptions S1-S3 and M1-M5, 
the sampling rate-distortion function corresponding to the AMP-MMSE estimator is given by
\begin{align}
\rho^{(\text{AMP-MMSE})} = \sup_{\tau \ge \sigma^2} \left\{ \frac{\snr^{-1} +  \mmse(\tau;p_X)}{\tau} \right\} \label{eq:amp-mmse}
\end{align}
where $\sigma^2 = \sigma^2_\text{awgn}(D;p_X)$. 
\end{theorem}

\begin{proof}
By the definition of the MMSE, we have $\mmse(\sigma^2;p_X)< \bE[X^2]=1$ for all $\sigma^2<\infty$. Therefore, any solution $\sigma^2$ to the fixed point equation 
\begin{align}
\sigma^2 = \frac{\snr^{-1} + \mmse(\sigma^2;p_X)}{\rho}\label{eq:fixed_point_mmse}
\end{align}
is strictly less than the initial noise power $\sigma^2_0$. Since $\mmse(\sigma^2;p_X)$ is a strictly decreasing function of $\sigma^2$, it thus follows that the limit $\sigma^2_\infty$ always exists and is given by the largest solution to \eqref{eq:fixed_point_mmse}, i.e.
\begin{align}
\sigma^2_\infty = \sup\left\{ \tau \ge 0 \,: \rho = \frac{ \snr^{-1} + \mmse(\tau;p_X)}{\tau}\right\}. \label{eq:sig2_mmse}
\end{align}
Since the right hand side of \eqref{eq:sig2_mmse} is a strictly decreasing function of $\rho$, Theorem~\ref{thm:AMP_MMSE} follows directly from Theorem~\ref{thm:amp_general} and the definition of the sampling rate-distortion function.
\end{proof}

It is important to note that the AMP-MMSE estimate is a function of the distribution $p_X$. If this distribution is unknown and the estimate is made using a postulated distribution that differs from the true one, then the performance of the algorithm could be highly suboptimal.

\subsubsection{Soft Thresholding}
Another special case of the AMP algorithm is when $\eta(z,\sigma^2)$ is given by the soft thresholding function
\begin{align}
\eta^{(\text{ST})}(z,\sigma^2;\alpha)&= 
\begin{cases}
z + \alpha \sigma, & \text{if $z < -\alpha \sigma$}\\
0,& \text{if $|z| \le \alpha \sigma$}\\
z - \alpha \sigma,& \text{if $z \ge \alpha \sigma$}
\end{cases}
\end{align}
for some threshold $\alpha \ge 0$. We will refer to this algorithm as AMP-ST.

\begin{remark}
It is argued in \cite{DonMalMon10} and shown rigorously in \cite{BayMon10b} that, for a fixed set $(p_X,\snr)$, the behavior of AMP-ST is equivalent to that of LASSO \cite{Tibsh96} under an appropriate calibration between the threshold $\alpha$ and the regularization parameter of LASSO. 
\end{remark}

To characterize the behavior of AMP-ST, we follow the steps outlined by Donoho et al. \cite{DonMalMon10} and define the noise sensitivity
\begin{align}
\mathcal{M}(\sigma^2,\alpha;p_X)  = \frac{\mse(\sigma^2;p_X,\eta^{(\text{ST})})}{\sigma^2}. \label{eq:noise_sensitivity}
\end{align}

\begin{theorem}\label{thm:AMP_ST}
Under Assumptions S1-S3 and M1-M5, the sampling rate-distortion function corresponding to the AMP-ST estimator is given by
\begin{align}
\rho^{(\text{AMP-ST})} = \frac{1}{\sigma^2 \snr} + \mathcal{M}(\sigma^2,\alpha;p_X)\label{eq:AMP-ST}
\end{align} 
where $\sigma^2 =  \sigma^2_\text{awgn}(D;p_X)$.   
\end{theorem}
\begin{proof}
This result is an immediate consequence of Theorem~\ref{thm:amp_general} and \cite[Lemma 4.1]{DonMalMon10} which shows that $\sigma^2_\infty$ exists and is given by the unique solution to the fixed point equation
\begin{align}
\rho = \frac{1}{\sigma^2_\infty\, \snr} + \mathcal{M}(\sigma^2_\infty,\alpha;p_X). \label{eq:fixed_point_ST}
\end{align}
\end{proof}

We note that Theorem~\ref{thm:AMP_ST} can be used to find the optimal value for the soft-thresholding parameter $\alpha$. If, for example, the goal is to minimize the sampling rate $\rho$ as a function of the tuple $(D,p_X,\snr)$, then the optimal value of $\alpha$ is given by the minimizer of $\mathcal{M}(\sigma^2,\alpha;p_X)$. Conversely, if the goal is to minimize the distortion $D$ as a function of the tuple $(\rho,p_X,\snr)$, then the optimal value of $\alpha$ is one that minimizes the value of $\sigma^2_\infty$ in the fixed point equation  \eqref{eq:fixed_point_ST}.

We emphasize that the soft-thresholding function is, in general, suboptimal for a given distribution $p_X$ (recall that the optimal version of AMP is given by AMP-MMSE). The main reason we study soft-thresholding is to deal with settings where the distribution $p_X$ is unknown. In Appendix~\ref{sec:ST_behavior}, it is shown how the function $\mathcal{M}(\sigma^2,\alpha;p_X)$ can be upper bounded uniformly over the class of distributions $\mathcal{P}_\kappa$, and how combining this upper bound with Theorem~\ref{thm:AMP_ST} gives bounds on the sampling rate-distortion function that hold uniformly over any class of distributions $\mathcal{P}_X \subset \mathcal{P}(\kappa)$.

%%%%%%%%%%%%%%%%%%%
%%%%%%%%%%%%%%%%%%%
\subsection{Minimum Mean-Squared Error via the Replica Method} 

Lastly, we consider the performance of the {\em minimum mean-squared error} (MMSE) estimator. For a known distribution $p_X$, this estimator is given by the conditional expectation
\begin{align}
\bx^{(\text{MMSE})} = \bE[\bX| A\bX + \snr^{-1/2} \bW = \by],
\end{align}
where the entries of $\bX$ are i.i.d.~$p_X$. %Heuristically, it can be argued that the MMSE estimator is the optimal estimator for our two-stage recovery framework analyzed in this paper. 

To analyze the behavior of the MMSE estimator, we develop a result based on the powerful but heuristic {\em replica method} from statistical physics. This method was developed originally in the context of spin glasses \cite{EdwAnd75} and has been applied to the vector estimation problem studied in this paper by a series of recent papers \cite{Tanaka02, Muller03, GuoVer05,KabWadTan09, RanFleGoy09, GuoBarSha09}.

In the replica analysis, the unknown vector is modeled as a random vector $\bX$ whose entries are  i.i.d. $p_X$. Accordingly, each realization of the measurement matrix $\bA = A$, induces a joint probability measure on the random input-output pair $(\bX,\bY)$, or equivalently on the random input-estimate pair $(\bX, \hat{\bX})$. At this point, the key argument exploited by the replica method is that, due to a certain type of ``replica symmetry'' in the problem, the joint probability measure on $(\bX, \hat{\bX})$ behaves similarly for all typical realizations of the measurement matrix $\bA$ in the high-dimensional setting. Based on this assumption, it can then be argued that the marginal joint distribution on the entries in $(\bX,\hat{\bX})$ converges to a nonrandom limit, characterized by the tuple $(\rho,p_X, \snr)$.

A detailed explanation of the replica analysis is beyond the scope of this paper.  The assumptions needed for our results are summarized below. 

%The analysis using the replica method requires several assumptions. 
\smallskip

{\noindent \em Replica Analysis Assumptions:} The key assumptions underlying the replica analysis are stated explicitly by Guo and Verd\'{u} in \cite{GuoVer05}. A concise summary can also be found in \cite[Appendix A]{RanFleGoy09}. Two assumptions that are used---and generally accepted throughout the literature---are the validity the ``replica trick'' and the self averaging property of a certain function defined on the random matrix $\bA$. A further assumption that is also required is that of {\em replica symmetry}. This last assumption is problematic, however, since it is known that there are cases where it does not hold, and there is currently no test to determine whether or not it holds in the setting of this paper. 
\smallskip

The following result characterizes the sampling rate-distortion function corresponding to the MMSE estimator under the condition that the replica assumptions are valid. The proof is given in Appendix~\ref{sec:mmse-proof}.

\begin{theorem}\label{thm:MMSE}
Assume that the replica analysis assumptions hold. Under Assumptions S1-S3 and M1-M4, the distortion $d(S^*,\hat{S})$ corresponding to the MMSE estimator converges in probability as $n \rw \infty$ to the limit $D_\text{awgn}(\tau^*;p_X)$ where
\begin{align}
\tau^* =  \arg \min_{\tau> 0}\Big\{ \rho \log \tau +  \frac{1}{ \tau\, \snrs} + 2 I(X;X+\sqrt{\tau} W)\Big\} .\label{eq:F_tau}
\end{align}
with $X \sim p_X$ and $W \sim \mathcal{N}(0,1)$ independent. 
\end{theorem}

We emphasize that a key difference between Theorem~\ref{thm:MMSE} and the previous bounds in this paper is that the replica analysis assumptions on which it is based are currently unproven. In the context of the recovery problem outlined in this paper, this means that Theorem~\ref{thm:MMSE} provides only a {\em heuristic prediction} for the true behavior of the MMSE estimator. The validity of this prediction for the setting of the paper depends entirely on the validity of the replica assumptions. 

In the next section, we will see that there are many parameter regimes in which the replica prediction for the MMSE estimator is tightly sandwiched between the rigorous upper given earlier in this paper and the information-theoretic lower bound in \cite{RG-Lower-Bounds}. Thus, beyond the context of sparsity pattern recovery, a significant contribution of this paper is that we provide strong evidence in support of the replica analysis assumptions.

\begin{remark}
One interesting implication of Theorem~\ref{thm:MMSE} is that the AMP-MMSE estimate is equivalent to the MMSE estimate whenever the noise power $\tau^*$ defined in \eqref{eq:F_tau} is equal to the limit $\sigma^2_\infty$ defined in \eqref{eq:sig2_mmse}. This suggests that the MMSE estimate can be computed efficiently in some problem regimes. 
\end{remark}

Finally, it is important to note that MMSE estimator is a function of the limiting distribution $p_X$. If this distribution is unknown and the estimate is made using a postulated distribution that differs from the true one, then the performance could be highly suboptimal. Using further results developed in \cite{GuoVer05} it is possible to characterize the sampling rate in terms of an arbitrary postulated prior and true limiting distribution. Such analysis, however, is beyond the scope of this paper.

%%%%%%%%%%%%%%%%
%\input{sections/analysis_upper}
\section{Analysis and Illustrations}\label{sec:analysis}

In this section, we show how the sampling rate-distortion functions given in Section~\ref{sec:main_results} depend on the desired distortion $D$, the SNR, and various properties of the distribution $p_X$. By comparison with the information-theoretic lower bounds in \cite{RG-Lower-Bounds}, we characterize problem regimes in which the behavior of the algorithms is near-optimal and other regimes in which the behavior is highly suboptimal.

%%%%%%%%%%%%%%%%%%%%%%%%%%
\subsection{Signal Classes}\label{sec:signal_classes}

Following the problem formulation outlined in Section~\ref{sec:model_assumptions}, a class of signals can be characterized by a of a class of limiting distributions $\mathcal{P}_X \subset \mathcal{P}(\kappa)$ where $\mathcal{P}(\kappa)$ is the class of all probability measures with second moment equal to one and probability mass $1-\kappa$ at zero. To facilitate our analysis in the following sections, we introduce the following three classes:

\begin{itemize}
\item {\em Bounded: } We use $\mathcal{P}_\text{Bounded}(\kappa,B)$ to denote the class of all distributions $p_X \in \mathcal{P}(\kappa)$ such that $$\Pr[ |X| < B | X \ne 0] = 0$$
for some {\em lower bound} $B >0$. Due to the second moment constraint, the lower bound $B$ cannot exceed $1/\sqrt{\kappa}$.

\item {\em Polynomial Decay: } We use $\mathcal{P}_\text{Poly.}(\kappa,L,\tau)$ to denote the class of all distributions $p_X\in \mathcal{P}(\kappa)$ such that $$\lim_{x \rw 0} \frac{\Pr[ |X| \le x | X \ne 0]}{x^L} = \tau$$ for some {\em polynomial decay rate} $L >0$ and limiting constant $\tau \in (0,\infty)$. 

\item {\em Bernoulli-Gaussian: }  We say that a distribution $p_X$ is Bernoulli-Gaussian with sparsity 
$\kappa$ if the nonzero part of $p_X$ is zero-mean Gaussian, i.e. if
\begin{align*}
X \sim 
\begin{cases}
0, & \text{with probability $1-\kappa$}\\
\mathcal{N}(0, \frac{1}{\kappa}), & \text{ with probability $\kappa$}
\end{cases}.
\end{align*}
\end{itemize}

The bounded class corresponds to the setting where the nonzero entries in $\bx$ have a fixed lower bound $B$ on their magnitudes, independent of the vector length $n$. By contrast, the polynomial decay class corresponds to the setting where the magnitude of the $\lceil \beta \,k\rceil$'th smallest nonzero entry is proportional to $\beta^{1/L}$ for small $\beta$. Note that in the case of polynomial decay, a vanishing fraction of the nonzero entries are tending to zero as the vector length $n$ becomes large.

The Bernoulli-Gaussian distribution is an example of a distribution with polynomial decay rate $L=1$ and limiting constant $\tau = \sqrt{2 \kappa /\pi}$.

%%%%%%%%%%%%
\subsection{Illustrations}

In the following sections we provide illustrations of the bounds derived in Section~\ref{sec:main_results} corresponding to either the Bernoulli-Gaussian distribution or the class of bounded distributions $\mathcal{P}_\text{Bounded}(\kappa,B)$ with lower bound $B= \sqrt{0.2/\kappa}$. Note that this choice of $B$ means that the nonzero entries in $\bx$ are lower bounded in squared magnitude by $20\%$ of their average power. 

The bounds corresponding to the Bernoulli-Gaussian distribution are optimized as a function of the relevant parameters. For the AMP-MMSE and MMSE bounds, this means that the true distribution $p_X$ is used to define the conditional expectations. For the AMP-ST bound, this means that the threshold $\alpha$ is chosen to either minimize the distortion as a function of the sampling rate or to minimize the sampling rate as a function of the distortion.

In order to derive uniform bounds for the class of bounded distributions $\mathcal{P}_\text{Bounded}(\kappa,B)$, it is necessary to consider the worst-case distribution in the class. For the ML and linear estimators, these bounds are obtained straightforwardly by lower bounding the functions $P(D;p_X)$ and $\sigma^2_\text{awgn}(D;p_X)$ (see Proposition~\ref{prop:bounded} below). For the AMP-ST we obtain a uniform bound by replacing the noise sensitivity $\mathcal{M}(\sigma^2,\alpha;p_X)$ in Theorem~\ref{thm:AMP_ST} with the upper bound $\mathcal{M}^*(\sigma^2,\alpha,\kappa)$ given in Appendix~\ref{sec:ST_behavior}, and then optimizing the resulting expression as a function of the threshold $\alpha$. Uniform bounds corresponding to the AMP-MMSE and MMSE cannot be derived using the results in this paper, since these estimators depend on the true underlying distribution $p_X$. 

For comparison, we also plot corresponding information-theoretic lower bounds derived in \cite{RG-Lower-Bounds}. These bounds correspond to the performance of the optimal sparsity pattern recovery algorithm under assumptions S1-S3 and M1-M4.

All illustrations correspond to a sampling rate of $\kappa = 10^4$. The qualitative behavior of the bounds does not change significantly for sparsity rates within several orders of magnitude of this value.

%%%%%%%%%%%%%%%%%%%%%%%%%%%%%
\subsection{Sampling Rate versus SNR}\label{sec:rho_SNR}

We begin our analysis of the bounds by studying the tradeoff between sampling rate and SNR. For a given recovery algorithm ALG, we use $\rho^{(\text{ALG})}_\infty$ to denote the infinite SNR limit of the sampling rate-distortion function:
\begin{align}
\rho^{(\text{ALG})}_\infty = \lim_{\snrs \rw \infty} \rho^{(\text{ALG})}.
\end{align}
This limit is a function of the pair $(D,p_X)$ and may be interpreted as the sampling rate required in the absence of noise. 

For the ML estimator, the infinite SNR limit of the upper bound in Theorem~\ref{thm:ML} is given by the sparsity rate $\kappa$, regardless of the distribution $p_X$ and distortion $D$. Since it can be shown that the ML estimate is equivalent to random guessing whenever $\rho < \kappa$, we thus conclude that the infinite SNR limit of $\rho^{(\text{ML})}$ is given explicitly by the piecewise constant function
\begin{align}
\rho^{(\text{ML})}_\infty  &= \begin{cases}
 \kappa, & \text{if $D \le 1-\kappa$}\\
 0, &\text{if $D > 1-\kappa$}
 \end{cases}.
\end{align}

An upper bound on the rate at which $\rho^{(\text{ML})}$ approaches its infinite SNR limit is given by the following result. The proof follows directly from the analysis of Theorem~\ref{thm:ML} given in Appendix~\ref{sec:behavior_ML}.

\begin{prop}\label{prop:ML_high_SNR}
For any nonzero distortion $D $ and distribution $p_X$, there exists a constant $C$ such that
\begin{align}
\rho^{(\text{ML})} \le  \kappa  + \frac{C}{\log(1+\snr)}. 
\end{align}
\end{prop}

The following result from \cite{RG-Lower-Bounds} shows that under some additional assumptions on the pair $(D,p_X)$, Proposition~\ref{prop:ML_high_SNR} is tight, in a scaling sense, with respect to the SNR.

%%%%%%%%%%%%
%%%%%%%%%%%%
\begin{prop}\label{prop:opt_high_SNR}\cite{RG-Lower-Bounds} Suppose that $p_X$ can be expressed as
\begin{align}
p_X = (1-\kappa) \delta_0 + \omega_c p_{X_c} + (\kappa - \omega_c) p_{X_d}
\end{align}
where $X_c$ is continuous with finite differential entropy $h(X_c)$ and $X_d$ is discrete. Let $D < 1-\kappa$ be any distortion that satisfies
\begin{align}
2H_b(\kappa_c) - 2 \mathcal{H}\big(\tfrac{\kappa}{\omega_c} D ; \kappa_c\big) & > \kappa _c \log \bigg( \frac{ \bE[X_c^2] - \kappa_c (\bE[X_c])^2}{N(X_c)} \bigg)  \nonumber \\
& \quad + (1-\kappa_c) \log\Big( \frac{1}{1-\kappa_c}\Big) \label{eq:D_Theta_LB}
\end{align}
where $\kappa_c = \omega_c/(1-\kappa + \omega_c)$ and $N(X_c) = (2 \pi e)^{-1} \exp(2 h(X_c))$. Then, under Assumptions S1-S2 and M1-M4, there exists a constant $C$ such that
\begin{align}
\rho > \omega_c +  \frac{C}{\log(1+\snr)} \label{eq:rho_snr_scaling_LB}
\end{align}
is a necessary condition for any recovery algorithm.
\end{prop}
%%%%%%%%%%%%
%%%%%%%%%%%%

%\begin{prop}\label{prop:opt_high_SNR}\cite{RG-Lower-Bounds} Suppose that $X\sim p_X$ can be expressed as
%\begin{align}
%X\sim  
%\begin{cases}
%0, & \text{with probability $1-\kappa$}\\
%X_c & \text{with probability $\omega_c$}\\
%\tilde{X} & \text{with probability $\kappa - \omega_c$}\\
%\end{cases}\label{eq:X_mixtures}
%\end{align}
%where $0 < \omega_c \le \kappa$ and $X_c$ is a continuous random variable with finite differential entropy $h(X_c)$. 
%Let $D < 1-\kappa$ be any distortion that satisfies
%\begin{align}
%H_b(\kappa) - \mathcal{H}(D;\kappa) > \frac{\omega_c}{2}  \log\bigg(\frac{ (1 - \omega_c/\kappa)^{\kappa/\omega_c - 1}}{\Theta(p_X) (1-\omega_c)^{1/\omega_c-1}} \bigg) \label{eq:D_Theta_LB}
%\end{align}
%where $\Theta(p_X)$ is the normalized entropy power
%\begin{align}
%\Theta(p_X) = \frac{\omega_c \exp(2 h(X_c))}{2 \pi e \text{\normalfont Var}(X)}
%\end{align}
%and $\mathcal{H}(D;\kappa)$ is defined in \eqref{eq:N_D}. Then, under Assumptions S1-S2 and M1-M4, there exists a constant $C$ such that
%\begin{align}
%\rho > \omega_c +  \frac{C}{\log(1+\snr)} \label{eq:rho_snr_scaling_LB}
%\end{align}
%is a necessary condition for any recovery algorithm.
%\end{prop}

Note that the constant $\omega_c$ in Proposition~\ref{prop:opt_high_SNR} is equal to the sparsity rate $\kappa$ whenever the nonzero part of $p_X$ is absolutely continuous with respect to Lebesgue measure. When this occurs, Propositions~\ref{prop:ML_high_SNR} and \ref{prop:opt_high_SNR} characterize the fundamental behavior of the recovery problem for any distortion $D$ satisfying \eqref{eq:D_Theta_LB}.

For the linear and AMP estimators, it is straightforward to show that the infinite SNR limits can be expressed as 
\begin{align}
\rho^{(\text{MF})}_\infty &= \frac{1}{\sigma^2}\\
\rho^{(\text{LMMSE})}_\infty &= \frac{1}{1+ \sigma^2}\\
\rho^{(\text{AMP-ST})}_\infty  &= \mathcal{M}(\sigma^2,\alpha,p_x)\\
\rho^{(\text{AMP-MMSE})}_\infty  &= \sup_{\tau > \sigma^2} \frac{\mmse(\tau;p_X)}{\tau}
\end{align}
where $\sigma^2 = \sigma^2_\text{awgn}(D;p_X)$. By comparison with the ML limit, we see that each of these algorithms is strictly suboptimal at high SNR whenever its limit exceeds the sparsity rate $\kappa$.

For the MMSE estimator, the infinite SNR limit of the sampling rate predicted by the replica method in Theorem~\ref{thm:MMSE} is characterized by the infinite SNR limit of the noise power $\tau^*$ given in \eqref{eq:F_tau}. It is easy to check that this limit is always less than or equal to $\kappa$, and thus the predicted MMSE infinite SNR limit is upper bounded by the ML infinite SNR limit. 

The rate at which the achievable sampling rates converge to their infinite SNR limits is illustrated in Fig~\ref{fig:rho_SNR_Gaussian} for the Bernoulli-Gaussian distribution. The relative tightness of the ML upper bound and the information-theoretic lower bound from \cite{RG-Lower-Bounds} provides rigorous verification of the MMSE behavior derived heuristically using the replica method. Moreover, as the SNR becomes large, the bounds corresponding to the AMP and linear estimate are significantly greater than the ML bounds, thus indicating that these methods are highly suboptimal at high SNR.

\begin{figure}[htbp]
\centering
\epsfig{file=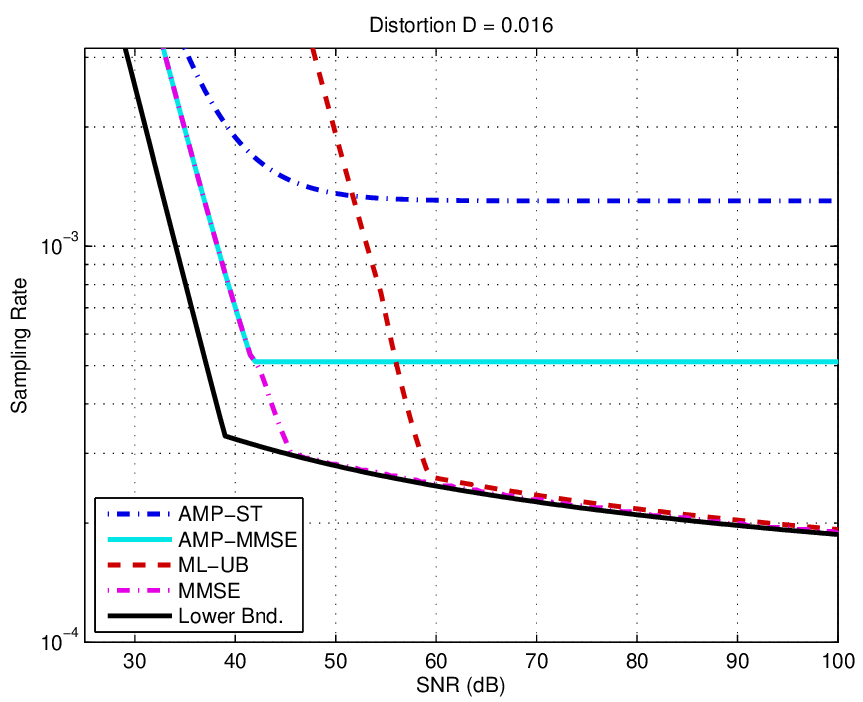, width = .46\textwidth}
\caption{Bounds on the achievable sampling rate $\rho$ as a function of the SNR when the nonzero entries are i.i.d.~zero-mean Gaussian and the sparsity rate is $\kappa = 10^{-4}$.} 
\label{fig:rho_SNR_Gaussian}
\end{figure}

In Fig.~\ref{fig:rho_D_inf_Gaussian}, the infinite SNR limits corresponding to the Bernoulli-Gaussian distribution are shown as a function of the distortion. For this distribution, the MMSE limit is equal to the minimum of the ML and AMP-MMSE limits. When the distortion is relatively small (i.e.~less than $\approx 0.9$), the  limits for ML, MMSE, and the information-theoretic lower bound are equal to the sparsity rate $\kappa$. When the distortion is relatively large, all of the bounds except for the ML bound converge to zero. If the goal is to minimize the distortion $D$ as a function of the sampling rate $\rho$, then this behavior shows that ML is strictly suboptimal whenever the sampling rate $\rho$ is strictly less than the sparsity rate $\kappa$. 

\begin{figure}[htbp]
\centering
\epsfig{file=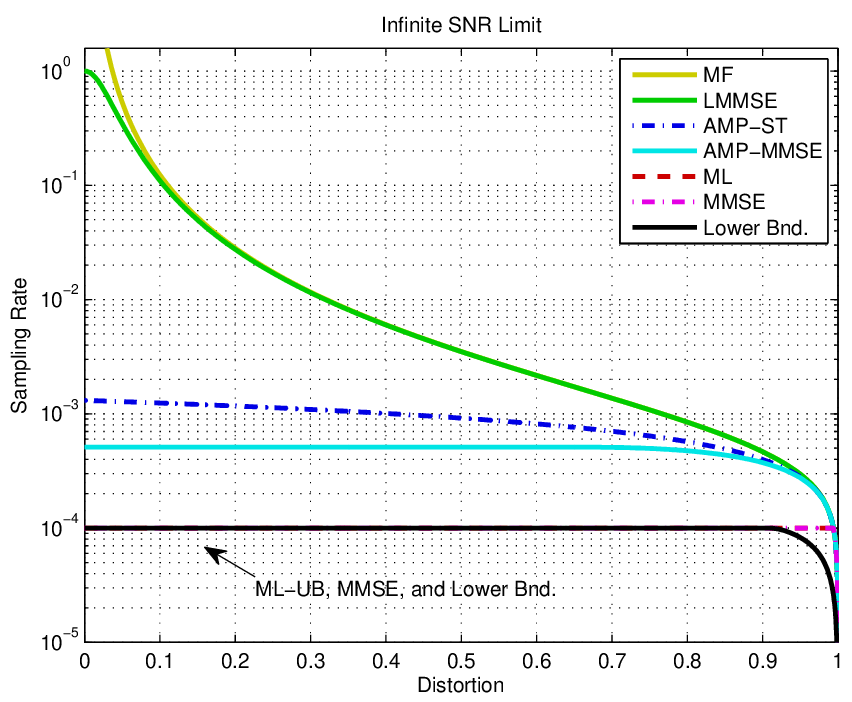, width = .46\textwidth}
\epsfig{file=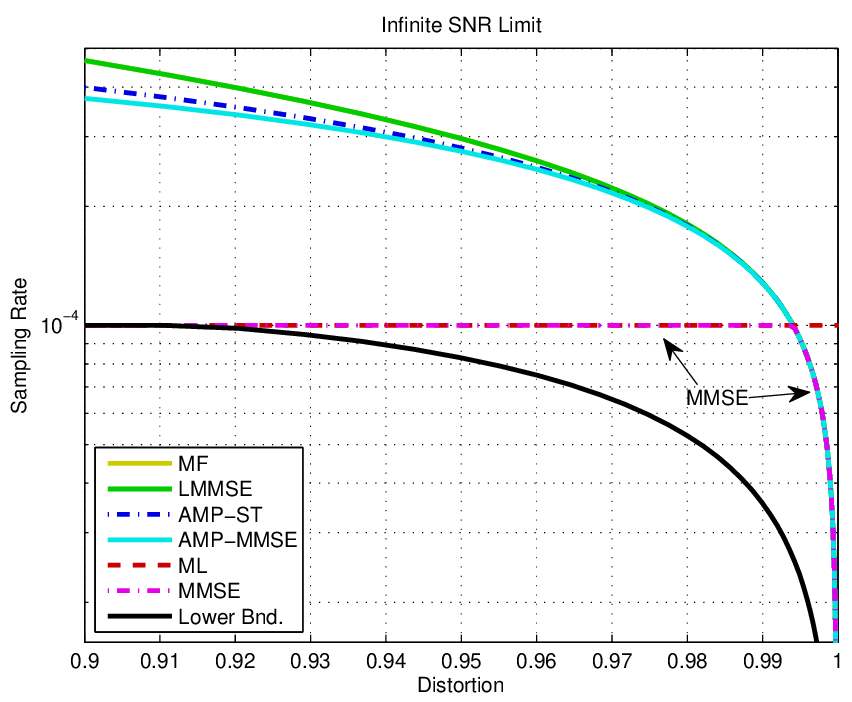, width =.46\textwidth}
\caption{Bounds on the infinite SNR limit of the achievable sampling rate $\rho$ as a function of the distortion $D$ when the nonzero entries are i.i.d.~zero-mean Gaussian and the sparsity rate is $\kappa = 10^{-4}$. The bottom panel highlights the large distortion behavior.} 
\label{fig:rho_D_inf_Gaussian}
\end{figure}

%%%%%%%%%%%%%%%%%%%%%%%%%%%%%%
\subsection{Stability Thresholds}\label{sec:stability_thresholds}

For a given recovery algorithm ALG, we define the {\em stability threshold} as follows:
\begin{align}
\varrho^{(\text{ALG})} = \lim_{D \rw 0} \, \lim_{\snrs \rw \infty}\rho^{(\text{ALG})}. \label{eq:stability_threshold_def}
\end{align}
This threshold is a function of the distribution $p_X$ and may be interpreted as the sampling rate required for exact recovery in the absence of noise. For future reference, its significance is summarized in the following result.  

\begin{prop}\label{prop:stability_threshold} Consider a fixed recovery algorithm ALG and distribution $p_X$ with stability threshold $\varrho^{(\text{ALG})}$.
\begin{enumerate}[(a)]
\item If $\rho > \varrho^{(\text{ALG})}$, then recovery is {\em stable} in the sense that the distortion $D$ can be made arbitrarily small by increasing the SNR. 
\item If $\rho < \varrho^{(\text{ALG})}$, then there exists a fixed lower bound on the achievable distortion $D$, regardless of the SNR.
\end{enumerate} 
\end{prop}
\begin{proof}
This result follows immediately from the definition of the sampling rate-distortion function and the definition of the stability threshold in \eqref{eq:stability_threshold_def}.
\end{proof}

Starting with the infinite SNR limits given in Section~\ref{sec:rho_SNR}, it is straightforward to show that the stability thresholds of the recovery algorithms studied in this paper are given by
\begin{align}
\varrho^{(\text{ML})} &= \kappa \\
\varrho^{(\text{MF})} &= \infty \\
\varrho^{(\text{LMMSE})}&= 1 \\
\varrho^{(\text{AMP-ST})} &= \mathcal{M}_0(\alpha,\kappa) \\
\varrho^{(\text{AMP-MMSE})} &= \sup_{\tau >0} \frac{\mmse(\tau;p_X)}{\tau}
\end{align}
where $M_0(\alpha,\kappa)$ is given by Eq. \eqref{eq:M_zero} in Appendix~\ref{sec:ST_behavior}. 

The ML stability threshold corresponds to the well known fact that $m = k+1$ random linear projections are, with probability one, sufficient to recover an arbitrary $k$-sparse vector. The LMMSE stability threshold corresponds to the fact that $m = n$ linearly independent projections are sufficient to recover an arbitrary vector of length $n$. The AMP-ST stability threshold, which depends only on the sparsity rate of the distribution $p_X$, has been studied previously in \cite{DonMalMon10} where it is shown that $\min_\alpha\mathcal{M}_0(\alpha,\kappa)$ corresponds to the $\ell_1/\ell_0$ equivalence threshold of Donoho and Tanner \cite{DT09}. The AMP-MMSE threshold has, to the best of our knowledge, not been studied previously. 

Starting with Proposition~\ref{prop:opt_high_SNR}, it can also be shown that the stability threshold of the optimal recovery algorithm is lower bounded by 
\begin{align}
\varrho^{(\text{Lower Bnd.})} = \omega_c
\end{align}
for any distribution $p_X$ for which the strict inequality in \eqref{eq:D_Theta_LB} holds with $D = 0$. In many cases, this lower bound is equal to the sparsity rate $\kappa$.

Finally, using the analysis of the MMSE bound provided in Appendix~\ref{sec:behavior_MMSE}, it can be shown that the stability threshold of the MMSE estimator, as predicted by the replica method, is given by
\begin{align}
\varrho^{(\text{MMSE})} &=  \lim_{ \tau \rw 0} \frac{\mmse(\tau;p_X)}{\tau} \label{eq:stability_threshold_mmse}
\end{align}
when the limit exists. The right hand side of \eqref{eq:stability_threshold_mmse} is referred to as the {\em MMSE dimension} of the distribution $p_X$ by the authors in \cite{WuVer10}, and it is equal to the weight on the continuous part of $p_X$ whenever $p_X$ is a purely continuous-discrete mixture. 

In Fig.~\ref{fig:rho_D_inf_Gaussian}, the stability thresholds corresponding to the Bernoulli-Gaussian distribution correspond to the zero distortion limit (i.e.~the intersection with the $y$-axis).

%%%%%%%%%%%%%%%%%%%%%
\subsection{Distortion versus Sampling Rate}

We now turn our attention to the tradeoff between the achievable distortion and the sampling rate. We begin with a precise characterization of the low-distortion behavior. 

\begin{prop}\label{prop:rho_D}
The low-distortion behavior corresponding to a fixed pair $(\snr,p_X)$ is given by
\begin{align}
\lim_{D \rw 0} \bigg( \frac{P(D;p_X)}{ \mathcal{H}(D;\kappa)} \bigg)\, \rho^{(\text{ML-UB})} &= \Big(\frac{2}{3-\sqrt{8}}\Big) \frac{1}{\snr}  \label{eq:low_D_ML}\\
\lim_{D \rw 0} \sigma^2_\text{awgn}(D,p_X) \, \rho^{(\text{MF})} &= \frac{1}{\snr} + 1 \label{eq:low_D_MF}\\
\lim_{D \rw 0} \sigma^2_\text{awgn}(D,p_X) \, \rho^{(\text{ALG})} &= \frac{1}{\snr} \label{eq:low_D_alg}
\end{align}
where \eqref{eq:low_D_alg} holds for the LMMSE, AMP-MMSE, AMP-ST, and MMSE recovery algorithms.
\end{prop}
\begin{IEEEproof}
The limits corresponding to the ML and MMSE estimators are proved in Appendices~\ref{sec:behavior_ML} and \ref{sec:behavior_MMSE} respectively. The limits corresponding to the linear estimators follow immediately from the fact that $\sigma^2_\text{awgn}(D,p_X) \rw 0$ as $D \rw 0$. For the AMP-ST estimator, we use the additional fact that  the noise sensitivity $\mathcal{M}(\sigma^2,\alpha,p_X)$ is bounded (see Appendix~\ref{sec:ST_behavior}) and hence $\sigma^2 \mathcal{M}(\sigma^2,\alpha,p_X) \rw 0$ as $\sigma^2 \rw 0$. For the AMP-MMSE estimator, we use the bound
\begin{align}
\Big|  \sigma^2(D;p_X) \rho^{(\text{AMP-MMSE})} - \frac{1}{\snr}\Big | \nonumber \\
\le \sigma^2(D;p_X) \sup_{\tau > 0} \left\{\frac{ \mmse(\tau;p_X)}{\tau}\right\} \label{eq:AMP-MMSE_low_D}
\end{align}
and note that the right hand side of \eqref{eq:AMP-MMSE_low_D} becomes arbitrarily small as $D\rw 0$. 
\end{IEEEproof}

\begin{figure*}[htbp]
\centering
\epsfig{file=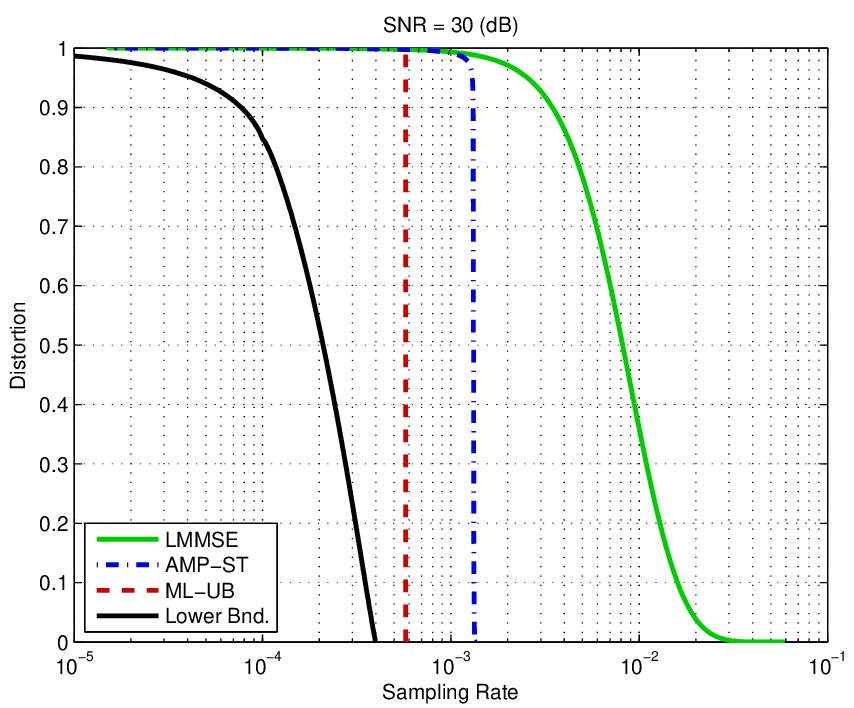, width = .45\textwidth}
\epsfig{file=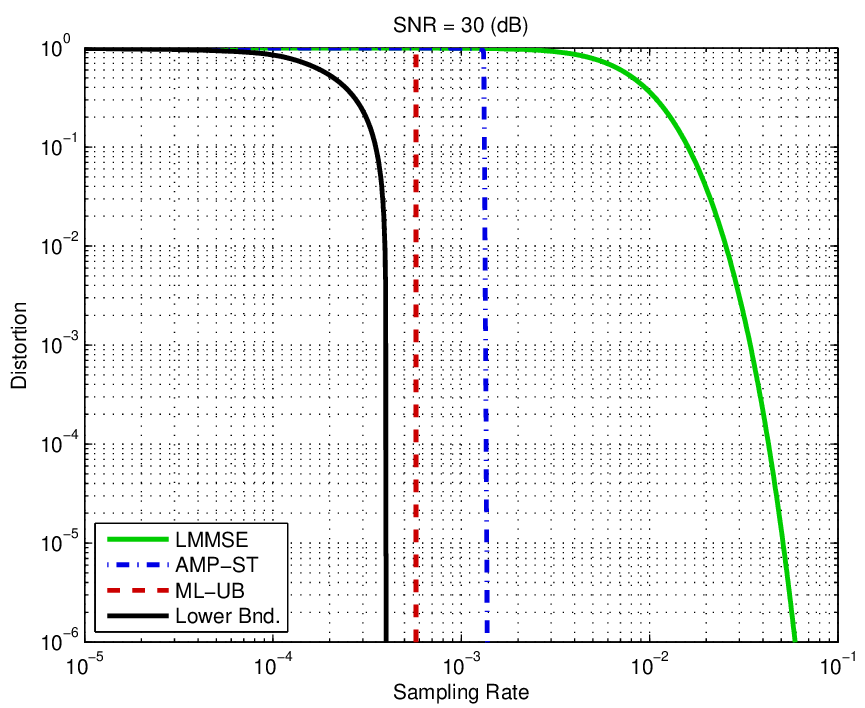, width =.45\textwidth}
\caption{Bounds on the achievable distortion $D$ as a function of the sampling rate $\rho$ when the nonzero entries are lower bounded in squared magnitude by $20\%$ of their average power, but are otherwise arbitrary and the sparsity rate is $\kappa = 10^{-4}$. The MF bound is comparable to LMMSE bound and is not shown.} 
\label{fig:D_rho_Bounded}
\end{figure*}

\begin{figure*}[htbp]
\centering
\epsfig{file=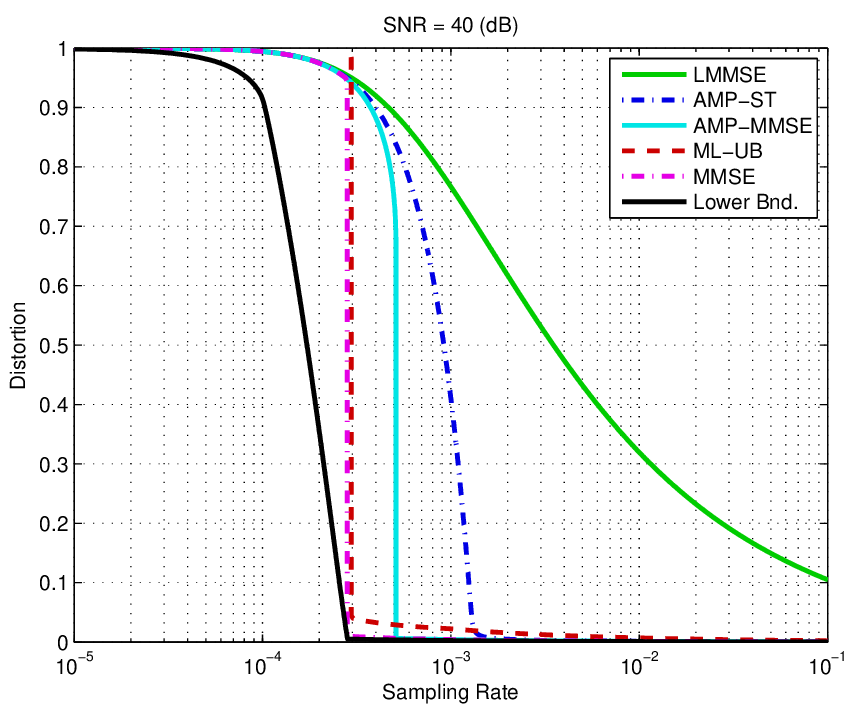, width = .45\textwidth}
\epsfig{file=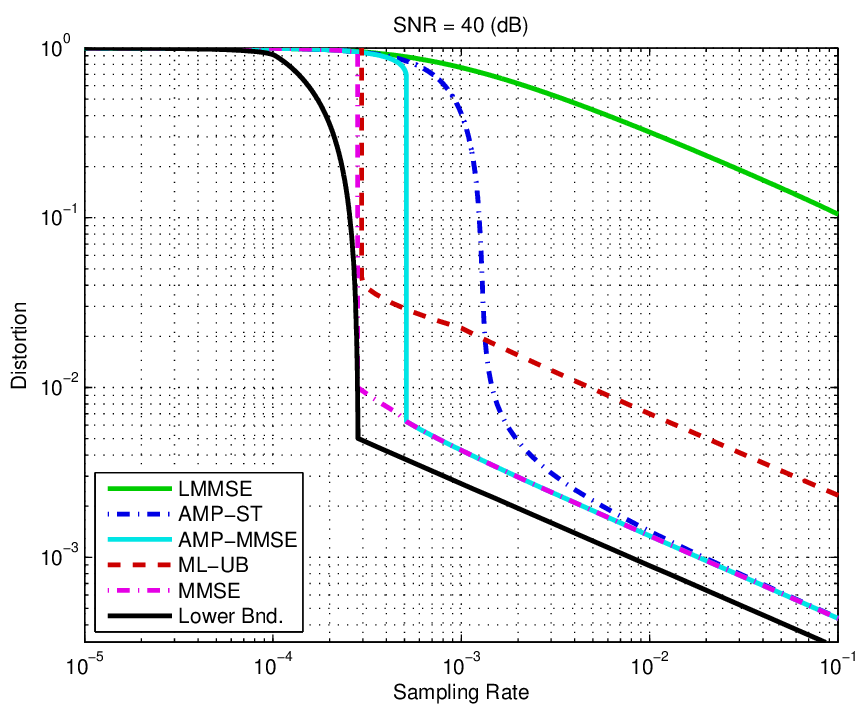, width =.45\textwidth}
\caption{Bounds on the achievable distortion $D$ as a function of the sampling rate $\rho$ when the nonzero entries are i.i.d.~zero-mean Gaussian and the sparsity rate is $\kappa = 10^{-4}$. The MF bound is comparable to LMMSE bound and is not shown.} 
\label{fig:D_rho_Gaussian}
\end{figure*}

In words, Proposition~\ref{prop:rho_D} says that as the desired distortion $D$ becomes small, the ML upper bound is inversely proportional to the ratio $P(D;p_X)/\mathcal{H}(D;\kappa)$ whereas the low distortion behavior of the remaining bounds is inversely proportional to the function $\sigma^2_\text{awgn}(D;p_X)$. The behavior of these terms is characterized for the bounded and polynomial decay signal classes in the following results. 

\begin{prop}[Bounded]\label{prop:bounded} If $p_X \in \mathcal{P}_\text{Bounded}(\kappa,B)$, then
\begin{align}
 \frac{P(D;p_X)}{\mathcal{H}(D;\kappa)} & \ge \frac{ B^2 }{2 [ \log\big(1/D) + 1 + \log\big({\textstyle \frac{1-\kappa}{\kappa}}\big) ] 
} \label{eq:P_D_bounded}
\end{align}
and
\begin{align}
\sigma^2_\text{awgn}(D;p_X)  & \ge  \frac{B^2}{8[\log(  1/D) + \log(\frac{1-\kappa}{\kappa} )]}\label{eq:sig2_D_bounded}. 
\end{align}
\end{prop}

\begin{prop}[Polynomial Decay]\label{prop:poly_decay}  If $p_X \in \mathcal{P}_\text{Poly.}(\kappa,L,\tau)$, then
\begin{align}
\lim_{D\rw 0} \bigg( \frac{\log(1/D)}{D^{2/L}} \bigg) \frac{P(D;p_X)}{\mathcal{H}(D;\kappa)}&=  \frac{ \tau^{-2/L}}{2(1+2/L)} \label{eq:P_D_poly_decay}
\end{align}
and
\begin{align}
\lim_{D \rw 0}  \bigg( \frac{\log(1/D)}{D^{2/L}} \bigg) \sigma^2_\text{awgn}(D;p_X)& =  \frac{\tau^{-2/L}}{2 }.\label{eq:sig2_D_poly_decay}
\end{align}
\end{prop}

The proofs of Propositions~\ref{prop:bounded} and \ref{prop:poly_decay} are given in Appendices~\ref{sec:prop:bounded-proof} and \ref{sec:prop:poly_decay-proof} respectively. 

One way to interpret these results is to think of the achievable distortion as a function of the sampling rate $\rho$. For a given tuple $(\rho,p_X,\snr)$ and recovery algorithm ALG, we use $D^{(\text{ALG})}$ to denote the smallest achievable distortion, i.e.
\begin{align}
D^{(\text{ALG})} = \inf\{ D \ge 0 \,:\, \text{$D$ is achievable}\}.
\end{align} 
An upper bound on the rate at which $D^{(\text{ALG})}$ decreases as the sampling rate becomes large is given in the following result, which is an immediate consequence of Propositions~\ref{prop:rho_D}, \ref{prop:bounded}, and \ref{prop:poly_decay}. 

\begin{prop}\label{prop:D_rho}
Consider a fixed pair $(\snr,p_X)$, and let ALG denote one of the ML, MF, LMMSE, AMP-MMSE, AMP-ST, MMSE recovery algorithms. 
\begin{enumerate}[(a)]
\item If $p_X \in \mathcal{P}_\text{Bounded}(\kappa,B)$ then there exists a constant $C$ such that
\begin{align}
D^{(\text{ALG})} \le \exp( - C \, \rho)
\end{align}
for all sampling rates $\rho > 0$. 
\item If $p_X \in \mathcal{P}_\text{Poly.}(\kappa,L,\tau)$ 
then there exists a constant $C$ such that
\begin{align}
\Big(\frac{1}{D^{(\text{ALG})}}\Big)^{2/L} \log\Big(\frac{1}{D^{(\text{ALG})}}\Big) \le C\, \rho
\end{align}
for all sampling rates $\rho > 0$. 
\end{enumerate}
\end{prop}

Proposition~\ref{prop:D_rho} shows that the low-distortion behavior depends critically on the behavior of the distribution $p_X$ around the point $x=0$. If the nonzero part of the distribution is bounded away from zero, then the distortion decays exponentially rapidly with the sampling rate. Conversely, if the nonzero part of $p_X$ has a polynomial decay rate $L>0$, then the distortion decays polynomially rapidly with the sampling rate, with an exponent that converges to $L/2$. 

In \cite{RG-Lower-Bounds}, it is shown that the scaling behavior in Proposition~\ref{prop:D_rho} is optimal in the sense that, up to constants, no recovery algorithm can do any better. Consequently, each of the algorithms presented in this paper is optimal in a scaling sense as the SNR becomes large whenever the sampling rate is strictly greater than the stability threshold. 

The behavior of the achievable distortion $D$ as a function of the sampling rate $\rho$ is illustrated in Fig.~\ref{fig:D_rho_Bounded} for the class of bounded distributions $\mathcal{P}_\text{Bounded}(\kappa,B)$ with $B = \sqrt{0.2/\kappa}$. In accordance with part (a) of Proposition~\ref{prop:D_rho}, the LMMSE bound decays exponentially rapidly as a function the sampling rate. The same scaling behavior also occurs for the ML and AMP-ST bounds as well as the lower bound from \cite{RG-Lower-Bounds}. However, due to the relatively large SNR, this behavior occurs only for for distortions much less than $10^{-6}$ and is therefore not visible in the range of distortions plotted in Fig.~\ref{fig:D_rho_Bounded}.

For comparison, the same behavior is illustrated in Fig.~\ref{fig:D_rho_Gaussian} for a Bernoulli-Gaussian distribution which has decay rate $L = 1$. In accordance with part (b) of Proposition~\ref{prop:D_rho}, the distortion decays polynomially with rate $1/2$. Interestingly, the AMP-MMSE and AMP-ST bounds converge to the MMSE bound, and are within a constant factor $\approx 1.18$ of the lower bound. This behavior shows that these algorithms are near-optimal when the sampling rate is relatively large. We suspect that the gap between these algorithms and the ML upper bound is due primarily to looseness in our bounding technique.

%%%%%%%%%%%%%%%%%
\subsection{Distortion versus SNR}

The previous section showed that computationally efficient algorithms can be near-optimal when the sampling rate is large. In the context of compressed sensing, a more interesting question is whether or not these same algorithms can be near-optimal when the sampling rate is fixed, and much less than one. In this section, we show that the answer to this question is `yes', provided that the sampling rate is strictly greater than the stability threshold of the algorithm. 

For a given tuple $(D,\rho,p_X)$ and recovery algorithm ALG, we let $\snr^{(\text{ALG})}$ denote the infimum over all $\snr\ge 0$ such that $D$ is achievable, i.e.
\begin{align*}
\snr^{(\text{ALG})} = \inf\{ \snr \ge 0 \, : \, \text{$D$ is achievable} \}. 
\end{align*}
The following result characterizes the low-distortion behavior with respect to the SNR. 

\begin{prop}\label{prop:snr_D}
The low-distortion behavior corresponding to a fixed pair $(\rho,p_X)$ is given by
\begin{align}
\lim_{D \rw 0} \bigg( \frac{P(D;p_X)}{\mathcal{H}(D;\kappa)}\bigg) \, \snr^{(\text{ML-UB})} &= \Big(\frac{2}{3-\sqrt{8}}\Big) \frac{1}{\rho -\kappa}  \label{eq:low_D_snr_ML}
\end{align}
if $\rho > \kappa$,
\begin{align}
\lim_{D \rw 0} \sigma^2_\text{awgn}(D,p_X) \, \snr^{(\text{AMP-MMSE})} &= \frac{1}{\rho - \varrho^{(\text{MMSE})}} \label{eq:low_D_snr_algb}
\end{align}
if $\rho > \varrho^{(\text{AMP-MMSE})}$, and 
\begin{align}
\lim_{D \rw 0} \sigma^2_\text{awgn}(D,p_X) \, \snr^{(\text{ALG})} &= \frac{1}{\rho - \varrho^{(\text{ALG})}} \label{eq:low_D_snr_alg}
\end{align}
if $\rho > \varrho^{(\text{ALG})}$ where \eqref{eq:low_D_snr_alg} holds for the LMMSE, AMP-ST, and MMSE recovery algorithms.
\end{prop}
\begin{IEEEproof}
The limits corresponding to the ML and MMSE recovery algorithms are proved in Appendices~\ref{sec:behavior_ML} and \ref{sec:behavior_MMSE} respectively. The limits corresponding to the  LMMSE and AMP recovery algorithms follow straightforwardly along the same lines as the proof of Proposition~\ref{prop:rho_D}.
\end{IEEEproof}

Proposition~\ref{prop:snr_D} is analogous to Proposition~\ref{prop:rho_D} except that it is valid only if the sampling rate $\rho$ exceeds the stability threshold. The reason that Proposition~\ref{prop:snr_D} does not provide a bound for the MF estimator is that the stability threshold of the MF estimator is infinite, and thus the corresponding limit in \eqref{eq:low_D_snr_alg} is not defined. 

Combining Proposition~\ref{prop:snr_D} with Propositions~\ref{prop:bounded}, and \ref{prop:poly_decay} leads to the following result, which bounds the rate at which $D^{(\text{ALG})}$ decreases as the SNR becomes large. 

\begin{prop}\label{prop:D_snr}
Consider a fixed pair $(\rho,p_X)$, and let ALG denote one of the ML, LMMSE, AMP-MMSE, AMP-ST, MMSE recovery algorithms. 
\begin{enumerate}[(a)]
\item If $p_X \in \mathcal{P}_\text{Bounded}(\kappa,B)$ and $\rho > \varrho^{(\text{ALG})}$, then there exists a constant $C$ such that
\begin{align}
D^{(\text{ALG})} \le  \exp( - C \, \snr)
\end{align}
for all $\snr > 0$. 
\item  If $p_X \in \mathcal{P}_\text{Poly.}(\kappa,L,\tau)$ and $\rho > \varrho^{(\text{ALG})}$, then there exists a constant $C$ such that
\begin{align}
\Big(\frac{1}{D^{(\text{ALG})}}\Big)^{2/L} \log\Big(\frac{1}{D^{(\text{ALG})}}\Big) \le C\, \snr
\end{align}
for all $\snr > 0$. 
\end{enumerate}
\end{prop}

\begin{figure*}[htbp]
\centering
\psfrag{Sampling Rate = 0.0013}[lb]{\small $\quad \rho/\kappa= 13$} 
\psfrag{Sampling Rate = 0.004}[lb]{\small $\quad \rho/\kappa= 40$} 
\psfrag{Sampling Rate = 0.04}[lb]{\small $\; \rho/\kappa= 400$} 
\epsfig{file=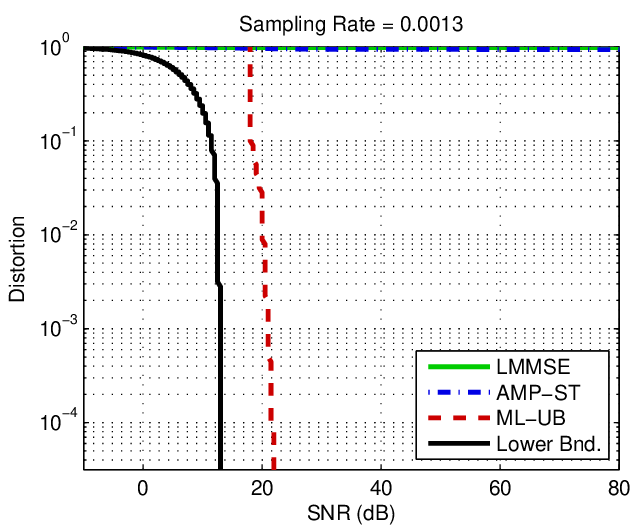, width = .32\textwidth}
\epsfig{file=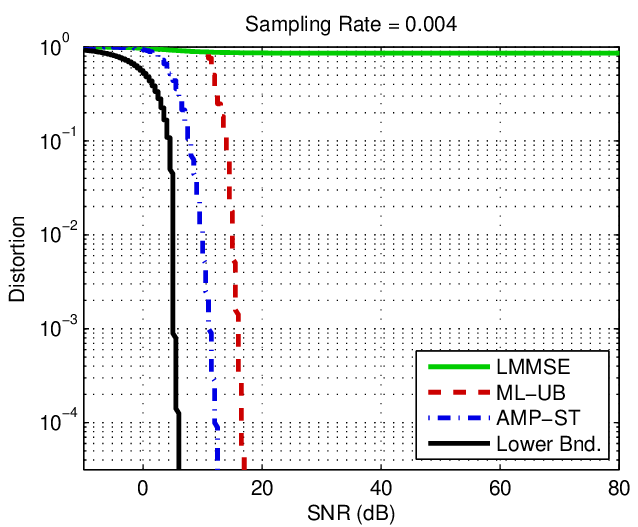, width =.32\textwidth}
\epsfig{file=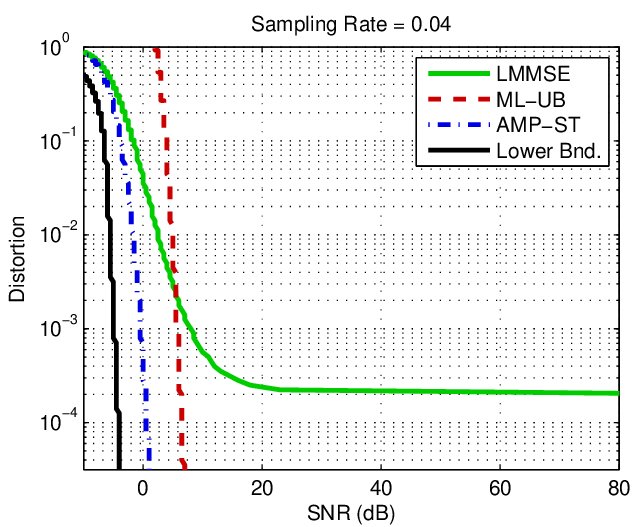, width =.32\textwidth}
\caption{Bounds on the achievable distortion $D$ as a function of the SNR for three different sampling rates $\rho$ when the nonzero entries are lower bounded in squared magnitude by $20\%$ of their average power, but are otherwise arbitrary and the sparsity rate is $\kappa = 10^{-4}$. The MF bounds is comparable to LMMSE bound and is not shown.} 
\label{fig:D_SNR_Bounded}
\end{figure*}

\begin{figure*}[htbp]
\centering
\psfrag{Sampling Rate = 0.0005}[lb]{\small $\quad \rho/\kappa= 5$} 
\psfrag{Sampling Rate = 0.0006}[lb]{\small $\quad \rho/\kappa= 6$} 
\psfrag{Sampling Rate = 0.005}[lb]{\small $\quad \rho/\kappa= 50$} 
\epsfig{file=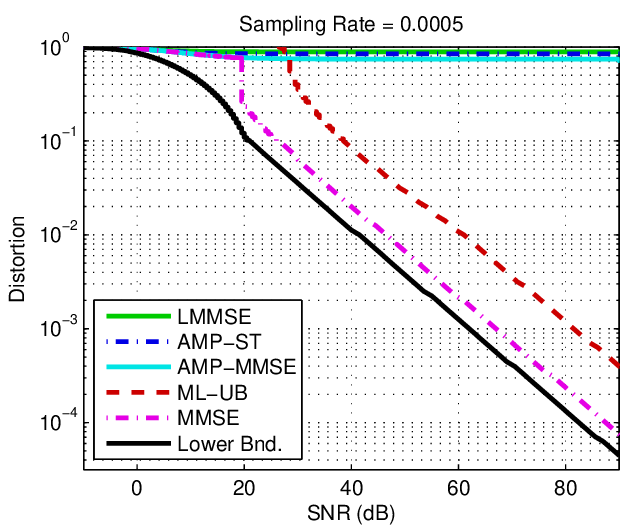, width = .32\textwidth}
\epsfig{file=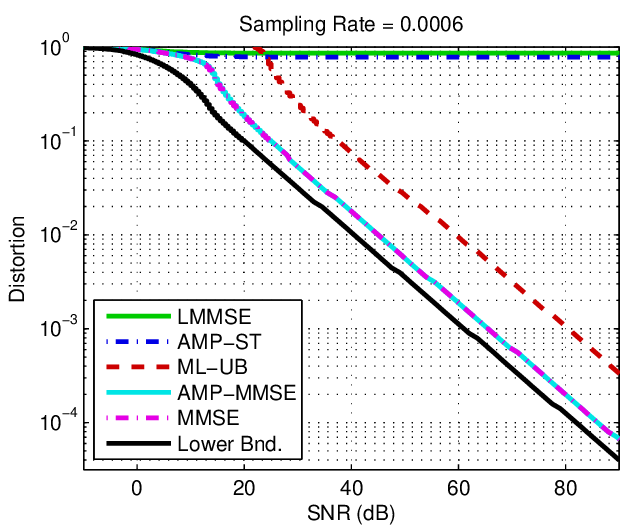, width =.32\textwidth}
\epsfig{file=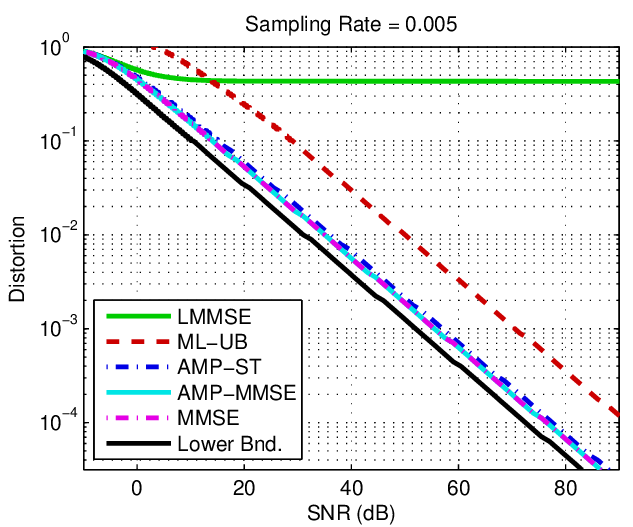, width =.32\textwidth}
\caption{Bounds on the achievable distortion $D$ as a function of the SNR for three different sampling rates $\rho$ when the nonzero entries are i.i.d.~zero-mean Gaussian and the sparsity rate is $\kappa = 10^{-4}$. The MF bound is comparable to LMMSE bound and is not shown.} 
\label{fig:D_SNR_Gaussian}
\end{figure*}

In \cite{RG-Lower-Bounds}, it is shown that the scaling behavior in Proposition~\ref{prop:D_snr} is optimal in the sense that, up to constants, no recovery algorithm can do any better. Consequently, each of the algorithms presented in this paper (except for the MF estimator) is optimal in a scaling sense as the SNR becomes large whenever the sampling rate is strictly greater than the stability threshold. 

The behavior of the achievable distortion $D^{(\text{ALG})}$ as a function of the SNR is illustrated in Fig.~\ref{fig:D_SNR_Bounded} for three different sampling rates $\rho$ and the class of bounded distributions $\mathcal{P}_\text{Bounded}(\kappa,B)$ with $B = \sqrt{0.2/\kappa}$. In the left panel, the sampling rate is greater than $\varrho^{(\text{ML})}$ but less than $\varrho^{(\text{AMP-ST})}$ and $\varrho^{(\text{LMMSE})}$. In accordance with part (a) of Proposition~\ref{prop:D_snr}, the ML distortion decays exponentially rapidly whereas the AMP-ST and LMMSE distortions are bounded away from zero. In the second panel, the sampling rate is greater than $\varrho^{(\text{ML})}$ and $\varrho^{(\text{AMP-ST})}$ but less than $\varrho^{(\text{LMMSE})}$, and hence the AMP-ST distortion also decays exponentially rapidly. In the third panel, $\rho$ is relatively large but still less than $\varrho^{(\text{LMMSE})}$. Thus, even though the LMMSE distortion is less than it was before, it is still bounded away from zero. 

For comparison, the same behavior is illustrated in Fig.~\ref{fig:D_SNR_Gaussian} for a Bernoulli-Gaussian distribution which has decay rate $L = 1$. In accordance with part (b) of Proposition~\ref{prop:D_snr}, the distortion of each algorithm decays polynomially with rate $1/2$ whenever the sampling rate is greater than the stability threshold of the algorithm. It is interesting to note that the relatively small difference in sampling rates between the left and middle panels marks the boundary between the setting where all of the computationally feasible algorithms studied in this paper are highly suboptimal and the setting where the distortion of the computationally feasible AMP-MMSE algorithm, is within a constant factor $\approx 1.75$ of the lower bound.

%%%%%%%%
%\input{sections/extensions_upper}
\section{Extensions}

This section demonstrates two ways in which the bounds given in Section~\ref{sec:main_results} can be extended to a larger class of signal and measurement models. In Section~\ref{sec:relaxed_recovery} we consider recovery when the unknown vector $\bx$ is only approximately sparse and in Section~\ref{sec:rate_sharing} we show how the rate-distortion region can be convexified using ``rate-sharing'' matrices.

%%%%%%%%%%%%%%%%%%%%%%%%%%%%%%%%%%%%%%
\subsection{Approximately Sparse Signal Models}\label{sec:relaxed_recovery}

The problem formulation given in Section~\ref{sec:formulation} assumes that a large fraction of the entries in $\bx$ are exactly equal to zero. More realistically though, it may be the case that many of these entries are only {\em approximately} equal to zero. This may occur, for instance, if a sparse vector is corrupted by a small amount of noise prior to being measured. In these cases, the vector $\bx$ is not, strictly speaking, sparse, but recovery of the locations of the ``significant'' entries is still a meaningful task. 

With these settings in mind, all of the bounds presented in this paper are first proved with respect to a {\em relaxed} sparsity pattern recovery task in which the goal is to recover the locations of the $\lceil \kappa\, n \rceil$ largest entries in $\bx$. To prove achievability for this task, we assume that the weak converge of Assumption S2 holds (specifically the fact that all but a fraction $\kappa$ of the entries in $\bx$ are tending to zero as $n$ becomes large) but do not require the strict sparsity constraint of Assumption S1. 

The relaxed sparsity pattern recovery task is defined as follows. For any vector $\bx$ and sparsity rate $\kappa$, let $\tilde{S}$ be a drawn uniformly at random from all subsets of $[n]$ of size $k=\lceil \kappa \, n \rceil$ obeying
\begin{align}
\min_{ i \in \tilde{S}} |x_i| \ge \max_{ i \in  [n]\backslash \tilde{S}} |x_i|. \label{eq:S_tilde}
\end{align}
The set $\tilde{S}$ corresponds to the $k$ largest entries in $\bx$ and is uniquely defined whenever the $k$'th largest entry of $\bx$ is unique. For any distortion $D$ and recovery algorithm ALG we define the relaxed sparsity pattern recovery error probability
\begin{align}
\tilde{\varepsilon}^{(\text{ALG})}_n(D) = \Pr[ d(\tilde{S},\hat{S}) > D] \label{eq:error_prob_approx}
\end{align}
where the probability is taken with respect to the distribution on $\tilde{S}$, the matrix $\bA$, the noise $\bW$, and any additional randomness in the recovery algorithm. The definition of achievability with respect to the error probability $\tilde{\varepsilon}_n(D)$ is exactly the same as for the error probability $\varepsilon_n(D)$ except that the strict sparsity of Assumption S1 is not required. 

To obtain the results given in Section~\ref{sec:main_results}, we use the following result which shows that under the additional constraint of Assumption S1, achievability of relaxed sparsity pattern recovery implies achievability of the sparsity pattern in the strict sense. 

\begin{lemma}\label{prop:relaxed_to_strict}
Under Assumption S1,
\begin{align}
\lim_{n \rw \infty} \big | d(\tilde{S},\hat{S}) - d(S^*,\hat{S}) \big | = 0.
\end{align}
\end{lemma}
\begin{proof} 
Two applications of the triangle inequality gives
\begin{align*}
 | d(\tilde{S},\hat{S}) - d(S^*,\hat{S}) \big |  \le d(\tilde{S},S^*).
\end{align*}
By the definition of $\tilde{S}$, it follows straightforwardly that $d(\tilde{S},S^*) \rw 0$ for any sequence of vectors obeying Assumption S1. 
\end{proof}

%%%%%%%%%%%%%%%%%%%
\subsection{Rate-Sharing Matrices}\label{sec:rate_sharing}

All of the bounds presented in Section~\ref{sec:main_results} assume that the measurement matrix $\bA$ has i.i.d. entires (Assumption M4). A natural question then, is whether relaxing this assumption can lead to better performance. Interestingly, the answer to this question can be `yes'. In this section, we show that certain {\em rate-sharing} matrices can achieve points in the sampling rate-distortion region that are impossible using i.i.d.~matrices. 

The concept of rate-sharing is analogous to the idea of time-sharing in communications and can be summarized as follows. By using an appropriately constructed block-diagonal measurement matrix it is possible to separate the recovery problem into two subproblems, each of which is statistically identical to the original problem. By assigning different sampling rates to each of the subproblems and then combining the resulting sparsity pattern estimates, it is possible to achieve an effective sampling rate-distortion pair $(\rho,D)$ that is a {\em linear combination} of the sampling rate-distortion pairs for each of the subproblems.

{\em Construction of a rate-sharing matrix:} For a fixed pair $(\snr,p_X)$ and recovery algorithm ALG, let $(\rho_1,D_1)$ and $(\rho_2,D_2)$ be two achievable sampling rate-distortion pairs. Let $\{\bA_1(n)\}_{n \ge 1}$ and $\{\bA_2(n)\}_{n \ge 1}$ be sequences of measurement matrices obeying Assumptions M1-M3 that achieve these rates. Then, for any $\lambda \in [0,1]$, a sequence of {\em rate-sharing} matrices is given by
\begin{align}
\bA(n) = 
\begin{bmatrix}
\bA_1(\lceil{\lambda \, n \rceil}) & \mathbf{0} \\
\mathbf{0}  &\bA_2(n-\lceil\lambda \, n \rceil)\\
\end{bmatrix}
\bP(n) \label{es:rate_sharing_matrix}
\end{align}
where $ \mathbf{0}$ denotes a matrix of zeros and $\bP(n)$ is a random matrix distributed uniformly over the set of $n \times n$ permutation matrices. 

{\em Recovery using of a rate-sharing matrix:} For a problem of size $n$, the measurements $\bY$ made using the rate-sharing matrix $\bA$ can be expressed as
\begin{align*} 
\begin{bmatrix} \bY_1 \\ \bY_2 \end{bmatrix} = 
\begin{bmatrix}
\bA_1 & \mathbf{0} \\
\mathbf{0}  &\bA_2\\
\end{bmatrix}
\begin{bmatrix} \tilde{\bX}_1 \\ \tilde{\bX}_2 \end{bmatrix}  + \begin{bmatrix}\bW_1 \\ \bW_2 \end{bmatrix} 
\end{align*}
where $\tilde{\bX} = [\tilde{\bX}_1, \tilde{\bX}_2]^T  = \bP \bx$ corresponds to a random permutation of the entries in $\bx$. To recover the sparsity pattern of $\bx$ from these measurements, the recovery algorithm performs the following two steps:
\begin{enumerate}[(1)]
\item Individually estimate the sparsity patterns of $\tilde{\bX}_1$ and $\tilde{\bX}_2$ assuming a sparsity rate of $\kappa$ for each vector. 
\item Use these estimates to produce an estimate $\hat{S}$ of the sparsity pattern of $\bx$. 
\end{enumerate}
 
\begin{prop}[Rate-Sharing]\label{thm:convex}
For a fixed pair $(\snr,p_X)$ and algorithm ALG, let $(\rho_1,D_1)$ and $(\rho_2,D_2)$ be two achievable sampling rate-distortion pairs. Then, for any parameter $\lambda \in [0,1]$, the sampling rate-distortion pair $(\rho,D)$ given by 
\begin{align}
\rho &= \lambda \rho_1 + (1-\lambda) \rho_2\\
D &= \lambda D_1 + (1-\lambda) D_2
\end{align}
is achievable using the rate-sharing strategy outlined above.
\end{prop}

\begin{proof}
Based on the assumptions on $\bA_1(n)$ and $\bA_2(n)$ and the fact that $$ \| \bA(n)\|^2_F = \| \bA_1(\lceil{\lambda \, n \rceil})\|^2_F + \|\bA_2(n-\lceil\lambda \, n \rceil)\|^2_F,$$
it is straightforward to verify that the sequence of rate-sharing matrices $\{\bA(n)\}_{n\ge1}$  defined by \eqref{es:rate_sharing_matrix} satisfies Assumptions M1-M3 with sampling rate $\rho = \lambda \rho_1 + (1-\lambda) \rho_2.$

The next step is to verify that the distortion $D$ is achievable. Since each permutation $\bP(n)$ is independent of the vector $\bx(n)$, the random sequences $\{ \tilde{\bX}_1(n)\}_{n \ge 1}$ and $\{ \tilde{\bX}_2(n)\}_{n \ge 1}$ obey Assumptions S1-S3 with probability one. Since the pairs $(\rho_1,D_1)$ and $(\rho_2,D_2)$ are achievable, it thus follows that the distortions $D_1$ and $D_2$ are achievable for the individual sparsity pattern estimates made in step (1).

Now, for a given problem of size $n$, let $S_1^*$, $\hat{S}_1$, $S_2^*$, $\hat{S}_2$ denote the true and estimated sparsity patterns corresponding to the vectors $\tilde{\bX}_1$ and $\tilde{\bX}_2$, and let $S^*$ and $\hat{S}$ denote the true and estimated sparsity pattern of $\bx$. As a simple exercise, it can be verified that
\begin{align*}
d(S^*,\hat{S})  \le \lambda_n d(S_1^*,\hat{S}_1) + (1-\lambda_n) d(S_2^*, \hat{S}_2)
\end{align*}
where
$$\lambda_n = \frac{\max(|S^*_1|, |\hat{S}_1|)}{\max(|S^*_1|, |\hat{S}_1|) + \max(|S^*_2|, |\hat{S}_2|)}.$$
Using the arguments outlined above, it can then be verified that $\lambda_n \rw \lambda$ almost surely as $n \rw \infty$, and thus we conclude that the distortion $D = \lambda D_1 + (1-\lambda) D_2$ is achievable. 
\end{proof}

As an immediate consequence of Proposition~\ref{thm:convex}, we have the following result.

\begin{cor}
For a fixed pair $(\snr,p_X)$ and algorithm ALG the sampling rate-distortion function is a convex function of the distortion $D$. 
\end{cor}

By comparing the convexified versions of the achievable bounds in this paper with the lower bounds  developed in \cite{RG-Lower-Bounds} for matrices obeying Assumptions M1-M4, it can be verified that there are cases where rate-sharing (even with a potentially suboptimal recovery algorithm) is strictly better using an i.i.d.~matrix and the optimal recovery algorithm. This difference is most dramatic in the high SNR setting when the sampling rate is relatively small compared to the sparsity rate.

%%%%%%%%%
%\input{sections/discussion_upper}
\section{Discussion}

In this section, we review the main contributions of the paper and discuss various implications of our analysis.

%%%%%%%%%%%%%%%%%%%%%%%%%%%%%%%
\subsection{Fundamental Behavior of Sparsity Pattern Recovery}

The achievable bounds derived in this paper, in conjunction with the information-theoretic lower bounds in \cite{RG-Lower-Bounds}, characterize the fundamental limit of what cannot be recovered in presence of noise. A major technical contribution of this paper is the upper bound on the sampling rate-distortion function for the maximum likelihood estimator (Theorem~\ref{thm:ML}). To our knowledge, this is the only achievable bound in the literature that converges to the noiseless limit as the SNR becomes large and correctly characterizes the high SNR behavior. 

Our bounds show that the tradeoffs between the sampling rate $\rho$, the distortion $D$, and the SNR can be characterized in terms of several key properties of the limiting distribution $p_X$. Roughly speaking, the high-SNR behavior is characterized by the differential entropy of the nonzero part of $p_X$ whereas the low-distoriton behavior is characterized by the behavior of the distribution around the point $x=0$. These dependencies can be summarized as follows:
\begin{itemize}
\item {\em High-SNR Behavior:}  If the nonzero part of $p_X$ has a relatively large differential entropy, then the tradeoff between sampling rate and SNR is given by 
$$\rho \approx \kappa + \frac{C}{\log(\snr)}.$$

\item {\em Low-SNR Behavior:} If the nonzero part of $p_X$ has a polynomial decay $L$, then the tradeoff between sampling rate and distortion is given by 
$$\rho \approx C \cdot \Big(\frac{1}{D}\Big)^{1/L}  \log\Big(\frac{1}{D}\Big), $$
and the tradeoff between SNR and distortion is given by
$$\snr \approx C \cdot \Big(\frac{1}{D}\Big)^{1/L}  \log\Big(\frac{1}{D}\Big) \quad \text{if} \quad \rho > \kappa, $$
where the condition $\rho >\kappa$ is necessary if the nonzero part of $p_X$ has a relatively large differential entropy. Note that $L=0$ if the nonzero part of $p_X$ is bounded away from zero. 
\end{itemize}

The high-SNR behavior of the bounds is illustrated in Figures \ref{fig:example1}, \ref{fig:rho_SNR_Gaussian}, and  \ref{fig:rho_D_inf_Gaussian}. The low-distortion behavior is illustrated in Figures \ref{fig:D_rho_Bounded}, \ref{fig:D_rho_Gaussian}, \ref{fig:D_SNR_Bounded}, and   \ref{fig:D_SNR_Gaussian}.

%%%%%%%%%%%%%%%%%%%%%%%%%%
\subsection{Near-Optimality of Efficient Algorithms}

From a practical standpoint, a key question is whether or not a particular computationally efficient algorithm is near-optimal. A positive answer to this question means that more complicated algorithms are unnecessary. A negative answer, however, suggests that it is worth investing resources in the design and implementation of better algorithms. 

In the absence of measurement noise, the tradeoffs for existing algorithms have been relatively well understood. For example, the number of measurements $m$ needed for exact recovery of a $k$-sparse vector of length $n$ can be summarized as follows: linear recovery (i.e.~solving a system of full rank linear equations) requires $m \ge n$; linear programming requires $m \ge C \cdot k \log(n/k)$ for some constant $C$; and an NP-hard exhaustive search requires $m \ge k+1$. 

One of the contributions of this paper, has been to extend the understanding of these tradeoffs to practically motivated settings where, due to measurement noise, only approximate recovery is possible. Interestingly, our results show that there are problem regimes where existing computationally efficient algorithms---such as the linear estimation or approximate message passing---are near-optimal and other regimes where they are highly suboptimal. 

For example, the dependence of the sampling rate on the SNR illustrated in Fig.~\ref{fig:example1} shows that computationally simple algorithms are near-optimal at low SNR, but suggests that increasing sophistication is required as the SNR increases.

Moreover, the bounds illustrated in Fig.~\ref{fig:D_SNR_Gaussian} show that a small change in the sampling rate can make the crucial difference between whether or not approximate message passing achieves the optimal tradeoff between SNR and distortion.

%%%%%%%%%%%%%%%%%%%%%%%%
\subsection{Comparison with Replica Predictions}

In this paper, we provide a comparison of rigorous bounds with the nonrigourous analysis of the replica method (Theorem~\ref{thm:MMSE}). Since the predictions of the replica method are sharp, they provide valuable insights about where our bounds are tight and where they can be improved. For example, in Fig.~\ref{fig:rho_SNR_Gaussian} there exists a gap between the upper and lower bounds for SNR in the range of 45 to 60 dB. In this region, the replica prediction suggests that the information-theoretic lower bound from \cite{RG-Lower-Bounds} is essentially correct and that the ML upper bound bound is loose. 

An additional contribution of this comparison, is that the relative tightness of our rigorous bounds provides evidence in support of the unproven replica assumptions. For example, in Fig.~\ref{fig:rho_SNR_Gaussian}, the upper and lower bounds are extremely close and sandwich the replica prediction for all SNR greater than 60 dB. Despite a vast amount of work on this topic, such evidence has been notoriously difficult to come by.

%%%%%%%%%%%%%%%%%
\subsection{Universality of Bounds} 

To characterize the limiting behavior of a sequence of vectors we assume convergence of the empirical distributions (Assumption S2). If the limiting distribution is known, it is possible to use optimized recovery algorithms based on the distribution (e.g.~the AMP-MMSE and MMSE recovery algorithms). In many cases, however, the limiting distribution is unknown. To address these settings, we develop bounds for fixed estimators which hold uniformly over a class of limiting distributions such as the class of all distributions bounded away from zero or the class of all distributions with polynomial decay (see Section~\ref{sec:signal_classes}) . 

Our the results show that, in many cases, prior information about the limiting distribution does not help significantly. For example, in the right panel of Fig.~\ref{fig:example1}, the upper and lower bounds on the sampling rate-distortion function are relatively tight, uniformly over the class of distribution bounded away from zero. Another example is given by Propositions \ref{prop:rho_D} and \ref{prop:snr_D} which show that the low distortion behavior depends entirely on certain properties of the underlying distributions (specifically, the behavior of the distribution around the point $x=0$). 

We remark that an important counterexample occurs if the limiting distribution is supported on a finite subset of the real line (see \cite{RG-Lower-Bounds}). Then, the high-SNR sampling rate-distortion behavior can depend crucially on prior information about the distribution.

%%%%%%%%%%%%%%%%%%%
\subsection{Role of Model Assumptions}

This paper focusses on the setting where a constant fraction of the entries are nonzero (Assumption S1). In Section~\ref{sec:relaxed_recovery} it is shown that the results in this paper still hold when all but a fraction $\kappa$ of the entries in $\bx$ are tending to zero as $n$ becomes large. In principle, many of the tools developed in the paper could also be used to address settings where the number of nonzero entries grows sub-linearly with the vector length, and hence there is a vanishing fraction of nonzero entries.

Our use of row normalization (Assumption M3) differs from many related works which use column normalization. The reason for our scaling is that, from a sampling perspective, one way to decrease the effect of noise is to take additional samples (all at a fixed per-measurement SNR). If the column norms of the measurement matrix are constrained, then this is not possible since the per-measurement SNR will necessarily decrease as the number of measurements increases. Since it is assumed throughout that the sampling rate $\rho$ is a fixed constant, all results in this paper can be compared to existing works under an appropriate rescaling of the SNR.

The proofs of our upper bounds rely heavily on the assumption that the measurement matrices have i.i.d.~entries (Assumption M4). The proofs of Theorem~\ref{thm:ML} and \ref{thm:amp_general} further assume that these entries are Gaussian (Assumption M5).  The extent to which these assumptions can be relaxed is an important direction for future research. 

In Section~\ref{sec:rate_sharing} it is shown that rate-sharing matrices (which are not i.i.d.) can convexify the sampling rate-distortion region, thus leading to better performance. This result shows that i.i.d.~matrices are strictly suboptimal in some settings.

%%%%%%%%%%%
\appendices

%%%%%%%%%%%%%
%\input{sections/proof-ML}

%%%%%%%%%%%%%%%%%%%%%%%%%%%%%%%%%
\section{Proof of Theorem \ref{thm:ML}}\label{sec:proof:ML}

Following the discussion in Section~\ref{sec:relaxed_recovery}, we first prove achievability with respect to relaxed sparsity pattern recovery.

\begin{theorem}\label{thm:ML2}
Under Assumptions S2 and M1-M5, the statement of Theorem~\ref{thm:ML} holds with respect to the relaxed sparsity pattern recovery error probability $\tilde{\varepsilon}_n(D)$ defined in \eqref{eq:error_prob_approx}.
\end{theorem}

%Then, Theorem \ref{thm:ML} follows straightforwardly from 

Combining Theorem~\ref{thm:ML2} with  Lemma~\ref{prop:relaxed_to_strict} and the fact that $\rho^{(\text{ML-UB})}$ is a continuous and monotonically decreasing function of $D$ completes the proof of Theorem \ref{thm:ML}.

The remainder of this appendix is dedicated to the proof of Theorem~\ref{thm:ML2}.
%The proof of Theorem~\ref{thm:ML2} is based on a direct analysis of the error probability $\tilde{\varepsilon}_n(D)$ and is given in the following two sections. %The proof follows from a direct analysis of the error probability $\tilde{\varepsilon}_n(D)$. 
 We begin with a general bound for the non-asymptotic setting in Section~\ref{sec:non-asymp-bound} and then extended this bound to the asymptotic setting in Section~\ref{sec:asymp-bound}.

Throughout the proof we use $\mathcal{S}_k^n$ to denote the set of all subsets of $[n]$ of size $k$, and for any set $\bs \subset [n]$, we use $\bs^c$ to denote its complement $[n] \backslash \bs$. 

%%%%%%%%%%%%%%%%%%%%%%%%%%%%%%%%%%
\subsection{A Non-Asymptotic Bound}\label{sec:non-asymp-bound}

Consider the measurement model given in \eqref{eq:estimation_problem} where $\bx \in \mathbb{R}^n$ is an arbitrary vector whose true sparsity is unknown. For a given parameter $\kappa$, let $k=\lceil \kappa \, n \rceil$, let $\tilde{S}$ be a drawn uniformly at random from all subsets of $[n]$ of size $k=\lceil \kappa \, n \rceil$ obeying \eqref{eq:S_tilde}, and let $\hat{S}$ be the output of the ML recovery algorithm.%, and 

Also, for each integer $b \in \{0,1,\cdots,k\}$, define
\begin{align} 
M(b) =  \min_{ \bs \in \mathcal{S}_k^n : |\bs \backslash \tilde{\bs}| = b} {\frac{1}{n}} \| \bx_{\bs^c}\|^2 \cdot \snr  \label{eq:P_nb}
\end{align}
By the definition of $\tilde{S}$, it is straightforward to see that $M(b)$ is defined uniquely by $\bx$ and $b$ (i.e.~it does not depend on the realization of $\tilde{\bs}$).% and is therefore defined unambiguously in terms of $\bx$. 

The following result gives an upper bound on $\tilde{\varepsilon}_n(D)$ that depends only on the distortion $D$, the dimensions $n,m,k$, and the function $M(b)$.% evaluated with  $\{M(b) \,: \, b=0, \lfloor D\kappa +1\rfloor \le b \le k\}$. 

\begin{lemma}\label{lem:E1}
If the entries of the measurement matrix $\bA$ are i.i.d.~$\mathcal{N}(0,1/n)$, then the following bound holds for any distortion $D \in [0,1]$ and integer $k< m$,
\begin{align}
\tilde{\varepsilon}^{(\text{ML})}_n(D) \le   \sum_{b= \lfloor D k +1 \rfloor }^k  \min( \Psi_1(b),\Psi_2(b))\label{eq:lemE1}
\end{align}
where $ \Psi_1(b)$ and  $\Psi_2(b)$ are given by \eqref{eq:Psi_1} and \eqref{eq:Psi_2} below.  \addtocounter{equation}{2}
\end{lemma}
 
% \newcounter{tempequationcounter}
\begin{figure*}[!t]
\normalsize
\setcounter{tempequationcounter}{\value{equation}}
\begin{align}
\setcounter{equation}{94}
\Psi_1(b) & = \min_{\lambda \in [0,1]} \bigg[ 
\Big( 2\log \Big [ {\textstyle \sqrt{\frac{ 1+  \lambda M(b) + (1\!-\!\lambda) M(0) }{1 + M(0)}}-1} \Big] \Big)^ {- { \frac{m-k}{4}}} +  {  {k \choose b}{n -k \choose b}} \Big(\textstyle \frac{1+M(b)}{1+ \lambda M(b) +(1\! -\!\lambda) M(0)} \Big)^ {- { \frac{m-k}{2}}}\bigg]
\label{eq:Psi_1}\\
\Psi_2(b) & = \min_{\substack{\theta,\mu \in (0,1) \\ \epsilon > 0}} {  {k \choose b}{n  - k \choose b}}  \bigg[ \Big(1+{\textstyle \frac{1}{4} (1\!-\!\theta)^2 M(b)}\Big)^{-\frac{m-k}{2}}  +\Big({\textstyle \frac{ \exp(\epsilon)}{2 M(0)} }\Big)^{- \frac{m-k}{2}}  +  \Big({\textstyle \frac{1 + \mu \theta M(b) }{\exp(\epsilon M(0))} }\Big)^{- \frac{m-k}{2}}  (1-\mu^2)^{-\frac{b}{2}}\bigg] \label{eq:Psi_2}
\end{align}
\setcounter{equation}{\value{tempequationcounter}}
\hrulefill
\vspace*{4pt}
\end{figure*}
 
 \begin{proof} For each $S\in \mathcal{S}^n_k$ let $\Pi(\bA_{\bs})$ denote the $m\times m$ orthonormal projection onto the null space of the $m \times k$ matrix $\bA_{\bs}$. If $\bA_{S}$ is full rank (an event that occurs with probability one over the assumed distribution on $\bA$) then this projection is given by
\begin{align}
\Pi(\bA_{\bs}) = I_{m \times m} - \bA_\bs (\bA^T_\bs \bA_\bs)^{-1} \bA^T_\bs.
\end{align}
Since
\begin{align}
%\min_{\bu \in \mathbb{R}^n\; :\;  \bu_{\bS^c} = 0  }  \| \bA \bu - \bY\| %\\
%& = 
\min_{\bu_\bs \in \mathbb{R}^{k}} \| \bA_{\bs} \bu_\bs - \bY\|
& =  \| \Pi(\bA_{\bs})\bY\|,
\end{align}
it can be easily verified that the ML estimate of size $k$ is an element of the set %distributed uniformly at random over the set
\begin{align}
\arg \min_{\bs \in \mathcal{S}^n_k} \| \Pi(\bA_{\bs})\bY\|. \label{eq:ML_argmin_set}
\end{align}

Now, for each integer $b \in \{0,1,2,\cdots,k\}$, define the event
\begin{align}
\mathcal{E}(b) = \Big\{ \min_{\bs \in B_{b}} \| \Pi(\bA_{\bs})\bY\| > \| \Pi(\bA_{\tilde{\bs}})\bY\| \Big\} \label{eq:def_Eb}
\end{align}
where $B_b = \{ \bs \in \mathcal{S}^n_k : |\bs \backslash \tilde{\bs}| = b\}$. In words, the event $\mathcal{E}(b)$ guarantees that the set of minimizers in \eqref{eq:ML_argmin_set} will not contain any set $S$ for which $d(S,\tilde{S}) = b/k$. Thus, a sufficient condition for recovery is given by the event $\bigcap_{\lfloor D k +1 \rfloor}^k \mathcal{E}(b)$, and by the union bound we have
\begin{align}
%\tilde{p}(D) &\le  \Pr\Big [ \bigcup_{\lfloor D k +1 \rfloor}^k \mathcal{E}(b)\Big] \nonumber \\
\tilde{\varepsilon}^{(\text{ML})}_n(D) & \le \sum_{b= \lfloor Dk +1 \rfloor }^k \Pr[ \mathcal{E}^c(b)]
\label{eq:PeUBc}
\end{align}
where $\mathcal{E}^c(b)$ denotes the complement $\mathcal{E}(b)$. 

The bounds $\Pr[\mathcal{E}^c(b)] \le \Psi_1(b)$ and $\Pr[\mathcal{E}^c(b)] \le \Psi_2(b)$ are proved in Sections~\ref{sec:Psi1-proof} and \ref{sec:Psi2-proof} respectively. 
\end{proof}
 
\begin{remark}Lemma~\ref{lem:E1} is general in the sense that it makes no assumptions about the sparsity of $\bx$ or the size of $\tilde{S}$. Therefore, it can be used to address a variety of recovery tasks such as recovering a subset or superset of the true support.
\end{remark}

\begin{remark}
If $M(0) < M(\lfloor D k +1 \rfloor)$, then the upper bound decreases exponentially rapidly with $m$, i.e.~there exists a constant $C$ such that $\tilde{\varepsilon}^{(\text{ML})}_n(D) \le \exp(- C \, m )$.
\end{remark}

%%%%%%%%%%%%%%%%%%%%%%%%%%%%%%%%%%
\subsection{The Asymptotic Setting}\label{sec:asymp-bound}

We now prove Theorem~\ref{thm:ML2} by applying Lemma~\ref{lem:E1} to a sequence of problems obeying Assumptions S2 and M1-M5. For each problem of size $n$ let $k_n = \lceil \kappa n \rceil$. Beginning with \eqref{eq:lemE1}, we have
\begin{align}
\tilde{\varepsilon}_{n}(D) &\le  \sum_{b= \lceil D k_n \rceil}^{k_n}  \min(\Psi_1(b),\Psi_2(b))\\
&\le n \max_{\lceil D k_n \rceil \le b \le k_n}  \min(\Psi_1(b),\Psi_2(b)) \\
&= n \exp\big(-n\min_{\beta \in [D,1]}\;  \psi_n(\beta)  \big). \label{eq:p_n_upperbound}
\end{align}
where $\psi_n(\beta) = - {\textstyle \frac{1}{n} }\min_{i \in \{1,2\}} \log \Psi_i(\lceil \beta k_n \rceil)$. To study the asymptotic behavior of this bound we need the following lemma. The proof is given in Section~\ref{sec:lem-psi_bound-proof}.

\begin{figure*}[!t]
\normalsize
\setcounter{tempequationcounter}{\value{equation}}
\begin{align}
\setcounter{equation}{105}
\psi_1(\beta) & =  \max_{\lambda \in [0,1]} \min \bigg\{ \frac{\rho - \kappa}{4} \big ( { \textstyle\sqrt{ 1+ \lambda P(\beta) \snr}-1} \big)^2 ,\; \frac{\rho\! -\! \kappa}{2} \Big[ \log \Big(\frac{1+P(\beta)\snr}{1+\lambda P(\beta)\snr} \Big) - \frac{(1\!-\!\lambda) P(\beta) \snr}{1+ P(\beta) \snr} \Big]  - \mathcal{H}(\beta;\kappa) \bigg\}
\label{eq:psi_1}\\
\psi_2(\beta) & = \max_{ \theta,\mu \in [0,1]} \min \bigg\{  \frac{\rho - \kappa}{2} \log\big(1+{\textstyle \frac{1}{4}  P(\beta)\snr} \big), \; \Big[  \frac{\rho - \kappa}{2} \log\big(1+\theta \mu P(\beta) \snr \big) + {\frac{\beta \kappa}{2} \log\big(1-\mu^2 \big) } \Big] \bigg\} - \mathcal{H}(\beta;\kappa)
 \label{eq:psi_2}
\end{align}

\setcounter{equation}{\value{tempequationcounter}}
\hrulefill
\vspace*{4pt}
\end{figure*}

\begin{lemma}\label{lem:psi_bound}
Under Assumption S2, the sequence of functions $\{\psi_n(\beta)\}_{n \ge 1}$ is uniformly asymptotically lower bounded in the following sense
\begin{align}
\limsup_{n \rw \infty} \max_{\beta \in [0,1]} \Big(  \psi(\beta) - \psi_n(\beta) \Big) \le 0. \label{eq:lempsi_bound} 
\end{align}
where $\psi(\beta) = \max_{ i \in \{1,2\}} \psi_i(\beta)$ and $\psi_1(\beta)$ and $\psi_2(\beta)$ are given by \eqref{eq:psi_1} and \eqref{eq:psi_2} below. 
\addtocounter{equation}{2}
\end{lemma}

\begin{remark}
Under the additional constraint of Assumption S3, the bound \eqref{eq:lempsi_bound} holds with respect to the absolute difference  $|\psi(\beta)-\psi_n(\beta)|$. For the proof of Theorem~\ref{thm:ML2}, however, only the lower bound is needed. 
\end{remark}

Returning to \eqref{eq:p_n_upperbound}, we can now write
\begin{align}
\liminf_{n \rw \infty} - {\ \frac{1}{n}} \log \tilde{\varepsilon}_n(D)
& \ge   \liminf_{n \rw \infty}   \min_{\beta \in [D,1]}  \psi_n(\beta) \nonumber \\
& \ge \min_{\beta \in [D,1]} \psi(\beta). \label{eq:p_n_upperbound2}
\end{align}
where the swapping of the limit and the minimum in \eqref{eq:p_n_upperbound2} is justified by Lemma~\ref{lem:psi_bound}. 

With a bit of algebra, it can be verified that
\begin{align}
\kappa + \Lambda(\beta;p_X,\snr) =  \inf\big\{ \rho :  \psi(\beta) > 0\big\},
\end{align}
and thus
\begin{align}
\rho^{(\text{ML-UB})} &= \inf\big\{ \rho : \min_{\beta \in [D,1]} \psi(\beta) > 0\big\}. 
\end{align}
Since $\psi(\beta)$ is a continuous and monotonically increasing function of $\rho$, it follows that the right hand side of \eqref{eq:p_n_upperbound2} is strictly positive for any $\rho > \rho^{(\text{ML})} $. This concludes the proof of Theorem~\ref{thm:ML2}.

%%%%%%%%%%%%%%%%%%%%%%%%
%%%%%%%%%%%%%%%%%%%%%%%%
\subsection{Proof of the bound $\Psi_1(b)$ in Lemma~\ref{lem:E1}}\label{sec:Psi1-proof}
We begin with a bounding technique used previously by Wainwright \cite{Wainwright_InfoLimits_IEEE09} for the study of exact sparsity pattern recovery. For notational simplicity, we define the random variable 
\begin{align}
Z_{\bs} =  \snr\, \| \Pi(\bA_{\bs}) \bY \|^2 \label{eq:defZ_s}
\end{align}
which corresponds to the distance between the samples $\bY$ and subspace spanned by $\bA_{\bs}$.

For any $t \in \mathbb{R}$, we can write 
\begin{align}
\Pr[\mathcal{E}^c(b)] & = \Pr[\mathcal{E}^c(b) \cap \{Z_{\tilde{\bs}} > t\}] + \Pr[\mathcal{E}^c(b) \cap \{Z_{\tilde{\bs}} \le t\}]  \nonumber \\
& \le  \Pr[Z_{\tilde{\bs}} > t] + \Pr[\mathcal{E}^c(b) \cap \{Z_{\tilde{\bs}} \le t\}]. \label{eq:PEb_first}
\end{align}
Furthermore,
\begin{align}
\Pr[\mathcal{E}^c(b) \cap \{Z_{\tilde{\bs}} \le t\}]  & = \Pr\big[ \{ \min_{\bs \in B_b} Z_{\bz} \le Z_{\tilde{\bs}} \} \cap \{Z_{\tilde{\bs}} \le t\}\big] \nonumber \\
& \le  \Pr\big[ \min_{\bs \in B_b} Z_{\bs} \le t \big] \nonumber  \\
&\le  \sum_{\bs \in B_b} \Pr[ Z_{\bs} \le t],\label{eq:PEb_first_d}
\end{align}
where \eqref{eq:PEb_first_d} follows from the union bound. Plugging \eqref{eq:PEb_first_d} back into \eqref{eq:PEb_first} gives\begin{align}
\Pr[\mathcal{E}^c(b)]  \le \Pr[Z_{\tilde{\bs}} > t] + \sum_{\bs \in B_b} \Pr[ Z_{\bs} \le t] . \label{eq:eaB1c}
\end{align}

Note that $\Pr[Z_{\bs} \le t ]$ depends only on the marginal distributions of the random variable $Z_{\bs}$. In Wainwright's analysis \cite{Wainwright_InfoLimits_IEEE09}, this probability is upper bounded in terms of a noncentral chi-squared random variable whose noncentrality parameter is unknown but bounded. In this proof however, we begin with the exact distribution on $Z_{\bs}$.

%%%%%%%%%%%%%%%%%%%%
\begin{lemma}\label{lem:chisquareML}
For each $\bs \in \mathcal{S}^n_k$, the random variable 
\begin{align*}
%\frac{1}{
\frac{Z_{\bs}}{1+{\textstyle \frac{1}{n}} \|\bx_{\bs^c}\|^2 \snr}
\end{align*}
has a chi-squared distribution with $m-k$ degrees of freedom.
\end{lemma}
\begin{proof}
Since $\bA_{\bs} \bx_{\bs} $ lies in the range space of $\bA_\bs$, we can write
\begin{align}
Z_{\bs} %&  =  \snr\, \| \Pi(\bA_{\bs}) \bY\|^2 \nonumber \\
%& = \| \Pi(\bA_{\bs})( \bA_{\bs^*} \bx_{\bs^*} + \bW) \|^2 \nonumber \\
& = \| \Pi(\bA_{\bs})(\sqrt{\snr} \bA \bx + \bW) \|^2 \nonumber \\
& = \| \Pi(\bA_{\bs})(\sqrt{\snr} \bA_{\bs^c} \bx_{\bs^c} + \bW) \|^2 .\nonumber%\label{eq:lemchisquare1}
\end{align}
The vector $\sqrt{\snr} \bA_{\bs} \bx_{\bs}  + \bW$ is independent of $\Pi(\bA_{\bs})$ and has i.i.d.~Gaussian entries with mean zero and variance $1 + \frac{1}{n}\| \bx_{\bs^c}\|^2\,\snr$. Also, with probability one over the distribution $\bA$, the matrix $\Pi(\bA_{\bs})$ has exactly $n-k$ singular values equal to 1 and $k$ singular values equal to 0. Therefore, the stated result follows immediately from the rotational invariance of the Gaussian distribution. 
\end{proof}

To proceed, let $V$ denote a chi-squared random variable with $m-k$ degrees of freedom and let $t=(1+\bar{M}) (m-k)$ where $\bar{M} = \lambda M(b)  + (1\!-\!\lambda) M(0)$ for some $\lambda \in (0,1)$. Then, by Lemma~\ref{lem:chisquareML}, 
\begin{align}
\Pr[Z_{\tilde{\bs}} > t] &= \Pr\Big [\Big( \frac{1}{m-k}\Big) V > \frac{1+\bar{M}}{1 + M(0)}  \Big ]\label{eq:Psi11_proof} %\nonumber \\
%& \le \Psi_1(b,\lambda) \label{eq:Psi1_proof}
\end{align}
and
\begin{align}
\Pr[Z_{\bs} \le t] &= \Pr\Big [\Big( \frac{1}{m-k} \Big)V \le \frac{1+\bar{M}}{1 + \frac{1}{n} \|\bx_{\bs}\|^2 \snr} \Big ] \\
& \le \Pr\Big [ \Big({ \frac{1}{m-k}}\Big) V  \le \frac{1+\bar{M}}{1 + M(b)} \Big ] \label{eq:Psi12_proof}%\\
%& \le \Psi_2(b,\lambda) \label{eq:Psi2_proofb}
\end{align}
where \eqref{eq:Psi12_proof} follows from the definition of $M(b)$. 

Both \eqref{eq:Psi11_proof} and \eqref{eq:Psi12_proof} can be upper bounded using the chi-squared large deviations bounds given in Lemma~\ref{lem:tailbounds_intext} below. Combining these bounds with \eqref{eq:eaB1c} and the simple fact that
\begin{align}
|B_b| = { k \choose b}{n-k \choose b} \label{eq:size_Bb},
\end{align}
shows that $\Pr[\mathcal{E}^c(b)] \le \Psi_1(b)$, which completes the proof.

\begin{lemma}\label{lem:tailbounds_intext}
Let $V$ be a chi-squared random variable with $d$ degrees of freedom. For any $x >1$,\begin{align}
\Pr[V\ge d\, x\big] &\le \exp \big(- d\, \textstyle \frac{1}{4}(\sqrt{2x-1}-1)^2 \big)\label{eq:lowertailbound},\\
\Pr[ V\le d/x] &\le \exp\!\big(-d \textstyle \frac{1}{2} [ \log x  + 1/x -1] \big)
\label{eq:uppertailbound}
\end{align}
\end{lemma}
\begin{proof}
The upper bound \eqref{eq:lowertailbound} follows directly from Laurent and Massart \cite[pp. 1325]{LaurentMassart_stat98}. 

For the lower bound \eqref{eq:uppertailbound}, observe that for any $\mu> 0$,
\begin{align}
&\Pr[V\le ({\textstyle \frac{1}{x}})d] \nonumber\\
& = \Pr[ \exp(-\mu V)  \ge \exp(-\mu ({\textstyle \frac{1}{x}}) d)\big] \nonumber \\
 &\le  \bE[ \exp( - \mu  X )\exp( \mu({\textstyle \frac{1}{x}}) d) \label{eq:chiLBb}]\\
 & = \exp(-d [ \textstyle \frac{1}{2} \log(1+2\mu) - \mu({\textstyle \frac{1}{x}})] ).\label{eq:chiLBc}
\end{align}
where \eqref{eq:chiLBb} follows from Markov's inequality and \eqref{eq:chiLBc} follows from the moment generating function of a chi-squared distribution. Letting $\mu = (x-1) /2$ completes the proof. 
\end{proof}

%%%%%%%%%%%%%%%%%%%%%%%%%%%%%%%%%%%%%%%%%%%
%%%%%%%%%%%%%%%%%%%%%%%%%%%%%%%%%%%%%%%%%%%
\subsection{Proof of the bound $\Psi_2(b)$ in Lemma~\ref{lem:E1}}\label{sec:Psi2-proof}

This proof uses a new decomposition of the error event to obtain a different upper bound on $\Pr[\mathcal{E}^c(b)]$. In some problem regimes, this bound improves significantly over the bound derived in the previous section. As before, we use the definition of $Z_{\bs}$ given in \eqref{eq:defZ_s}.

To motivate our proof strategy, observe that one weakness of the bound \eqref{eq:eaB1c} is that the threshold $t$ is a fixed constant whereas the event $\mathcal{E}(b)$ depends on the {\em relative} magnitudes of the variables $Z_{\bS}$. 

In this proof, we begin with the union bound as follows
\begin{align}
\Pr[\mathcal{E}^c(b)] & \le \sum_{\bs \in B_b} \Pr[ Z_{\bs} \le Z_{\tilde{\bs}}]. \label{eq:Eb_bound2}
\end{align}
Unlike \eqref{eq:eaB1c}, each probability on the right hand side of \eqref{eq:Eb_bound2} depends on the relative magnitudes of $Z_S$ and $Z_{\tilde{S}}$. In the remainder of this proof, our goal is to derive an upper bound on $\Pr[Z_{\bs}\le Z_{\tilde{\bs}}]$ that exploits the dependence between $Z_{\tilde{\bs}}$ and $Z_{\bs}$. A key step is the following lemma. 

\begin{lemma}\label{lem:TUV}
For any $\bs \in \mathcal{S}_k^n$, define the random variables
\begin{align*}
T_{\bs} &= \sqrt{\snr}\,   \| \Pi(\bA_{\bs})\bA_{\bs^c} \bx_{\bs^c} \|\\
U_\bs &= \frac{ \langle \Pi(\bA_{\bs})\bA_{\bs^c } \bx_{\bs^c}, \bW \rangle}{ \| \Pi(\bA_{\bs})\bA_{\bs^c } \bx_{\bs^c } \|} \\
V_\bs &= \| \Pi(\bA_{\bs})\bW \|.
\end{align*}
The following statements hold:
\begin{enumerate}[(a)]
\item $ \displaystyle Z_{\bs} = T_\bs^2 +2 T_\bs U_\bs + V^2_\bs$
\item $T^2_{\bs}/ (\frac{1}{n}\| \bx_{\bs^c}\|^2\, \snr)$ has a chi-squared distribution with $m-k$ degrees of freedom. 
\item $U_{\bs}$ is independent of $T_{\bs}$ and has a Gaussian distribution with mean zero and variance one.
\item $V_{\bs}$ is independent of $T_{\bs'}$ for any $\bs,\bs' \in \mathcal{S}_k^n$. .
\end{enumerate}
\end{lemma}
\begin{proof} To prove part (a) observe that
\begin{align}
Z_{\bs}
& = \| \Pi(\bA_{\bs})( \sqrt{\snr} \bA_{\bs^c} \bx_{\bs^c} + \bW) \|^2 \label{eq:Z_nullspace}\\
& =  \snr\, \| \Pi(\bA_{\bs})\bA_{\bs^c} \bx_{\bs^c} \|^2 +  \| \Pi(\bA_{\bs})\bW \|^2  \nonumber \\\
& \quad + 2  \sqrt{\snr} \,\langle \Pi(\bA_{\bs})\bA_{\bs^c} \bx_{\bs^c},  \Pi(\bA_{\bs})\bW \rangle \nonumber
\end{align}
Part (b) follows from the proof of  Lemma~\ref{lem:chisquareML}.  Part (c) follows from the fact that the vector $\Pi(\bA_{\bs}) \bA_{\bs^c} \bx_{\bs^c}$ is independent of $\bW$ and is nonzero with probability one. Part (d) follows from the fact that $\Pi(\bA_{\bs})$, $\bA_{\bs^c} \bx_{\bs^c}$, and $\bW$ are mutually independent and $\Pi(\bA_{\bs})$ has rank $m-k$ with probability one. 
\end{proof}

To proceed, fix any $\theta \in (0,1)$ and $\epsilon > 0$ and define the event $\mathcal{A} = \cap_{i=1}^3 \mathcal{A}_i$ where
\begin{align}
\mathcal{A}_1 &= \big\{ T^2_{\bs}  + 2T_{\tilde{\bs}} U_{\bs} \ge \theta T^2_{\bs} \big\}\\
\mathcal{A}_2 &= \big\{  T^2_{\tilde{\bs}}  + 2T_{\tilde{\bs}}U_{\tilde{\bs}} \le \epsilon(m-k) \big\}\\
\mathcal{A}_3 &= \big\{ \theta T_\bs^2 + V_\bs^2 - V_{\tilde{\bs}}^2 > \epsilon (m-k)\}.
\end{align}
Using part (a) of Lemma~\ref{lem:TUV} it can be verified that $ \{Z_{\bs} \le Z_{\tilde{\bs}} \} \cap \mathcal{A} = \{\emptyset\}$, and thus
\begin{align}
\Pr[ Z_{\bs} \le  Z_{\tilde{\bs}}]  &= 
\Pr[ \{Z_{\bs} \le  Z_{\tilde{\bs}}\} \cap \mathcal{A}^c] \nonumber \\
& \le  \sum_{i=1}^3 \Pr[ \mathcal{A}_i^c] \label{eq:PboundA1}
\end{align}
where \eqref{eq:PboundA1} follows from the union bound. In the following three subsections, we prove upper bounds on the probabilites $\Pr[\mathcal{A}_j^c]$, $j \in \{1,2,3\}$. Plugging these bounds back into \eqref{eq:PboundA1} and using the fact that the cardinality of $B_b$ is given by \eqref{eq:size_Bb} completes the proof.

%%%%%%%%%%%%%%%%%%%%%%%%%%%%%%%%%%
\subsubsection{Upper Bound on $\Pr[\mathcal{A}_1^c]$ } The first error event is relatively straightforward to bound. Observe that
\begin{align}
\Pr[\mathcal{A}_1^c]% & = \Pr[ (1\!-\!\theta) T^2_\bs  +2 T_\bs U_\bs < 0]\\
& = \Pr[ \textstyle \frac{(1\!-\!\theta)^2}{4} T^2_\bs  +\frac{(1\!-\!\theta)}{2}T_\bs U_\bs < 0] \nonumber \\
& = \Pr[ \exp(- \textstyle \frac{(1\!-\!\theta)^2}{4} T^2_\bs  -\frac{(1\!-\!\theta)}{2}T_\bs U_\bs ) \ge 1] \nonumber\\
& \le \bE[ \exp(- \textstyle \frac{(1\!-\!\theta)^2}{4} T^2_\bs  -\frac{(1\!-\!\theta)}{2}T_\bs U_\bs )] \label{eq:UBAc1_a}\\
& =  \bE[\exp(- \textstyle \frac{(1\!-\!\theta)^2}{8} T^2_\bs  )]\label{eq:UBAc1_b} \\
&= \big(1 - \textstyle \frac{(1\!-\!\theta)^2}{4} \frac{1}{n} \|\bx_{\bs^c}\|^2\,\snr  \big)^{-(m-k)/2} \label{eq:UBAc1_c} \\
&\le \big(1 - \textstyle \frac{(1\!-\!\theta)^2}{4}M(b) \big)^{-(m-k)/2} \label{eq:UBAc1_d}
%& = \Psi_{21}(b;\theta) \nonumber
\end{align}
where: \eqref{eq:UBAc1_a} follows from Markov's inequality; \eqref{eq:UBAc1_b} follows from part (c) of Lemma~\ref{lem:TUV} and the moment generating function of the Gaussian distribution; \eqref{eq:UBAc1_c} follows from part (b) of Lemma~\ref{lem:TUV} and the moment generating function of the chi-squared distribution; and \eqref{eq:UBAc1_d} follows from the definition of $M(b)$.

%%%%%%%%%%%%%%%%%%%%%%%%%%%%%%%%%%
\subsubsection{Upper Bound on $\Pr[\mathcal{A}_2^c]$} The second error event is similar to the first one, except that the inequality is in the other direction and there is a constant term to deal with. If we let $t = \epsilon M(0) (m-k)$ and $\lambda  = 1/(2 M(0))$, then
\begin{align}
\Pr[\mathcal{A}_2^c] & = \Pr[\lambda(-t + T_{\tilde{\bs}}^2 + 2T_{\tilde{\bs}}U_{\tilde{\bs}}) > 0] \nonumber \\
& = \Pr[ \exp( -\lambda t + \lambda T_{\tilde{\bs}}^2 + 2\lambda T_{\tilde{\bs}}U_{\tilde{\bs}} ) > 1 \nonumber ]\\
& \le \bE[ \exp( - \lambda t + \lambda T_{\tilde{\bs}}^2 + 2 \lambda T_{\tilde{\bs}}U_{\tilde{\bs}} )] \label{eq:UBAc2_a}\\
%& = \bE[ \exp(- \lambda t +  \lambda T_{\bs^*}^2 - 2 \lambda^2 T^2_{\bs^*} )]\\
& = \bE[ \exp(- \lambda t +  (\lambda-2\lambda^2) T_{\tilde{\bs}}^2)] \label{eq:UBAc2_b} \\
& = \exp(- \lambda t) \big(1- 2(\lambda - 2\lambda^2)M(0)\big)^{-\frac{m-k}{2}} \label{eq:UBAc2_c}\\
& = \Big( \frac{\exp(\epsilon)}{2 M(0)} \Big)^{- \frac{m-k}{2}} \label{eq:UBAc2_d}
\end{align}
where: \eqref{eq:UBAc2_a} follows from Markov's inequality; \eqref{eq:UBAc2_b} follows from part (c) of Lemma~\ref{lem:TUV} and the moment generating function of the Gaussian distribution; \eqref{eq:UBAc2_c} follows from part (b) of Lemma~\ref{lem:TUV} and the moment generating function of the chi-squared distribution; and \eqref{eq:UBAc1_d} follows from plugging in the definitions of $t$ and $\lambda$.

%%%%%%%%%%%%%%%%%%%%%%%%%%%%%%%%%%
\subsubsection{Upper Bound on $\Pr[\mathcal{A}_3^c]$} The third error event requires the most work. Part of the difficulty is that the random variables $V^2_{\bs}$ and $V^2_{\tilde{\bs}}$ are not independent. The following result uses the fact that they are positively correlated to obtain a nontrivial upper bound on the moment generating function of their difference; the proof is given in Section~\ref{sec:lem-V_s-proof}.

\begin{lemma}\label{lem:V_s}
For any $\mu \in (0,1)$,
\begin{align} 
\bE[ \exp( \textstyle \frac{\mu}{2}[ V^2_{\tilde{\bs}} - V^2_{\bs}])] \le (1-\mu^2)^{-b/2} \label{eq:lem:V_s}
\end{align}
\end{lemma}

We remark that the exponent in \eqref{eq:lem:V_s} is proportional to the overlap $b$. If $V^2_{\bs}$ and $V^2_{\tilde{\bs}}$ were independent, then the exponent would be proportional to $k$. This difference in the exponents is the key reason why this bounding technique works well in settings where the previous technique failed. 

With Lemma~\ref{lem:V_s} in hand, we are now ready to upper bound the probability $\Pr[\mathcal{A}_3^c]$. Let $t = \epsilon P_n(0) (m-k)$ and fix any $\mu \in (0,1)$. Then,
\begin{align}
\Pr[\mathcal{A}_3^c] &= \Pr[{\textstyle \frac{\mu}{2}}(t - \theta T_{\bs}^2 - V_{\bs}^2 + V_{\tilde{\bs}}^2) \ge 0] \nonumber \\
&= \Pr[ \exp({\textstyle \frac{\mu}{2}}(t - \theta T_{\bs}^2 - V_{\bs}^2 + V_{\tilde{\bs}}^2)) \ge 1] \nonumber \\
&\le \bE[ \exp({\textstyle \frac{\mu}{2}}(t - \theta T_{\bs}^2 - V_{\bs}^2 + V_{\tilde{\bs}}^2)) ]\label{eq:UBAc3_a}\\
%&\le  \exp({\textstyle \frac{\mu}{2}}t)
&=  \bE[\exp( {\textstyle \frac{\mu}{2}}[t- \theta T_{\bs}^2])] \bE[\exp( {\textstyle \frac{\mu}{2}} [ V_{\tilde{\bs}}^2 - V_{\bs}^2]) ]\label{eq:UBAc3_b}\\
&=  e^{\frac{\mu}{2} t} \big(1 + \mu \theta \|\bx_{\bs^c}\|^2\,\snr \big)^{- \frac{m-k}{2}}  (1-\mu^2)^{-\frac{b}{2}} \label{eq:UBAc3_c}\\
&\le  e^{\frac{\mu}{2} t} \big(1 + \mu \theta M(b) \big)^{- \frac{m-k}{2}}  (1-\mu^2)^{-\frac{b}{2}} \label{eq:UBAc3_d}\\
&=  \Big(\frac{1 + \mu \theta M(b) }{\exp(\epsilon P_n(0))}\Big)^{- \frac{m-k}{2}}  (1-\mu^2)^{-\frac{b}{2}}\label{eq:UBAc3_e}
%\bE[\exp( {\textstyle \frac{\mu}{2}} [ V_{\bs^*}^2 - V_{\bs}^2]) ]\\
%&=  e^{\frac{\mu}{2} t} \big(\frac{1 + \mu \theta \|\bx_{\bs^c}\|^2}{\exp(\epsilon)} \big)^{- \frac{m-k}{2}} \bE[\exp( {\textstyle \frac{\mu}{2}} [ V_{\bs^*}^2 - V_{\bs}^2]) ]
\end{align}
where: \eqref{eq:UBAc3_a} follows from Markov's inequality; \eqref{eq:UBAc3_b} follows from part (d) of Lemma~\ref{lem:TUV};  \eqref{eq:UBAc3_c} follows from part (b) of Lemma~\ref{lem:TUV}, the moment generating function of the chi-squared distribution, and Lemma~\ref{lem:V_s};  \eqref{eq:UBAc3_d} follows from the definition of $M(b)$; and  \eqref{eq:UBAc3_e} follows from the definition of $t$.

%%%%%%%%%%%%%%%%%%%%%%%%%%%%%%%%%%
\subsection{Proof of Lemma~\ref{lem:psi_bound}}\label{sec:lem-psi_bound-proof}
 
To simplify notation we will write $k$ instead of $k_n$ where the dependence on $n$ is implicit. 

Since $\Psi_1(b)$ and $\Psi_2(b)$ are non-increasing functions of $M(b)$, it is sufficient to show that the following limits hold:
%For this proof, it is sufficient to show that the following limits hold:
%It is sufficient to prove that the following limits
\begin{align}
&\lim_{n \rw \infty} \sup_{\beta \in [0,1]} \bigg| \mathcal{H}(\beta,\kappa) - \frac{1}{n} \log {  {k \choose \lceil \beta k \rceil}{n\! -\!k\choose  \lceil \beta k \rceil}}   \bigg| =0 \label{eq:lim_N}\\
& \lim_{n \rw \infty} M(0) = 0 \label{eq:lim_M}\\
&\limsup_{n \rw \infty} \max_{\beta \in[0,1]} \Big( P(\beta)\,\snr - M(\lceil \beta k \rceil) \Big) < 0. \label{eq:Mb_LB}
\end{align}
Then, % since $\Psi_1(b)$ and $\Psi_2(b)$ are non-increasing functions of $M(b)$,
it follows immediately that
\begin{align}x
\limsup_{n \rw \infty} \max_{\beta \in [0,1]} \bigg( \psi_i(\beta) + { \frac{1}{n}} \log( \Psi_i(\lceil \beta k_n \rceil) \bigg) < 0
\end{align}
for $i \in \{1,2\}$, which proves the desired result.

To begin, note that \eqref{eq:lim_N} follows directly from a strong form of Stirling's approximation  \cite[Lemma 17.5.1]{ElementsofIT}.

Next, we consider the term $M(0)$.  For each problem of size $n$, let $\{n_i\}_{i \in [n]}$ be a permutation of $[n]$ such that $x^2_{n_1} \le x^2_{n_2} \le \cdots \le x_{n_n}^2$. Starting with the definition of $\tilde{\bs}$, we can write %$M(0)$ as
\begin{align}
\snr^{-1}M(0) &= \min_{\bs  \in \mathcal{S}_k^n} \frac{1}{n} \|\bx_{\bs^c}\|^2\\
& = \frac{1}{n} \sum_{i=1}^{n-k} x_{n_i}^2 \\
%& = \int_0^\infty \frac{1}{n} \sum_{i=1}^{n-k} \one( x_{n_i}^2 > u) du\\
& = \int_0^\infty\Big(\frac{n-k}{n}-  \frac{1}{n} \sum_{i=1}^{n-k} \one(x_{n_i}^2\le u)\Big ) du\\
%\end{align}
%\begin{align}
%& = \int_0^\infty\Big( 1- \textstyle \frac{k}{n}  -   \min\big( F_n(u), 1- \textstyle \frac{k}{n}) \Big ) du\\
& = \int_0^\infty \max\big(1- G_n(u) - \textstyle \frac{k}{n}  ,0\big) du,
\end{align}
where $G_n(u)$ 
%For each $n$, let
%\begin{align}
%G_n(u) = \frac{1}{n} \sum_{i=1}^n \one(x_i^2 \le u)
%\end{align}
denotes the empirical distribution function of $\{x^2_i\}_{i \in [n]}$, 
Thus, for any $\epsilon > 0$, 
\begin{align}
\snr^{-1}M(0) & = \int_0^\epsilon  \max\big(1- G_n(u) - \textstyle \frac{k}{n}  ,0\big) du \nonumber \\
& \quad +\int_\epsilon^\infty  \max\big(1- G_n(u) - \textstyle \frac{k}{n}  ,0\big) du \nonumber \\
& \le \epsilon + \max\big(1- G_n(\epsilon) - \textstyle \frac{k}{n}  ,0\big). \label{eq:M0_bound}
\end{align} 

By the weak convergence of Assumption S2, it follows that the second term on the right hand side of \eqref{eq:M0_bound} converges to zero as $n \rw \infty$. Since epsilon is arbitrary, we conclude that $\lim_{n \rw \infty} M(0) = 0$.
 
We now consider the final term $M(b)$. Since
\begin{align*}
 n\,  \snr^{-1} M(b)& = \|\bx\|^2 - \max_{\bs \in B_b} \|\bx_{\bs}\|^2\\
 %& = \|\bx\|^2 -  \max_{\bs \in B_b} \big( \| \bx_{ \bs \cap \tilde{\bs}}\|^2 + \|\bx_{\bs\backslash \tilde{\bs}} \|^2 \big) \\
 %& \ge \|\bx\|^2 -  \max_{\bs \in B_b} \big( \| \bx_{ \bs \cap \tilde{\bs}}\|^2 + \|\bx_{\bs\backslash \tilde{\bs}} \|^2  + \| \bx_{\bs^c \cap \tilde{\bs}^c}\|^2 \big) \\
  & \ge \|\bx\|^2 -  \max_{\bs \in B_b} \big( \| \bx_{ \bs}\|^2+ \| \bx_{\bs^c \cap \tilde{\bs}^c}\|^2 \big) \\
 & = \|\bx\|^2 -  \max_{\bs \in B_b} \| \bx_{ \bs \cap \tilde{\bs}}\|^2 - \|\bx_{\tilde{\bs}^c} \|^2\\
& = \min_{\bs \in \mathcal{S}^n_{k-b}} \|\bx_{\bs^c}\|^2 - n\, \snr^{-1} M(0),
\end{align*}
it is sufficient to show that
\begin{align}
\limsup_{n \rw \infty} \max_{\beta \in[0,1]} \big( P(\beta) - P_n(\beta) \big) < 0. \label{eq:Mb_LB2}
\end{align}
where $P_n(\beta) = \frac{1}{n}\min_{\bs \in \mathcal{S}^n_{k-b}} \|\bx_{\bs^c}\|^2 $.

Following the same steps we used for $M(0)$, we have
\begin{align*}
P_n(\beta)& = \int_0^\infty \max\big(1- G_n(u) - \textstyle \frac{k-\lceil \beta k \rceil}{n}  ,0\big) du.
\end{align*}
Also, by definition
\begin{align}
P(\beta) = \int_0^\infty \max\big(1- G(u) - (1-\beta)\kappa  ,0\big) du
\end{align}
where $G(u) = \Pr[X^2 \le u]$. Thus, for any $\epsilon >0$ we have
\begin{align}
&P(\beta) - P_n(\beta) \nonumber \\%\\
%  &= \int_0^\infty \max\big(1- F(u) - (1-\beta)\kappa  ,0\big) du\\
%& \quad - \int_0^\infty \max\big(1- F_n(u) - \textstyle \frac{k-\lceil \beta k \rceil}{n}  ,0\big) du\\
& = \int_0^\epsilon \max\big(1- G(u) - (1-\beta)\kappa  ,0\big) du + \int_0^\infty \varphi_n(u) du \nonumber \\
& \le \epsilon + \int_0^\infty \max( \varphi_n(u) ,0) du \label{eq:Pdiff}
\end{align}
where
\begin{align*}
\varphi_n(u) &=  \Big [  \max\big(1- G(u+\epsilon) - (1-\beta)\kappa  ,0\big)  \\
& \quad - \max\big(1- G_n(u) - \textstyle \frac{k-\lceil \beta k \rceil}{n}  ,0\big) \Big] 
\end{align*}
To deal with the second term in \eqref{eq:Pdiff}, observe that
\begin{align}
\varphi_n(u)  \le \big|(1-\beta) \kappa -\textstyle \frac{k-\lceil \beta k \rceil}{n}\big | +  G_n(u) - G(u+\epsilon).
\end{align}
Thus, by the weak convergence of Assumption S2,
\begin{align}
\lim_{n \rw \infty} \max_{\beta \in [0,1]} \max(\varphi_n(u),0) = 0
\end{align}
for every $u \in \mathbb{R}$. Since $\varphi_n(u)$ is upper bounded by the integrable function $1-G(u+\epsilon)$, it follows from the bounded convergence theorem that the second term in \eqref{eq:Pdiff} is equal to zero. Since $\epsilon$ is arbitrary, this proves \eqref{eq:Mb_LB2} and hence \eqref{eq:Mb_LB}.

%%%%%%%%%%%%%%%%%%%%%%%%%%%%%%%%%%
\subsection{Proof of Lemma~\ref{lem:V_s}} \label{sec:lem-V_s-proof}
The key to this proof is to consider the eigenvalues of the matrix $M = \Pi(\bA_{\tilde{\bs}}) - \Pi(\bA_\bs)$. Since $M$ is symmetric, it can be expressed as $M = Q \Lambda Q^T$ where $Q$ is an $m \times m $ orthonormal matrix and $\Lambda$ is a real valued diagonal matrix whose diagonal entries obey $\lambda_1 \ge \lambda_2 \ge ... \ge \lambda_m$.  Letting $\tilde{\bW} = Q^T \bW$, we have
\begin{align}
V^2_{\tilde{\bs}} - V_{\bs}^2% &= \| \Pi(\bA_{\tilde{\bs}})\bW \|^2 -  \| \Pi(\bA_{\bs})\bW \|^2\\
%& = \bW^T \Pi(\bA_{\bs^*}) \bW - \bW^T \Pi(\bA_{\bs}) \bW \\
& = \bW^T M \bW^T
%& = \tilde{\bW}^T \Lambda \tilde{\bW}\\
= \sum_{i=1}^m \lambda_i \tilde{W}_i^2
\end{align}
where $\tilde{W}_1,\tilde{W}_2,\cdots,\tilde{W}_m$ are i.i.d.~Gaussian $\mathcal{N}(0,1)$, and thus
%Gaussian random variables with mean zero and variance one. Thus, by the independence of $\tilde{W}_i$, it follows that
\begin{align}
\bE[ \exp( \textstyle \frac{\mu}{2}[ V^2_{\tilde{\bs}} - V^2_{\bs}])] 
%\bE[ \exp( \textstyle \frac{\mu}{2} V_\bs)]  
&=  \prod_{i=1}^m \bE[ \exp( \textstyle \frac{\mu}{2} \lambda_i \tilde{W}_i^2)]\\
& = \prod_{i=1}^m (1 - \mu \lambda_i)^{-1/2}
. \label{eq:boundVa} 
\end{align}

To characterize the eigenvalues, we now consider two cases. If $m \ge 2k$, then 
\begin{align*}
\text{rank}(M) &= \text{rank}\big([I - \Pi(\bA_{\bs})] - [I - \Pi(\bA_{\tilde{\bs}})]\big)\\
&\le \text{rank}(I - \Pi(\bA_{\bs})) + \text{rank}(I - \Pi(\bA_{\tilde{\bs}}))\\
& \le 2k,
\end{align*}
%$\text{rank}(M) \le 2k$ 
and so at least $m-2k$ eigenvalues are equal to zero. It can be shown (see \cite[p. 8]{PaiWei94}), that the remaining $2k$ singular values are given by $\lambda_i = \sin \theta_i$ and $\lambda_{m-i+1} = - \sin\theta_i$ for $i = 1,2,\cdots,k$ where $\pi/2 \ge \theta_1 \ge \theta_2 \ge \cdots \ge \theta_k \ge 0$ are known as the {\em principal angles} between the $k$-dimensional subspaces $ \mathcal{R}(\bA_{\bs})$ and  $\mathcal{R}(\bA_{\tilde{\bs}})$ spanned by the columns of $\bA_{\bs}$ and $\bA_{\tilde{\bs}}$ respectively. Since the number of principal angles that are equal to zero is given by the dimension of the intersection of the two subspaces, it follows that
\begin{align*}
|\{ i : \theta_i = 0\}| &= \text{dim}\big( \mathcal{R}(\bA_{\bs}) \cap \mathcal{R}(\bA_{\tilde{\bs}}) \big)\\
& \ge  \text{dim}\big( \mathcal{R}(\bA_{\bs \cap \tilde{\bs}})  \big)\\
& = k-b
\end{align*}
where the last equality holds with probability one over the distribution on $\bA$. 

Returning to \eqref{eq:boundVa}, we can now write
%Thus, returning to \eqref{eq:boundVa} we have
\begin{align}
\bE[ \exp( \textstyle \frac{\mu}{2}[ V^2_{\tilde{\bs}} - V^2_{\bs}])] % &=  \prod_{i=1}^b \Big (\bE[ \exp( \textstyle \frac{\mu}{2} \sin \theta_i \tilde{W}_i^2)] \nonumber \\
%& \quad \times 
%\bE[ \exp( \textstyle - \frac{\mu}{2} \sin \theta_{m-i+1} \tilde{W}_{m-i+1}^2)]\Big) \nonumber \\
& = \prod_{i=1}^b (1-   \mu^2 \sin^2 \theta_i)^{-1/2} \label{eq:boundVb}  \\
& \le (1-   \mu^2)^{-b/2} \label{eq:boundVc} 
\end{align}
where %\eqref{eq:boundVb} follows from the moment generating function of the chi-squared distribution and
\eqref{eq:boundVc} follows from the fact that $0\le \sin^2 \theta_i \le 1$.  

For the case $m < 2k$ we use similar arguments. Since
\begin{align*}
\text{rank}(M) %&= \text{rank}\big( \Pi(\bA_{\bs}) - (I - \Pi(\bA_{\bs^*})\big)\\
&\le \text{rank}(\Pi(\bA_{\bs})) + \text{rank}( \Pi(\bA_{\tilde{\bs}}))\\
& \le 2(m-k),
\end{align*}
at least $2k-m$ eigenvalues of $M$ are equal to zero. The remaining $2(m-k)$ eigenvalues are given by $\lambda_i = \sin \theta_i$ and $\lambda_{m-i+1} = - \sin\theta_i$ for $i = 1,2,\cdots,m-k$ where the $\theta_i$ are the principal angles between the $m-k$ dimensional subspaces $\mathcal{N}(\bA_{\bs})$ and  $\mathcal{N}(\bA_{\tilde{\bs}})$ corresponding to the orthogonal complements of $ \mathcal{R}(\bA_{\bs})$ and  $\mathcal{R}(\bA_{\tilde{\bs}})$ respectively. Thus, we have
\begin{align*}
|\{ i : \theta_i = 0\}| &= \text{dim}\big( \mathcal{N}(\bA_{\bs}) \cap \mathcal{N}(\bA_{\tilde{\bs}}) \big)\\
&= m - \text{dim}\big( \mathcal{R}(\bA_{\bs}) + \mathcal{R}(\bA_{\tilde{\bs}}) \big)\\
&\ge \max\big( 0,m - 2k + \text{dim}\big( \mathcal{R}(\bA_{\bs \cap \tilde{\bs}})  \big) \big)\\
& =\max\big( 0, m -k -b)
\end{align*}
where the last equality holds with probability one over the distribution on $\bA$. Therefore, there are at most $b$ nonzero principle angles. Following the same steps used in the previous case, leads again to the upper bound \eqref{eq:boundVc}. This concludes the proof of Lemma~\ref{lem:V_s}.

%%%%%%%%%%%%%%
%\input{sections/proof-MF}
\section{Proofs of Theorems~\ref{thm:MF}, \ref{thm:LMMSE}, \ref{thm:amp_general}, and \ref{thm:MMSE}}
\label{sec:thm_MF_proof}

Each of these proofs follows a similar outline. First, we establish convergence of the empirical joint distribution on the entries in $\bx$ and the vector estimate $\hat{\bX}$ generated in the first stage recovery (see Fig~\ref{fig:two_stage}). Then, we show that this convergence characterizes the limiting distortion with respect to the relaxed sparsity pattern recovery task described in Section~\ref{sec:relaxed_recovery}. 

In these proofs, we use the superscripts $\overset{prob}{\rw}$ and $\overset{dist}{\rw}$ to denote convergence in probability and distribution, respectively.

%%%%%%%%%%%%%%%%%%%%%%%%%%%%%%%%%
\subsection{From Convergence in Distribution to Relaxed Recovery}

For each problem of size $n$, let $\hat{\bX}$ denote the estimate of the unknown vector $\bx$ generated in the first stage of sparsity pattern recovery and let $F_n(x,\hat{x})$ denote the cumulative distribution function (CDF) of the empirical joint distribution on the entries in $(\bx,\hat{\bX})$, i.e.
\begin{align}
F_n(x,\hat{x}) = \frac{1}{n} \sum_{i = 1}^n \one\big(x_i \le x, \hat{X}_i \le \hat{x}\big).
\end{align}
Note that $F_n(x,\hat{x})$ is a random function due to the randomness in $\hat{\bX}$. Also, let $F(x,z)$ denote the CDF of the random pair $(X,Z)$ given in Definition~\ref{def:scalar_model}, i.e.
\begin{align}
F(x,z) = \Pr[ X \le x , Z \le z]
\end{align}

According to standard terminology, $F_n(x,\hat{x})$ converges weakly in probability to the limit $F(x,z)$ if 
\begin{align}
\lim_{n \rw \infty} \Pr\big[ \big|F_n(x,z) - F(x,z)\big| > \epsilon \big] = 0
\end{align}
for any fixed $\epsilon >0$ and $(x,z) \in \mathbb{R}^2$ such that $(x,z)$ are continuity points of $F(x,z)$. Since $Z$ is a continuous random variable, the last constraint simplifies to all $(x,z) \in \mathbb{R}^2$ such that $p_X(\{x\}) = 0$. 

\begin{lemma}\label{prop:dist_to_relaxed}
If $F_n(x,\hat{x})$ convergence weakly in probability to a limit $F(x,z)$ characterized by a distribution $p_X$ and noise power $\sigma^2 > 0$, then the distortion between the sparsity pattern estimate $\hat{S}$ generated in the second stage of recovery and the set $\tilde{S}$ described in Section~\ref{sec:relaxed_recovery} obeys
\begin{align}
\lim_{n \rw \infty} d(\tilde{S},\hat{S}) \overset{prob}{=} D_\text{awgn}(\sigma^2;p_X)
\end{align}
where $D_\text{awgn}(\sigma^2;p_X)$ is given by \eqref{eq:D_sig2}.
\end{lemma}
\begin{proof}
For each problem of size $n$, define
\begin{align*}
 \tilde{U} = \{ i \in [n] : |x_i| > \delta \} \quad \text{and} \quad \hat{U} = \{ i \in [n] : |\hat{X}_i| > t\},
\end{align*} 
where $\delta > 0$ satisfies $\Pr[|X|=\delta] = 0$ and $t$ is the unique solution to $\Pr[|Z| \ge t] =\kappa$. Note that $t$ corresponds to the minimizer of the right hand side of \eqref{eq:D_sig2}.

By the triangle inequality, we have
\begin{align}
\big| d(\tilde{\bs},\hat{\bs}) -d(\tilde{U},\hat{U}) \big|  \le d(\tilde{\bs},\tilde{U}) + d(\hat{U},\hat{\bs})\label{eq:lim_d_triangle}
\end{align}
Furthermore, by the weak convergence of $F_n(x,\hat{x})$ to $F(x,z)$ and the definitions of $\tilde{S}$ and $\hat{S}$, it can be shown that,
\begin{align}
\lim_{n \rw \infty} d(\tilde{\bs},\tilde{U}) &=\Pr[ |X| \le \delta | X \ne 0] \label{eq:lim_d_c}\\
\lim_{n \rw \infty} d(\tilde{U},\hat{U}) &\overset{prob}{=} \Pr[|X| \le \delta \, |\, |Z| > t] \label{eq:lim_d_a}\\
\lim_{n \rw \infty} d(\hat{U},\hat{\bs}) &\overset{prob}{=}  0 \label{eq:lim_d_b},
\end{align}
where \eqref{eq:lim_d_a} and \eqref{eq:lim_d_b} follow from the definition of $t$.

By the assumptions on $p_X$ and the definition of $D_\text{awgn}(\sigma^2;p_X)$, there exists, for any $\epsilon >0$, a $\delta > 0$ such that $\Pr[|X| = \delta] = 0$ and
\begin{align}
 \Pr[ |X| \le \delta | X \ne 0] &\le \epsilon\\
\big| \Pr[|X| \le \delta \,|\, |Z| > t] - D_\text{awgn}(\sigma^2;p_X) \big| & \le \epsilon.
\end{align}
Hence, we have shown that
\begin{align}
\lim_{n \rw \infty} \Pr\big[ \big| d(\tilde{S},\hat{S}) - D_\text{awgn}(\sigma^2;p_X) \big| > \epsilon' \big] = 0
\end{align}
for any $\epsilon' > 0$ which completes the proof. 
\end{proof}

%%%%%%%%%%%%%%%%%%%%
\subsection{Proof of Theorem \ref{thm:MF}}
\label{sec:MF-proof}
In this section, we prove convergence of the empirical CDF $F_n(x,\hat{x})$ corresponding to the MF estimate. Theorem \ref{thm:MF} then follows immediately from Lemmas \ref{prop:relaxed_to_strict} and \ref{prop:dist_to_relaxed}.

The crucial step in this proof is the following result which characterizes the limiting joint distribution of a randomly chosen subset of the entries in $(\bx,\hat{\bX}^{(\text{MF})})$. Due to the simplicity of the MF estimate, we are able to prove this convergence generally for any i.i.d.~distribution on the entries of the measurement matrix $\bA$.

\begin{lemma}\label{lem:MF_decoupling}
Let $L$ be a fixed integer. For each problem of size $n \ge L$, let $\mathcal{L}$ be distributed uniformly over all subsets of $[n]$ of size $L$. Then, under Assumptions S2-S3 and M1-M4, the joint distribution on $\{(x_\ell,\hat{X}^{(\text{MF})}_\ell)\}_{\ell \in \mathcal{L}}$ converges weakly to the joint distribution on $L$ independent copies of the random pair $(X,Z)$ given in Definition~\ref{def:scalar_model} where $\sigma^2$ is given by
\begin{align}
\sigma^2 = \frac{\snr^{-1}+1}{\rho}. \label{eq:sig2_MF}
\end{align}
\end{lemma}
\begin{proof}To gain intuition, observe that the entries in the MF estimate indexed by $\mathcal{L}$ can be decomposed as follows:
\begin{align}
\hat{\bX}^{(\text{MF})}_\mathcal{L} &= \textstyle \big( \frac{n}{m}\big)
\bA_\mathcal{L}^T  \bA_{\mathcal{L}} \bx_{\mathcal{L}} +
 {\textstyle \big( \frac{n}{m}\big)} \bA_\mathcal{L}^T \big( \bA_{\mathcal{L}^c} \bx_{\mathcal{L}^c} + \frac{1}{\sqrt{\snrs}} \bW\big).  \label{eq:MF_decouple_one}
\end{align}
By the law of large numbers, it is straightforward to show that the first term on the right hand side of \eqref{eq:MF_decouple_one} converges in distribution to random vector $\bX$ whose entires are i.i.d.~copies of $X$. Also, by the central limit theorem, it is straightforward to show that the second term converges in distribution to a vector whose entries are i.i.d.~$\mathcal{N}(0,\sigma^2)$. However, since the terms in \eqref{eq:MF_decouple_one} are {\em not} mutually independent, these arguments are, by themselves, insufficient to prove Lemma~\ref{lem:MF_decoupling}.

To proceed, we will introduce an additional term that allows us to decompose $\hat{\bX}_{\mathcal{L}}^{(\text{MF})}$ into independent components. Specifically, for each problem of size $n$, let $\tilde{\bA}$ be an $m \times L$ random matrix whose columns are independent copies of the columns of $\bA$ and define the random vectors
\begin{align}
\bU &= \big[ {\textstyle \big( \frac{n}{m}\big)} \bA_\mathcal{L}^T\big(   \bA_{\mathcal{L}} - \tilde{\bA}\big) - I_{L \times L}\big] \bx_{\mathcal{L}} \\
\bV&=  {\textstyle \big( \frac{n}{m}\big)} \bA_\mathcal{L}^T \big(\tilde{\bA} \bx_\mathcal{L}  + \bA_{\mathcal{L}^c} \bx_{\mathcal{L}^c} + \snr^{-1/2} \bW\big).
\end{align}
Then, we can write
\begin{align}
\hat{\bX}^{(\text{MF})}_\mathcal{L} &=\bx_{\mathcal{L}}  + \bU  + \bV\label{eq:decouple_one}
\end{align}
where the vectors $\bx_{\mathcal{L}}$ and $\bV$ are independent.

From here, the proof is straightforward. If the following limits hold,
\begin{align}
&\lim_{n \rw \infty} \bx_\mathcal{L} \overset{dist}{=} \bX \label{eq:decouple_a} \\
&\lim_{n \rw \infty}  \bU \overset{prob}{=} \mathbf{0}_{L \times 1} \label{eq:decouple_b} \\
&\lim_{n \rw \infty}  \bV \overset{dist}{=} \mathcal{N}(0,\sigma^2 I_{L \times L} ) \label{eq:decouple_c},
\end{align}
then the desired convergence follows immediately from Slutsky's theorem. 

The limit \eqref{eq:decouple_a} follows from Assumption S2, and the fact that $L$ is finite. To prove \eqref{eq:decouple_b}, observe that by Assumptions M1-M4 and the weak law of large numbers, $ \bA_\mathcal{L}^T   \bA_{\mathcal{L}} \rw (m/n) I_{L \times L}$ and $ \bA_\mathcal{L}^T   \tilde{\bA}_{\mathcal{L}} \rw\mathbf{0}_{L \times L} $ in probability as $n \rw \infty$. Combining these facts with \eqref{eq:decouple_a} shows that $\bU$ converges to $\mathbf{0}_{L \times 1}$ in distribution, and thus also in probability.%, which proves \eqref{eq:decouple_b}.

Finally, to prove \eqref{eq:decouple_c}, observe that $\bV = \sum_{i=1}^m \bV_i$ where
\begin{align}
\bV_i = {\textstyle \big( \frac{n}{m}\big)} (\bA_\mathcal{L}^T)_i \big(\tilde{\bA} \bx_\mathcal{L}  + \bA_{\mathcal{L}^c} \bx_{\mathcal{L}^c} + \snr^{-1/2} \bW\big)_i
\end{align}
and $(\bA_\mathcal{L}^T)_i$ denotes the $i$'th column of the $L \times m$ matrix $\bA_\mathcal{L}^T$. Since the entries in $\bA$, $\tilde{\bA}$, and $\bW$ are mutually independent, it can be verified that the vectors $\{\bV_i\}_{i \in [m]}$ are i.i.d.~with mean zero and covariance
\begin{align}
 \bE[ \bV_i\bV_i^T]
 & =   \Big( \frac{n}{m^2}\Big)  \Big({\frac{1}{n}} \|\bx\|^2 + \snr^{-1}\Big)  I_{L \times L}.
\end{align}
Therefore, the limit \eqref{eq:decouple_c} follows from the multivariate central limit theorem and Assumption S3. 
\end{proof}

With Lemma~\ref{lem:MF_decoupling} in hand, we can now prove convergence of the empirical CDF $F_n(x,\hat{x})$ directly from Chebyshev's inequality. 

\begin{lemma}\label{lem:MF_convegence_P} Under Assumptions S2-S3 and M1-M4, the empircal CDF $F_n(x,\hat{x})$ corresponding to the MF estimate converges weakly in probability to a limit $F(x,z)$ with noise power $\sigma^2$ given by \eqref{eq:sig2_MF}. 
\end{lemma}

\begin{proof}
Beginning with Chebyshev's inequality, we have
\begin{align}
&\Pr\big[ \big|F_n(x,\hat{x}) - F(x,\hat{x})\big| > \epsilon \big] \nonumber \\
& \le \epsilon^{-2}\, \bE\big[\big|F_n(x,\hat{x}) - F(x,\hat{x})\big|^2 \big] \nonumber \\
& =  \epsilon^{-2} \,\big| \bE\big[ F_n^2(x,\hat{x})] - F^2(x,\hat{x})\big|   \nonumber \\
& \quad -  \epsilon^{-2} 2 \big| \bE[F_n(x,\hat{x})] - F(x,\hat{x}) \big| \label{eq:mc_Fn2F}
\end{align}
for any $\epsilon > 0$.
By the linearity of expectation, we can write
\begin{align}
\bE[ F_n(x,\hat{x})] 
& =   \Pr[ x_{\ell_1} \le x, \hat{X}_{\ell_1} \le \hat{x}]\\
\bE[ F^2_n(x,\hat{x})]  &=  \textstyle \frac{n-1}{n}\Pr[ x_{\ell_1} \le x, \hat{X}_{\ell_1} \le \hat{x}, x_{\ell_2} \le x, \hat{X}_{\ell_2} \le \hat{x}]\nonumber \\
&  \quad + \textstyle \frac{1}{n} \Pr[ x_{\ell_1} \le x, \hat{X}_{\ell_1} \le \hat{x}] 
\end{align}
where $\ell_1$ and $\ell_2$ are drawn uniformly at random without replacement from $[n]$. Hence, by Lemma~\ref{lem:MF_decoupling}, it follows that
\begin{align*}
\lim_{n \rw \infty} \bE[F_n(x,\hat{x})] = F(x,\hat{x})\\
\lim_{n \rw \infty} \bE[F^2_n(x,\hat{x})] = F^2(x,\hat{x}).
\end{align*}
Therefore, both terms on the right hand side of \eqref{eq:mc_Fn2F} converge to zero as $n \rw \infty$, thus completing the proof.
\end{proof}

%%%%%%%%%%%%%%%%%%%%%%%%
\subsection{Proof of Theorem~\ref{thm:LMMSE}}
\label{sec:LMMSE-proof}

For this proof, we use the well known fact (see e.g.~\cite{GraSch01}) that matrix inversion can be performed using iterative methods. Specifically, for a fixed tuple $(\by,A,\snr)$, let $\gamma$ be the unique positive solution to quadratic equation 
\begin{align}
\snr =  \gamma \Big( \frac{m}{n} - \frac{1}{1+\gamma}\Big),\label{eq:gamma_n}
\end{align}
and consider the AMP algorithm with $\eta(z,\sigma^2) = z/(1+\gamma)$. If the sequences $\{\bx^t\}_{t\ge 1}$ and $\{\bu^t\}_{t \ge 1}$ converge to a fixed point $(\bx^\infty,\bu^\infty)$, then it can be verified by checking the update equation \eqref{eq:amp_iteration_a} and \eqref{eq:amp_iteration_b} that $\bx^\infty = \bx^{(\text{LMMSE})}$, $A^T \bu^\infty = \gamma \bx^{(\text{LMMSE})}$, and thus
\begin{align}
\bx^{(\text{AMP})} = (1+\gamma)  \bx^{(\text{LMMSE})}.
\end{align}
Therefore, the LMMSE estimate can be computed using the appropriate linear version of AMP, provided that the AMP algorithm converges. 

We now use the analysis of Bayati and Montanari to characterize the limiting behavior of the AMP estimate. For each problem of size $n$ let $\hat{\bX}^{(\text{AMP})}$ denote the output of the AMP algorithm corresponding to the function $\eta(z,\sigma^2) = z/(1+\gamma_n)$ where $\gamma_n$ is the unique positive solution to \eqref{eq:gamma_n}. Then, under Assumptions S2-S3 and M1-M5, it follows from part (b) of \cite[Lemma~1]{BayMon10a} that the empirical CDF corresponding to $\hat{\bX}^{(\text{AMP})}$ converges weakly almost surely to a limit $F(x,z)$ with a noise power  $\sigma^2_\infty$ that is the unique solution to the quadratic equation
\begin{align}
\rho = \frac{1}{\sigma^2_\infty\,\snr} + \frac{1}{1+\sigma^2_\infty}. 
\end{align}
Since the LMMSE estimate is proportional to the AMP estimate, this result, along with Lemmas \ref{prop:relaxed_to_strict} and \ref{prop:dist_to_relaxed}, completes the proof of Theorem~\ref{thm:LMMSE}.

%%%%%%%%%%%%%%%%%%%%%%%%%%
\subsection{Proof of Theorem \ref{thm:amp_general}}
\label{sec:amp_general-proof}

To begin, consider a modified version of the AMP algorithm in which the sequence of noise power estimates $\{\hat{\sigma}^2_t\}_{t \ge 1}$ is replaced with the sequence of noise powers $\{\sigma^2_t\}_{t \ge 1}$ defined by the state evolution recursion \eqref{eq:state_evolution}. (Note that this modified algorithm depends explicitly on the distribution $p_X$.) For each problem of size $n$, let 
\begin{align}
\hat{\bX}^t = \bA^T \bU^t + \bX^t
\end{align}
denote the modified AMP estimate at iteration $t$. Then, under Assumptions S2-S3 and M1-M5, it follows from part (b) of \cite[Lemma~1]{BayMon10a} that the empirical CDF corresponding to $\hat{\bX}^t$ converges weakly almost surely to a limit $F(x,z)$ with noise power  $\sigma^2_t$. 

Moreover, by part (c) of \cite[Lemma~1]{BayMon10a} it can be shown that, under the same assumptions, the residuals $\bU^t$ corresponding to the modified AMP algorithm obey
\begin{align}
\lim_{n \rw \infty} \textstyle \frac{1}{n} \| \bU^t(n)\|^2 = \sigma^2_t
\end{align}
almost surely. Thus, by the continuity of $\eta(z,\sigma^2)$ with respect to $\sigma^2$, it follows that the AMP algorithm using the empirical estimates $\hat{\sigma}^2_t$ has the same limiting behavior as the modified AMP algorithm. 

By the above arguments, the empirical CDF $F_n(x,\hat{x})$ corresponding to the AMP estimate \eqref{eq:x_hat_amp} converges weakly almost surely, and hence also in probability, to a limit $F(x,z)$ with noise power $\sigma^2_\infty$ given in \eqref{eq:sig2_infty}. Combining this result with Lemmas \ref{prop:relaxed_to_strict} and \ref{prop:dist_to_relaxed} completes the proof of Theorem \ref{thm:amp_general}.

%%%%%%%%%%%%%%%%%%%%%%%%%%
\subsection{Proof of Theorem~\ref{thm:MMSE}}
\label{sec:mmse-proof}

This proof follows along the same lines as the proof of Theorem~\ref{thm:MF}. The key step is the following result which is analogous to Lemma~\ref{lem:MF_decoupling} except that it also requires the replica analysis assumptions. This result is stated as Claim~3 in \cite{GuoBarSha09}, and its proof follows directly from the analysis in \cite[Section~IV-B]{GuoVer05}.

\begin{lemma}\label{lem:MMSE_decoupling} 
Assume that the replica analysis assumptions hold.  
Let $L$ be a fixed integer. For each problem of size $n \ge L$, let $\mathcal{L}$ be distributed uniformly over all subsets of $[n]$ of size $L$. Then, under Assumptions S2-S3 and M1-M4, the joint distribution on $\{(x_\ell,\hat{X}^{(\text{MMSE})}_\ell)\}_{\ell \in \mathcal{L}}$ converges weakly to the joint distribution on $L$ independent copies of the random pair $(X,Z)$ given in Definition~\ref{def:scalar_model} where $\sigma^2$ is given by the noise power $\tau^*$ defined in \eqref{eq:F_tau}. 
\end{lemma}

From here, the rest of the proofs follows immediately from Chebyshev's inequality (see Lemma~\ref{lem:MF_convegence_P}).

%%%%%%%%%%%%%%%%%%%%%
%\input{sections/appendix_scaling_behavior}
\section{Scaling Behavior}

This appendix provides additional analysis of the sampling rate-distortion bounds presented in Section~\ref{sec:main_results}.

%%%%%%%%%%%%%%%%%%%%%%%%%%%%%%%%%
\subsection{Behavior of the ML Upper Bound}\label{sec:behavior_ML}

This section studies the scaling behavior of the upper bound $\rho^{(\text{ML-UB})}$ given in Theorem~\ref{thm:ML}. For notational simplicity, we will use the notation $\Lambda(D)$, $P(D)$ and $\mathcal{H}(D)$ where the dependence on $\snr$ and $p_X$ is implicit. Recall that the upper bound is given by
 $$\rho^{(\text{ML-UB})} =  \kappa + \max_{\tilde{D} \in [D,1]} \Lambda(D).$$ 
 
We first consider the behavior as $D \rw 0$. Note that the function $\Lambda(D)$ is finite for all $D > 0$ but grows without bound as $D \rw 0$, and hence
\begin{align}
\lim_{D \rw 0} \frac{P(D)}{\mathcal{H}(D)} \max_{\tilde{D} \in [D,1]} \Lambda(\tilde{D})
&=\lim_{D \rw 0}  \frac{P(D)}{\mathcal{H}(D)} \Lambda(D).
\end{align}
Starting with the definition of $\Lambda(D)$ given in \eqref{eq:C_beta}, it is straightforward to show that
\begin{align}
\lim_{D \rw 0}  \frac{P(D)}{\mathcal{H}(D)} \Lambda(D) & = \frac{2}{\snr} \lim_{D \rw 0} \lambda(D) 
\end{align}
where
\begin{align}
&\lambda(D)  =  \min_{\theta,\mu \in (0,1)} \max\Big\{ \frac{4}{(1\!-\!\theta)^2}, \frac{1}{\mu \theta} - \frac{ D\kappa \log(1\!-\!\mu^2)}{2\mu\theta \mathcal{H}(D)}  \Big\} \label{eq:lambda_D}. 
\end{align}
Using the fact that $D/\mathcal{H}(D)\rw 0$ as $D \rw 0$ gives
\begin{align*}
\lim_{D \rw 0} \lambda(D) &= \min_{\theta \in (0,1)} \max\Big\{\frac{4}{(1-\theta)^2}, \frac{1}{\theta}\Big\}= \frac{1}{3 - \sqrt{8}},
\end{align*}
and putting everything together gives
\begin{align*}
\lim_{D \rw 0} \frac{P(D)}{\mathcal{H}(D)} \left[\rho^{(\text{ML-UB})} - \kappa \right] = \left(\frac{2}{3 - \sqrt{8}}\right) \frac{1}{\snr}.
\end{align*}

We next consider the behavior as a function of the SNR. For any $D> 0$ it is easy to verify that $\Lambda(D) \rw 0$ as $\snr \rw \infty$ and hence the infinite SNR limit is given by
\begin{align}
\lim_{\snrs \rw \infty} \rho^{(\text{ML-UB})}  = \kappa.
\end{align}
To characterize the rate at which the upper bound approaches this limit, let $D>0$ be fixed and observe that
\begin{align}
&\lim_{\snrs \rw \infty} \, \log (\snr) \big[ \rho^{(\text{ML-UB})} - \kappa\big] \nonumber\\ 
& = \lim_{\snrs \rw \infty} \, \log (\snr ) \max_{\tilde{D} \in [D,1]} \Lambda(\tilde{D}) \nonumber \\
& = \max_{\tilde{D} \in [D,1]} 2 \mathcal{H}(\tilde{D})\label{eq:ML_high_SNR_b} \\
& = 2H_b(\kappa)  \label{eq:rho_ML_high_SNR}
\end{align}
where \eqref{eq:ML_high_SNR_b} follows from the fact that $P(D;p_X)$ is strictly positive for any $D > 0$. 

Alternatively, with a bit of work it can be shown that the low SNR behavior is given by
\begin{align}
&\lim_{\snrs \rw 0} \, \snr\, \big[ \rho^{(\text{ML-UB})} - \kappa\big] \nonumber \\
& = \lim_{\snrs \rw 0} \, \snr\, \max_{\tilde{D} \in [D,1]} \Lambda(\tilde{D}) \nonumber \\
& = \max_{\tilde{D} \in [D,1]}  \frac{\mathcal{H}(\tilde{D})}{P(\tilde{D})}  2 \lambda(\tilde{D}), \label{eq:rho_ML_low_SNR}
\end{align}
where $\lambda(D)$ is given by \eqref{eq:lambda_D}. Note that this limit is strictly positive for any $D>0$. 

Combining \eqref{eq:rho_ML_high_SNR} and \eqref{eq:rho_ML_low_SNR} shows that there exists, for each fixed pair $(D,p_X)$, a constant $C$ such that 
\begin{align}
\rho^{(\text{ML-UB})} \le \kappa + \frac{C}{\log(1+\snr)}
\end{align}
for all $\snr$. 

Lastly, we consider the tradeoff between the distortion $D$ and the SNR. For a given tuple $(\rho,\snr,p_X)$, let $D^{(\text{ML-UB})}$ denote the infimum over all distortions $D\ge 0$ such that $\rho^{(\text{ML-UB})} \le \rho$. If $\rho > \kappa$, then the analysis given above shows that $D^{(\text{ML-UB})} \rw 0$ as $\snr \rw \infty$. Since $\Lambda(D)$ is finite for all $D> 0$ but grows without bound as $D \rw 0$, this means that the following limit must be satisfied:
\begin{align}
\lim_{\snrs \rw \infty} \Lambda(D^{(\text{ML-UB})}) = \rho - \kappa \label{eq:low_D_snr_ML_a}.
\end{align}
Starting with the definition of $\Lambda(D)$ given in \eqref{eq:C_beta}, it is straightforward to show that \eqref{eq:low_D_snr_ML_a} is satisfied if and only if
\begin{align}
\lim_{\snrs \rw \infty} \snr \, \frac{P(D^{(\text{ML-UB})})}{\mathcal{H}(D^{(\text{ML-UB})})} = \left( \frac{2}{3-\sqrt{8}}\right) \frac{1}{\rho -\kappa}.
\end{align}

%%%%%%%%%%%%%%%%%%%%%%%%%%%%%%%%%%%
\subsection{Behavior of the MMSE Noise Power}\label{sec:behavior_MMSE}
This section studies the behavior of the effective noise power $\tau^*$ defined in Theorem~\ref{thm:MMSE}. Since there is a one-to-one correspondence between $\tau^*$ and the resulting distortion $D$, the results in this section immediately extend to the behavior of the distortion. 

Starting with the definition in \eqref{eq:F_tau}, this noise power can be expressed as $$\tau^* = \arg \min_{\tau > 0 } \Gamma(\tau)$$ where
\begin{align*}
\Gamma(\tau) = \rho \log(\tau) + \frac{1}{\tau\, \snr} + 2I(X;X+\sqrt{\tau} W).
\end{align*}
For any fixed tuple $(\rho,\snr,p_X)$, the function $\Gamma(\tau)$ grows without bound as either $\tau \rw 0$ or $\tau \rw \infty$. Therefore, the minimizer $\tau^*$ must be a solution to $\Gamma'(\tau^*,\snr) = 0$ where $\Gamma'(\cdot,\cdot)$ denotes the derivative of $\Gamma(\cdot, \cdot)$ with respect to the first argument. Using the following result of Guo et al. \cite{GSV05}:
\begin{align} 
\frac{d}{d\gamma} 2 I(X; X +\sqrt{1/\gamma} W) =  \mmse(1/\gamma;p_X), \label{eq:mmse_info}
\end{align}
it is straightforward to show that the condition $\Gamma'(\tau^*,\snr) = 0$ is equivalent to
\begin{align}
 \rho\, \tau^* =  \frac{1}{\snr} + \mmse(\tau^*;p_X).\label{eq:rho_tau_star_cond}
\end{align}
Note that \eqref{eq:rho_tau_star_cond} may have additional fixed point solutions (other than $\tau^*$) corresponding to local minima or maxima of the function $\Gamma(\tau)$.

We first consider the behavior as $\rho \rw \infty$. By the optimality of the MMSE estimate (with respect to mean square error) the noise power $\tau^*$ is a non-increasing function $\rho$, and thus $\mmse(\tau^*,p_X)$ is a non-increasing function of $\rho$. Combining this fact with  \eqref{eq:rho_tau_star_cond} shows that $\tau^* \rw 0$ as $\rho \rw \infty$. Since $\mmse(\tau,p_X) \rw 0$ as $\tau \rw 0$, we obtain the limit
\begin{align}
\lim_{\rho \rw \infty} \rho\, \tau^* = \frac{1}{\snr}.
\end{align}

We next consider the behavior as $\snr \rw \infty$. If $\tau$ is a fixed constant, independent of $\snr$, then $\Gamma(\tau)$ converges to a finite constant. 
However, if $\tau = \tau(\snr)$ scales with $\snr$ in such a way that $\tau(\snr) \rw 0
$ then
\begin{align}
&\lim_{\snrs \rw \infty} \frac{1}{\log \tau(\snr)} \left[ \Gamma(\tau(\snr))  - \frac{1}{\tau(\snr)\, \snr } \right] \nonumber \\
&=\lim_{\snrs \rw \infty} \frac{1}{\log \tau(\snr)} \left[ \rho \log \tau(\snr) + I\left(X;X+\sqrt{\tfrac{1}{\snrs}}W\right)  \right] \nonumber \\
 & = \rho - \lim_{\epsilon \rw 0}  \frac{2 I(X;X+\sqrt{\epsilon}\, W)}{\log(1/\epsilon)} \nonumber \\
 & = \rho - \lim_{\epsilon \rw 0} \frac{\mmse(\epsilon;p_X)}{\epsilon} \label{eq:lhoptial4mmse}
\end{align}
where \eqref{eq:lhoptial4mmse} follows form L'Hopital's rule and \eqref{eq:mmse_info}.

We consider two cases. If the right hand side of \eqref{eq:lhoptial4mmse} is strictly positive, then there exists a scaling $\tau(\snr)$ such that $\Gamma(\tau(\snr))$ decreases without bound. Since $\Gamma(\tau)$ is finite for fixed $\tau$ and grows without bound as $\tau \rw \infty$, this means that $\tau^* \rw 0$ as $\snr  \rw \infty$. Conversely, if the right hand side of \eqref{eq:lhoptial4mmse} is strictly negative, then $\Gamma(\tau(\snr))$ increases without bound for any scaling where $\tau(\snr) \rw 0$. This means that $\tau^*$ is bounded away from zero for all $\snr$. 

Combining these cases, we can conclude that the stability threshold $\varrho^{(\text{MMSE})}$ of the MMSE estimator is given by 
\begin{align}
\varrho^{(\text{MMSE})} = \lim_{\epsilon \rw 0} \frac{\mmse(\epsilon;p_X)}{\epsilon}.
\end{align}

To characterize the rate at which $\tau^*$ decreases as $\snr \rw \infty$, we rearrange \eqref{eq:rho_tau_star_cond} to obtain
\begin{align}
\snr \, \tau^*= \left[ \rho - \frac{\mmse(\tau^*;p_X)}{\tau^*} \right]^{-1}. \label{eq:snr_tau_star_cond}
\end{align}
If $\rho > \varrho^{(\text{MMSE})}$, then $\tau^* \rw 0$ as $\snr \rw \infty$. Hence, by \eqref{eq:snr_tau_star_cond} and the definition of $\varrho^{(\text{MMSE})}$, we obtain the limit 
\begin{align}
\lim_{\snrs \rw \infty}  \snr\, \tau^* = \frac{1}{\rho - \varrho^{(\text{MMSE})}}.
\end{align}

%%%%%%%%%%%%%%%%%%%%%%%%%
\subsection{Proof of Proposition~\ref{prop:bounded}}\label{sec:prop:bounded-proof}

Using the bound $H_b(p) \le p\log(1/p) + p$ we obtain 
\begin{align}
\mathcal{H}(D;\kappa) \le 2 \kappa D [ \log(1/D) + 1 + \log(\textstyle \frac{1-\kappa}{\kappa})] \label{eq:N_upper_bound}.
\end{align}
Using the definition of $P(D;p_X)$ and the fact that $X$ is lower bounded, we obtain
\begin{align}
P(D;p_X) & = \int_0^\infty \left( \Pr[X^2 \ge u] - (1-D)\kappa)\right)_+ du \nonumber \\
& \ge \int_0^\infty \left( \kappa \one(u < B^2) - (1-D)\kappa)\right)_+ du \nonumber \\
& = \kappa D B^2. \label{eq:P_D_bounded_a}
\end{align}
Combining \eqref{eq:N_upper_bound} and \eqref{eq:P_D_bounded_a} completes the proof of \eqref{eq:P_D_bounded}. 

The bound \eqref{eq:sig2_D_bounded} follows immediately from the upper bound
\begin{align}
& D_\text{awgn}(\sigma^2;p_X) \nonumber\\
&= \min_{t \ge 0} \textstyle \max\Big( \Pr[ |X + \sigma W| \le t], \frac{1-\kappa}{\kappa} \Pr[ |\sigma W| > t]\big) \nonumber \\
&\le \min_{t \ge 0} \textstyle \max\Big( \Pr[ |B + \sigma W| \le t], \frac{1-\kappa}{\kappa} \Pr[ |\sigma W| > t]\big) \nonumber \\
& \le \textstyle \max\Big( \Pr[ |B + \sigma W| \le \frac{B}{2}], \frac{1-\kappa}{\kappa} \Pr[ |\sigma W| > \frac{B}{2}]\big) \nonumber \\
&\textstyle \le \big(\frac{1-\kappa}{\kappa}\big)   \Pr[ |\sigma W| > \frac{B}{2}] \label{eq:D_awgn_bounded_b}\\
& \le \textstyle \big(\frac{1-\kappa}{\kappa}\big) \exp\big(-\frac{B^2}{8\, \sigma^2}\big)\label{eq:D_awgn_bounded_c}
\end{align}
where \eqref{eq:D_awgn_bounded_b} follows from the triangle inequality and \eqref{eq:D_awgn_bounded_c} follows from the well known upper bound (see e.g. \cite{tse_WirelessComm05}) $\Pr[|W| > t] \le \exp(-t^2/2)$.

%%%%%%%%%%%%%%%%%%%%%%%%%%
\subsection{Proof of Proposition~\ref{prop:poly_decay}}
\label{sec:prop:poly_decay-proof}
For this proof, it is convenient to define the quantile function %of $U^2$:
\begin{align*}
\xi(D)
& =  \inf\{ t \ge 0 : \Pr[ |X|^2 \le  t | X\ne 0] \ge D \},
\end{align*}
and note that
\begin{align}
\lim_{D \rw 0} \frac{\xi(D)}{D^{2/L}} =  \tau^{-2/L}. \label{eq:xi_lim}
\end{align}

We first consider \eqref{eq:P_D_poly_decay}. Using the bounds $ p \log(1/p) \le H_b(p) \le p\log(1/p) + p$, we obtain
\begin{align}
\lim_{D \rw 0} \frac{\mathcal{H}(D;\kappa)}{D\log(1/D)} = 2 \kappa. \label{eq:H_lim}
\end{align}
Next, starting from the definition of $P(D;p_X)$ and using a change of variables leads to the expression
\begin{align*}
P(D;p_X) %&= \int_0^\infty \Big( \Pr[X^2 > u] - (1-D)\kappa\Big)_+ du\\
%& = \int_0^\infty \Big( \Pr[U^2 > t] - (1-D)\Big)_+ dt\\
%& = \int_0^\infty \Big(D -  \Pr[U^2 \le t] \Big)_+ dt\\
%& = \int_0^\infty \int_0^D  \one\big( \beta > \Pr[U^2 \le t]\big ) d \beta\, dt\\
%& = \int_0^D \int_0^\infty \one\big( \beta > \Pr[U^2 \le t]\big ) dt \,d\beta\\
%& = \int_0^D \xi(\beta) d\beta.%\\
& = \kappa D \int_0^1  \xi(\beta D) d\beta.
\end{align*}
Thus, we can write
\begin{align}
\lim_{D \rw 0} \frac{P(D;p_X)}{D^{1+2/L}}
& =\kappa \, \int_0^1 \lim_{D \rw 0}  \frac{\xi(\beta D)}{D^{2/L}} d\beta \nonumber \\
& = \kappa\, \tau^{-2/L} \int_0^1 \beta^{2/L} d\beta \nonumber \\
& =  \frac{\kappa\, \tau^{-2/L}}{1+2/L} \label{eq:P_D_lim1}
\end{align}
where swapping the limit and the integral is justified by the fact that $\xi(D)$ is continuous and monotonically increasing. Combining  \eqref{eq:H_lim} and \eqref{eq:P_D_lim1} completes the proof of \eqref{eq:P_D_poly_decay}.

We next consider \eqref{eq:sig2_D_poly_decay}. Let $w_D$ be the unique solution to $\Pr[|W| > w_D] = \kappa D/(1-\kappa)$. Using standard bounds on the cumulative distribution function of the Gaussian distribution (see e.g. \cite{tse_WirelessComm05}) it can be verified that
\begin{align}
\lim_{D \rw 0} \frac{w_D^2}{\log(1/D)} = 2.\label{eq:w_D_lim}
\end{align}
Therefore, by \eqref{eq:xi_lim} and \eqref{eq:w_D_lim}, the limit \eqref{eq:sig2_D_poly_decay} follows immediately if we can show that
\begin{align}
\lim_{D \rw 0} \Big(\frac{\xi(D)}{w_D^2}  \Big)^{-1} \sigma^2_\text{awgn}(D;p_X) = 1. \label{eq:sig2_D_lim}
\end{align}

To proceed, define the probabilities
\begin{align*}
p_1(\theta) &= \textstyle \Pr\big[ \big| \frac{X}{\sqrt{ \xi(D)}} + \frac{W}{w_D}  \big| \le \theta | X \ne 0 \big]\\
p_2(\theta) &=\textstyle \big(\frac{1-\kappa}{\kappa} \big) \Pr[ |\frac{W}{w_D}| > \theta].
\end{align*}
and note that
\begin{align}
 \textstyle D_\text{awgn}\big( \frac{\xi(D)}{w_D^2} ;p_X \big) &= \inf_{\theta \in \mathbb{R}} \max\big(p_1(\theta),p_2(\theta)\big). \label{eq:D_awgn_sig2_D}
\end{align}
By a change of variables, we can write 
\begin{align*}
p_1(1) 
& = D \int_0^\infty  \textstyle \one(\beta \le \frac{1}{D})  \Pr\big[ \big| \sqrt{\frac{\xi(\beta D)}{ \xi(D)}}  + \frac{W}{w_D}  \big| \le 1\big] d\beta.
\end{align*}
Using \eqref{eq:xi_lim} and the fact that $\xi(D)$ is a strictly decreasing function when $D$ is small, it can be shown that the integrand of the above expression converges pointwise to $\one( \beta \le 1)$ and hence
\begin{align}
\lim_{D \rw 0} D^{-1} p_1(1)= 1 \label{eq:lim_sig_D_a}.
\end{align}
Since $p_1(\theta)$ is a strictly increasing function of $\theta$ and $p_2(\theta)$ is a strictly decreasing function of $\theta$ 
with $p_2(1) =D$, it thus follows that
\begin{align*}
\lim_{D \rw 0} D^{-1}  \textstyle D_\text{awgn}\big(\frac{\xi(D)}{w_D^2} ;p_X \big)  = 1. 
\end{align*}
Since $D_\text{awgn}(\sigma^2;p_X)$ is a strictly increasing function of $\sigma^2$, this proves the limit \eqref{eq:sig2_D_lim}, and thus completes the proof of \eqref{eq:sig2_D_poly_decay}.

%%%%%%%%%%%%%%%%%%%%%%%
\section{Properties of Soft Thresholding}\label{sec:ST_behavior}

This Appendix reviews several useful properties of the soft-thresholding  noise sensitivity $\mathcal{M}(\sigma^2,\alpha,p_X)$ introduced in Section~\ref{sec:AMP}.

To begin, observe that the noise sensitivity defined in \eqref{eq:noise_sensitivity} can be expressed as
\begin{align}
\mathcal{M}(\sigma^2,\alpha,p_X) & =\frac{ \bE\big[ | \eta^{(\text{ST})}(X + \sigma W,\sigma^2;\alpha) - X|^2\big]}{\sigma^2}\\ 
&=  \bE[ \mu(X/\sigma,\alpha)]
\end{align}
where $\mu(z,\alpha)$ is given by
\begin{align}
\mu(z,\alpha) &= \bE\big[ | \eta^{(\text{ST})}(z + W,1;\alpha) -z|^2\big]. \label{eq:mu_ST}
\end{align}
With a bit of calculus, it can then be verified that
\begin{align}
\mu(z,\alpha) & =  z^2 \big[ 1 - \Phi(-\alpha + z) - \Phi(-\alpha - z)\big]  \nonumber \\
& \quad + (1+ \alpha^2) \big[\Phi(-\alpha + z)  +  \Phi(-\alpha - z)\big]  \nonumber \\
& \quad
 - (\alpha + z)\phi(\alpha-z)  - (\alpha - z)\phi(\alpha+z),
\end{align}
where $\phi(x) = (2 \pi)^{-1/2} e^{-x^2/2}$ and $\Phi(x) = \int_{-\infty}^x \phi(t) dt$. 

If we let $\tilde{X}$ be distributed according to the nonzero part of $p_X$, then we obtain the general expression
\begin{align}
\mathcal{M}(\sigma^2,\alpha,p_X) = (1-\kappa)\mu(0,\alpha) + \kappa\, \bE[\mu( \textstyle \frac{1}{\sigma} \tilde{X}, \alpha )]. \label{eq:M_ST_general}
\end{align}

%%%%%%%%%%%%%%%
\subsection{Infinite SNR Limit}
The infinite SNR limit of the AMP-ST bound corresponds the limit of $\mathcal{M}(\sigma^2,\alpha,p_X)$ as the noise power $\sigma^2$ tends to zero. A simple exercise shows that
\begin{align}
\lim_{\sigma^2 \rw 0}  \bE[\mu( \textstyle \frac{1}{\sigma} \tilde{X}, \alpha )] = (1+\alpha^2)
\end{align}
for any random variable $\tilde{X}$ with $\Pr[\tilde{X} = 0] = 0$. Therefore, for any distribution $p_X \in \mathcal{P}(\kappa)$, we obtain the general limit
\begin{align}
\lim_{\sigma^2 \rw 0} \mathcal{M}(\sigma^2,\alpha,p_X) = \mathcal{M}_0(\alpha,\kappa)
\end{align}
where
\begin{align}
&\mathcal{M}_0(\alpha,\kappa) \nonumber\\%= (1-\kappa) (1+\alpha^2) 2 \Phi(-\alpha) - 2 \alpha \phi(\alpha)  + \kappa (1+\alpha^2)
& = \kappa (1+\alpha^2) + (1-\kappa) 2\big [ (1+\alpha^2) \Phi(-\alpha) - \alpha \phi(\alpha)\big]. \label{eq:M_zero}
\end{align}
Minimizing $\mathcal{M}_0(\alpha,\kappa)$ as a function of $\alpha$ recovers the $\ell_1/\ell_0$ equivalence threshold of Donoho and Tanner \cite{DT09}.

%%%%%%%%%%%%%%%
\subsection{Universal Bounds}
In \cite{DonJoh94}, it is shown that, over the class of distributions $\mathcal{P}(\kappa)$, the noise sensitivity is maximized at a ``three-point'' distribution that places all of its nonzero mass at $\pm 1/\sqrt{\kappa}$. Combining this result with \eqref{eq:M_ST_general} leads to the uniform upper bound
\begin{align}
\sup_{p_X \in \mathcal{P}(\kappa)} \mathcal{M}(\sigma^2,\alpha,p_X) = \mathcal{M}^*(\sigma^2,\alpha,\kappa) \label{ew:M_ST_upperbound}
\end{align}
where
\begin{align}
 \mathcal{M}^*(\sigma^2,\alpha,\kappa)& =  (1-\kappa)\mu(0,\alpha) + \kappa \mu( \tfrac{1}{\sigma \sqrt{\kappa}}, \alpha ).
\end{align}

Using \eqref{ew:M_ST_upperbound} it is now possible extend the bound given in Theorem~\ref{thm:AMP_ST} to a given class of distributions $\mathcal{P}_X\subset \mathcal{P}(\kappa)$. Specifically, we can conclude that a distortion $D$ is achievable for a tuple $(\rho, \mathcal{P}_X, \snr)$ if
\begin{align}
\rho > \frac{1}{ \sigma^2\, \snr} +  \mathcal{M}^*(\sigma^2,\alpha,\kappa) \label{eq:AMP-ST_minimax}
%D > \max_{p_X \in \mathcal{F}} \sigma^2_\text{awgn}(D
\end{align}
where
\begin{align}
\sigma^2 = \min_{p_X \in \mathcal{P}_X} \sigma^2_\text{awgn}(D;p_X). \label{eq:sig2_universal}
\end{align} 

We note that the bounds \eqref{eq:AMP-ST_minimax} and \eqref{eq:sig2_universal} can be used to find a value of soft-thresholding parameter $\alpha$ that works well uniformly over the class $\mathcal{P}_X$. However, since these universal bounds are not tight, we cannot conclude that the resulting value of $\alpha$ is minimax optimal.

%%%%%%%%%%%
\section*{Acknowledgment} We would like to thank Martin Wainwright for helpful discussions and pointers in early versions of this work, David Donoho for a careful reading of the manuscript, and the anonymous reviewers for their helpful comments and suggestions. 
\bibliographystyle{ieeetran}
%\bibliography{../gbib-long}
%\bibliography{appox-sampling-upper-bounds.bbl}

\begin{thebibliography}{10}
\providecommand{\url}[1]{#1}
\csname url@samestyle\endcsname
\providecommand{\newblock}{\relax}
\providecommand{\bibinfo}[2]{#2}
\providecommand{\BIBentrySTDinterwordspacing}{\spaceskip=0pt\relax}
\providecommand{\BIBentryALTinterwordstretchfactor}{4}
\providecommand{\BIBentryALTinterwordspacing}{\spaceskip=\fontdimen2\font plus
\BIBentryALTinterwordstretchfactor\fontdimen3\font minus
  \fontdimen4\font\relax}
\providecommand{\BIBforeignlanguage}[2]{{%
\expandafter\ifx\csname l@#1\endcsname\relax
\typeout{** WARNING: IEEEtran.bst: No hyphenation pattern has been}%
\typeout{** loaded for the language `#1'. Using the pattern for}%
\typeout{** the default language instead.}%
\else
\language=\csname l@#1\endcsname
\fi
#2}}
\providecommand{\BIBdecl}{\relax}
\BIBdecl

\bibitem{Donoho_IT06}
D.~L. Donoho, ``Compressed sensing,'' \emph{IEEE Trans. Inf. Theory}, vol.~52,
  no.~4, pp. 1289--1306, Apr. 2006.

\bibitem{CandesRombergTao06A}
E.~J. Cand\`{e}s, J.~Romberg, and T.~Tao, ``Robust uncertainty principles:
  exact signal reconstruction from highly incomplete frequency information,''
  \emph{IEEE Trans. Inf. Theory}, vol.~52, no.~2, pp. 489--509, Feb. 2006.

\bibitem{CandesTao06}
E.~J. Cand\`{e}s and T.~Tao, ``Near optimal signal recovery from random
  projections: Universal encoding strategies?'' \emph{IEEE Trans. Inf. Theory},
  vol.~52, no.~12, pp. 5406--5425, Dec. 2006.

\bibitem{DevoLore93}
R.~A. DeVore and G.~G. Lorentz, \emph{Constructive Approximation}.\hskip 1em
  plus 0.5em minus 0.4em\relax New York, NY: Springer Verlag, 1993.

\bibitem{Chen_JSC98}
S.~S. Chen, D.~L. Donoho, and M.~A. Saunders, ``Atomic decomposition by basis
  pursuit,'' \emph{SIAM J. of Sci. Comp.}, vol.~20, no.~1, pp. 33--61, 1999.

\bibitem{Miller90}
A.~J. Miller, \emph{Subset selection in regression}.\hskip 1em plus 0.5em minus
  0.4em\relax New York, NY: Chapman-Hall, 1990.

\bibitem{MeinBuhl06}
N.~Meinshausen and P.~B\"{u}hlmann, ``High-dimensional graphs and variable
  selection with the lasso,'' \emph{Annals of Stat.}, vol.~34, pp. 1436--1462,
  2006.

\bibitem{ZhaoYu_JMLR06}
P.~Zhao and B.~Yu, ``On model selection consistency of lasso,'' \emph{J. of
  Machine Learning Research}, vol.~51, no.~10, pp. 2541--2563, Nov. 2006.

\bibitem{Wainwright_SharpThresholds_IEEE09}
M.~J. Wainwright, ``Sharp thresholds for noisy and high-dimensional recovery of
  sparsity using $\ell_1$-constrained quadratic programming (lasso),''
  \emph{IEEE Trans. Inf. Theory}, vol.~55, no.~5, pp. 2183--2202, May 2009.

\bibitem{Wainwright_InfoLimits_IEEE09}
------, ``Information-theoretic limitations on sparsity recovery in the
  high-dimensional and noisy setting.'' \emph{IEEE Trans. Inf. Theory},
  vol.~55, pp. 5728--5741, Dec. 2009.

\bibitem{FLetcher_IEEE09}
A.~K. Fletcher, S.~Rangan, and V.~K. Goyal, ``Necessary and sufficient
  conditions for sparsity pattern recovery,'' \emph{IEEE Trans. Inf. Theory},
  vol.~55, no.~12, pp. 5758--5772, Dec. 2009.

\bibitem{Wang_IEEE10}
W.~Wang, M.~J. Wainwright, and K.~Ramchandran, ``Information-theoretic limits
  on sparse signal recovery: Dense versus sparse measurement matrices,''
  \emph{IEEE Trans. Inf. Theory}, vol.~56, no.~6, pp. 2967--2979, Jun. 2010.

\bibitem{akcakaya_IEEE10}
M.~Akcakaya and V.~Tarokh, ``Shannon theoretic limits on noisy compressive
  sampling,'' \emph{IEEE Trans. Inf. Theory}, vol.~56, no.~1, pp. 492--504,
  Jan. 2010.

\bibitem{AerSalZha10}
S.~Aeron, V.~Saligrama, and M.~Zhao, ``Information theoretic bounds for
  compressed sensing,'' \emph{IEEE Trans. Inf. Theory}, vol.~56, no.~10, pp.
  5111--5130, Oct. 2010.

\bibitem{Reeves_masters}
G.~Reeves, ``Sparse signal sampling using noisy linear projections,''
  Department of EECS, UC Berkeley, Tech. Rep. UCB/EECS-2008-3, Jan. 2008.

\bibitem{RG08}
G.~Reeves and M.~Gastpar, ``Sampling bounds for sparse support recovery in the
  presence of noise,'' in \emph{Proc. IEEE Int. Symp. on Inf. Theory}, Toronto,
  Canada, Jul. 2008.

\bibitem{RG-Lower-Bounds}
------, ``Approximate sparsity pattern recovery: Information-theoretic lower
  bounds,'' Feb. 2010, arXiv:1002.4458v1 [cs.IT].

\bibitem{Tanaka02}
T.~Tanaka, ``A statistical-mechanics approach to large-system analysis of cdma
  multiuser detectors,'' \emph{IEEE Trans. Inf. Theory}, vol.~48, no.~11, pp.
  2888--2910, Nov. 2002.

\bibitem{Muller03}
R.~R. Muller, ``{Channel capacity and minimum probability of error in large
  dual antenna array systems with binary modulation},'' \emph{IEEE Trans. Inf.
  Theory}, vol.~51, no.~11, pp. 2821--2822, Nov. 2003.

\bibitem{GuoVer05}
D.~Guo and S.~Verd\'{u}, ``Randomly spread {CDMA}: Asymptotics via statistical
  physics,'' \emph{IEEE Trans. Inf. Theory}, vol.~51, no.~6, pp. 1983--2010,
  Jun. 2005.

\bibitem{KabWadTan09}
Y.~Kabashima, T.~Wadayama, and T.~Tanaka, ``A typical reconstruction limit of
  compressed sensing based on lp norm minimization,'' \emph{J. Stat. Mech.}, p. L09003, 2009.

\bibitem{RanFleGoy09}
S.~Rangan, A.~K. Fletcher, and V.~K. Goyal, ``Asymptotic analysis of map
  estimation via the replica method and applications to compressed sensing,''
  in \emph{Proc. Neural Information Processing Systems Conf.}, vol.~22,
  Vancouver, CA, Dec. 2009, pp. 1545--1553.

\bibitem{GuoBarSha09}
D.~Guo, D.~Baron, and S.~Shamai, ``A single-letter characterization of optimal
  noisy compressed sensing,'' in \emph{Proc. Allerton Conf. on Comm., Control,
  and Computing}, Monticello, IL, Sep. 2009.

\bibitem{DonMalMon09}
D.~L. Donoho, A.~Maleki, and A.~Montanari, ``Message passing algorithms for
  compressed sensing,'' in \emph{Proc. Nat. Acad. Sci.}, vol. 106, 2009, pp.
  18\,914--18\,919.

\bibitem{DonMalMon10}
\BIBentryALTinterwordspacing
------, ``The noise-sensitivity phase transition in compressed sensing,'' Apr.
  2010. [Online]. Available: \url{http://arxiv.org/abs/1004.1218}
\BIBentrySTDinterwordspacing

\bibitem{BayMon10a}
M.~Bayati and A.~Montanari, ``The dynamics of message passing on dense graphs,
  with applications to compressed sensing,'' \emph{IEEE Trans. Inf. Theory},
  vol.~57, no.~2, pp. 764--785, Feb. 2011.

\bibitem{Montanari11}
\BIBentryALTinterwordspacing
A.~Montanari. (2011, Mar.) Graphical model concepts in compressed sensing.
  [Online]. Available: \url{http://arxiv.org/abs/1011.4328}
\BIBentrySTDinterwordspacing

\bibitem{Mallat93}
S.~G. Mallat and Z.~Zhang, ``Matching pursuits with time-frequency
  dictionaries,'' \emph{IEEE Trans. Signal Process.}, vol.~41, no.~12, pp.
  3397--3415, Dec. 1993.

\bibitem{Tibsh96}
R.~Tibshirani, ``Regression shrinkage and selection via the lasso,,'' \emph{J.
  Royal Stat. Soc., Ser. B}, vol.~58, no.~1, pp. 267--288, 1996.

\bibitem{Fuchs_IT05}
J.~J. Fuchs, ``Recovery of exact sparse representations in the presence of
  noise,'' \emph{IEEE Trans. Inf. Theory}, vol.~51, no.~10, pp. 3601--3608,
  Oct. 2005.

\bibitem{CandesTao_IT05}
E.~J. Cand\`{e}s and T.~Tao, ``Decoding by linear programming,'' \emph{IEEE
  Trans. Inf. Theory}, vol.~51, no.~12, pp. 4203--4215, Dec. 2005.

\bibitem{DonohoEladTemlyakov_IEEtrans06}
D.~Donoho, M.~Elad, and V.~Temlyakov, ``Stable recovery of sparse overcomplete
  representations in the presence of noise,'' \emph{IEEE Trans. Inf. Theory.},
  vol.~52, no.~1, pp. 6--18, Jan. 2006.

\bibitem{CandesRombergTao06}
E.~J. Cand\`{e}s, J.~Romberg, and T.~Tao, ``Stable signal recovery from
  incomplete and inaccurate measurements,'' \emph{Comm. on Pure and Applied
  Math.}, vol.~59, pp. 1207--1223, Feb. 2006.

\bibitem{Tropp_IT06}
J.~Tropp, ``Just relax: Convex programming methods for identifying sparse
  signals in noise,'' \emph{IEEE Trans. Inf. Theory}, vol.~52, no.~3, pp.
  1030--1051, Mar. 2006.

\bibitem{DonoForMost06}
D.~L. Donoho, ``For most large underdetermined systems of linear equations, the
  minimal $\ell_1$-norm solution is also the sparsest solution,'' \emph{Comm.
  on Pure and Applied Math.}, vol.~59, no.~6, pp. 797--829, Jun. 2006.

\bibitem{Haupt06}
J.~Haupt and R.~Nowak, ``Signal reconstruction from noisy random projections,''
  \emph{IEEE Trans. Inf. Theory}, vol.~59, no.~8, pp. 1207--1223, Aug. 2006.

\bibitem{Massart07}
P.~Massart, \emph{Concentration Inequalities and Model Selection}.\hskip 1em
  plus 0.5em minus 0.4em\relax Springer, 2007.

\bibitem{Reeves_thesis}
G.~Reeves, ``Sparsity pattern recovery in compressed sensing,'' Ph.D.
  dissertation, University of California, Berkeley, 2011.

\bibitem{VerSha99}
S.~Verd\'{u} and S.~Shamai, ``Spectral efficiency of {CDMA} with random
  spreading,'' \emph{IEEE Trans. Inf. Theory}, vol.~45, no.~2, pp. 622--640,
  Mar. 1999.

\bibitem{TseHan99}
D.~N.~C. Tse and S.~V. Hanly, ``Linear multiuser receivers: Effective
  interference, effective bandwidth and user capacity,'' \emph{IEEE Trans. Inf.
  Theory}, vol.~45, no.~2, pp. 641--657, Mar. 1999.

\bibitem{Rangan10b}
\BIBentryALTinterwordspacing
S.~Rangan, ``Generalized approximate message passing for estimation with random
  linear mixing,'' Oct. 2010. [Online]. Available:
  \url{http://arxiv.org/abs/1010.5141}
\BIBentrySTDinterwordspacing

\bibitem{BarSarBar10}
D.~Baron, S.~Sarvotham, and R.~G. Baraniuk, ``Bayesian compressive sensing via
  belief propagation,'' \emph{IEEE Trans. Signal Process.}, vol.~58, no.~1, pp.
  269--280, Jan. 2010.

\bibitem{Rangan10a}
\BIBentryALTinterwordspacing
S.~Rangan, ``Estimation with random linear mixing, belief propagation and
  compressed sensing,'' Jan. 2010. [Online]. Available:
  \url{http://arxiv.org/abs/1001.2228}
\BIBentrySTDinterwordspacing

\bibitem{RG09}
G.~Reeves and M.~Gastpar, ``Efficient sparsity pattern recovery,'' in
  \emph{Proc. 30th Symposium on Information Theory}, Eindhoven, May. 2009.

\bibitem{BayMon10b}
M.~Bayati and A.~Montanari, ``The lasso risk for gaussian matrices,'' Aug.
  2010, arXiv:1008.2581v1 [math.ST].

\bibitem{EdwAnd75}
S.~F. Edwards and P.~W. Anderson, ``Theory of spin glasses,'' \emph{J. Phys. F:
  Metal Physics}, vol.~5, pp. 965--974, 1975.

\bibitem{DT09}
D.~L. Donoho and J.~Tanner, ``Counting faces of randomly-projected polytopes
  when the projection radically lowers dimension,'' \emph{J. Amer. Match.
  Soc.}, vol.~22, no.~1, pp. 1--53, Jan. 2009.

\bibitem{WuVer10}
Y.~Wu and S.~Verd\'{u}, ``Renyi information dimension: Fundamental limits of
  almost lossless analog compression,'' \emph{IEEE Trans. Inf. Theory},
  vol.~56, no.~8, pp. 3721--3748, Aug. 2010.

\bibitem{LaurentMassart_stat98}
B.~{Laurent} and P.~{Massart}, ``{Adaptive estimation of a quadratic functional
  by model selection},'' \emph{Annals of Statistics}, vol. 28(5), pp.
  1302--1338, 2000.

\bibitem{ElementsofIT}
T.~M. Cover and J.~A. Thomas, \emph{Elements of Information Theory}.\hskip 1em
  plus 0.5em minus 0.4em\relax New York: Wiley, 1991.

\bibitem{PaiWei94}
C.~C. Paige and M.~Wei, ``{History and generality of the CS decomposition},''
  \emph{Linear Algebra Appl.}, vol. 208-209, pp. 303--326, 1994.

\bibitem{GraSch01}
A.~Grant and C.~Schlegel, ``Convergence of linear interference cancellation
  multiuser receivers,'' \emph{IEEE Trans. Inf. Theory}, vol.~49, no.~10, pp.
  1823--1834, Oct. 2001.

\bibitem{GSV05}
D.~Guo, S.~Shamai, and S.~Verd\'{u}, ``Mutual information and minimum
  mean-square error in {G}aussian channels,'' \emph{IEEE Trans. Inf. Theory},
  vol.~51, pp. 1261--1282, Apr. 2005.

\bibitem{tse_WirelessComm05}
D.~Tse and P.~Viswanath, \emph{Fundamentals of Wireless Communication}.\hskip
  1em plus 0.5em minus 0.4em\relax Cambridge University Press, 2005.

\bibitem{DonJoh94}
D.~L. Donoho and I.~M. Johnstone, ``Ideal spatial adaptation by wavelet
  shrinkage,'' \emph{Biometrika}, vol.~81, pp. 425--455, 1994.

\end{thebibliography}


\begin{IEEEbiographynophoto}{Galen Reeves} received the B.S. degree in electrical and computer engineering from Cornell University in 2005 and the M.S. and Ph.D. degrees in electrical engineering and computer sciences from the University of California at Berkeley in 2007 and 2011 respectively. He is currently a postdoctoral scholar at Stanford university. His his research interests include compressed sensing, statistical signal processing, information theory, and machine learning.
% In the Summer of 2011 he was a postdoctoral researcher in the School of Computer and Communication Sciences, EPFL.
\end{IEEEbiographynophoto}

%\begin{IEEEbiographynophoto}{Michael Gastpar}

\begin{IEEEbiographynophoto}{Michael Gastpar}
received the Dipl. El.-Ing. degree from ETH Z\"urich, in 1997, the M.S. degree from the University of Illinois at Urbana-Champaign, Urbana, IL, in 1999, and the
Doctorat \`es Science degree from Ecole Polytechnique F\'ed\'erale (EPFL), Lausanne, Switzerland, in 2002, all in electrical engineering. He was also a student in engineering and philosophy at the Universities of Edinburgh and Lausanne.

He is a Professor in the School of Computer and Communication
Sciences, Ecole Polytechnique F\'ed\'erale (EPFL), Lausanne, Switzerland.
He was an Assistant (2003-2008) and tenured Associate Professor (2008-2011)
with the Department of Electrical Engineering and Computer Sciences,
University of California, Berkeley, where he still holds a faculty position.
He also holds a faculty position at Delft University of Technology, The Netherlands,
and was a Researcher with the Mathematics of Communications Department,
Bell Labs, Lucent Technologies, Murray Hill, NJ.
His research interests are
in network information theory and related coding and signal processing techniques,
with applications to sensor networks and neuroscience.

Dr. Gastpar won the 2002 EPFL Best Thesis Award, an NSF CAREER Award
in 2004, an Okawa Foundation Research Grant in 2008, and an ERC Starting Grant in 2010.
He is an Information Theory Society Distinguished Lecturer (2009Ð2011). He was an Associate
Editor for Shannon Theory for the IEEE TRANSACTIONS ON INFORMATION
THEORY (2008--2011), and he has served as Technical Program Committee Co-Chair for the
2010 International Symposium on Information Theory, Austin, TX.
\end{IEEEbiographynophoto}

\end{document}